\documentclass{article}

\usepackage[top=1.36in, bottom=1.36in, left=1.35in, right=1.35in]{geometry}
\usepackage[english]{babel}
\usepackage{amsmath}
\usepackage{amssymb}
\usepackage{amsthm}
\usepackage{cite}    
\usepackage[usenames,dvipsnames,svgnames,table]{xcolor}
\usepackage{tikz}
\usepackage{xspace}
\usepackage{ifthen}
\usepackage{paralist}
\usepackage[colorlinks,final,citecolor=blue]{hyperref}
\usepackage{epstopdf}
\usepackage{comment}

\makeatletter
\newtheorem*{rep@theorem}{\rep@title}
\newcommand{\newreptheorem}[2]{%
\newenvironment{rep#1}[1]{%
 \def\rep@title{#2 \ref{##1}}%
 \begin{rep@theorem}}%
 {\end{rep@theorem}}}
\makeatother
\newreptheorem{theorem}{Theorem}
\newreptheorem{lemma}{Lemma}

\usetikzlibrary{%
  arrows,%
  positioning,%
  decorations.pathmorphing,%
  decorations.pathreplacing,%
}

\usepackage{listings}
\newcommand\ignore[1]{}
\lstnewenvironment{code}
    {\lstset{}%
      \csname lst@SetFirstLabel\endcsname}
    {\csname lst@SaveFirstLabel\endcsname}
\newcommand{\north}{{\tt N}}
\newcommand{\west}{{\tt  W}}
\newcommand{\south}{{\tt S}}
\newcommand{\east}{{\tt  E}}

\newcommand{\sat}{{\scshape Sat}\xspace}

\theoremstyle{plain}
\newtheorem{theorem}{Theorem}

\newtheorem{lemma}{Lemma}

\theoremstyle{definition}
\newtheorem{definition}[theorem]{Definition}
\newtheorem{example}[theorem]{Example}
\newtheorem{invariant}[theorem]{Invariant}

\newtheorem*{prob}{Problem}

\theoremstyle{remark}

\newcommand{\problem}[3]{%
	\begin{prob}[#3]\ \\ \vspace{-\baselineskip}
		\begin{compactdesc}
			\item{\sc Given:} #1;
			\item{\sc Output:} #2.
		\end{compactdesc}
	\end{prob}}

\newcommand{\cT}{\ensuremath{\mathcal{T}}\xspace}

\newcommand\bul{\ensuremath{{\vphantom u\raisebox{1pt}{$\scriptscriptstyle\circ$}}}\xspace}
\newcommand\wild{\ensuremath{{\#}}\xspace}

\newcommand\unc{\ensuremath{{\mathbf{\mathsf u}}}\xspace}

\newcommand{\Tile}[4]{
	\begin{scope}[xshift=#1,yshift=#2]
		\fill [#3] (-.45,-.45) rectangle (.45,.45);
		#4
		\draw (-.45,-.45) rectangle (.45,.45);
	\end{scope}
}

\newcommand{\Tug}[2]{
	\Tile{#1}{#2}{white}
	{\draw [Green,ultra thick] (0,-.45) -- (0,0);
	\draw [Red,dotted,ultra thick] (0,0) -- (0,.45);
	\node at (-.5,0) [anchor=west] {\unc};
	\node at (.5,0) [anchor=east] {\unc};}
}
\newcommand{\Tur}[2]{
	\Tile{#1}{#2}{lightgray}
	{\draw [Red,dotted,ultra thick] (0,-.45) -- (0,0);
	\draw [Green,ultra thick] (0,0) -- (0,.45);
	\node at (-.5,0) [anchor=west] {\unc};
	\node at (.5,0) [anchor=east] {\unc};}
}
\newcommand{\Tub}[2]{
	\Tile{#1}{#2}{lightgray}
	{\node at (-.5,0) [anchor=west] {\unc};
	\node at (.5,0) [anchor=east] {\unc};
	\node at (0,.5) [anchor=north] {\bul};
	\node at (0,-.5) [anchor=south] {\bul};}
}

\newcommand{\Tbg}[2]{
	\Tile{#1}{#2}{white}
	{\draw [Green,ultra thick] (0,-.45) -- (0,.45);
	\node at (-.5,0) [anchor=west] {\bul};
	\node at (.5,0) [anchor=east] {\bul};}
}
\newcommand{\Trg}[2]{
	\Tile{#1}{#2}{white}
	{\draw [Red,dotted,ultra thick] (-.45,0) -- (.45,0);
	\draw [Green,ultra thick] (0,-.45) -- (0,.45);}
}
\newcommand{\Tbr}[2]{
	\Tile{#1}{#2}{white}
	{\draw [Red,dotted,ultra thick] (0,-.45) .. controls (0,0) .. (.45,0);
	\node at (-.5,0) [anchor=west] {\bul};
	\node at (0,.5) [anchor=north] {\bul};}
}
\newcommand{\Trb}[2]{
	\Tile{#1}{#2}{white}
	{\draw [Red,dotted,ultra thick] (-.45,0) .. controls (0,0) .. (0,.45);
	\node at (0,-.5) [anchor=south] {\bul};
	\node at (.5,0) [anchor=east] {\bul};}
}
\newcommand{\Tbb}[2]{
	\Tile{#1}{#2}{white}{
	\node at (0,-.5) [anchor=south] {\bul};
	\node at (.5,0) [anchor=east] {\bul};
	\node at (-.5,0) [anchor=west] {\bul};
	\node at (0,.5) [anchor=north] {\bul};}
}

\newcommand{\Tsg}[2]{
	\Tile{#1}{#2}{white}
	{\draw [Blue,ultra thick,{-[}] (-.45,0) -- (.45,0);
	\draw [Green,ultra thick] (0,-.45) -- (0,.45);}
}
\newcommand{\Tsr}[2]{
	\Tile{#1}{#2}{white}
	{\draw [Blue,ultra thick] (-.45,0) -- (.45,0);
	\draw [Red,dotted,ultra thick] (0,-.45) -- (0,.45);}
}
\newcommand{\Tsb}[2]{
	\Tile{#1}{#2}{white}
	{\draw [Blue,ultra thick] (-.45,0) -- (.45,0);
	\node at (0,.5) [anchor=north] {\bul};
	\node at (0,-.5) [anchor=south] {\bul};}
}
\newcommand{\TSb}[2]{
	\Tile{#1}{#2}{white}
	{\draw [Blue,ultra thick,{]-}] (-.45,0) -- (.45,0);
	\node at (0,.5) [anchor=north] {\bul};
	\node at (0,-.5) [anchor=south] {\bul};}
}
\newcommand{\TSr}[2]{
	\Tile{#1}{#2}{white}
	{\draw [Red,dotted,ultra thick] (0,-.45) -- (0,0);
	\draw [Blue,ultra thick,{]-}] (-.45,0) -- (.45,0);
	\node at (0,.5) [anchor=north] {\bul};}
}

\newlength{\XCoord}
\newlength{\YCoord}

\newcommand{\InitCoordinates}{\setlength\XCoord{0cm}\setlength\YCoord{0cm}}
\newcommand{\LineUp}{\setlength\XCoord{0cm}\addtolength\YCoord{1cm}}

\newcommand{\MoveRight}{\addtolength\XCoord{1cm}}

\newcommand{\Pug}{\Tug{\XCoord}{\YCoord}\MoveRight}
\newcommand{\Pur}{\Tur{\XCoord}{\YCoord}\MoveRight}
\newcommand{\Pub}{\Tub{\XCoord}{\YCoord}\MoveRight}
\newcommand{\Pbg}{\Tbg{\XCoord}{\YCoord}\MoveRight}
\newcommand{\Prg}{\Trg{\XCoord}{\YCoord}\MoveRight}
\newcommand{\Pbr}{\Tbr{\XCoord}{\YCoord}\MoveRight}
\newcommand{\Prb}{\Trb{\XCoord}{\YCoord}\MoveRight}
\newcommand{\Pbb}{\Tbb{\XCoord}{\YCoord}\MoveRight}
\newcommand{\Psg}{\Tsg{\XCoord}{\YCoord}\MoveRight}
\newcommand{\Psr}{\Tsr{\XCoord}{\YCoord}\MoveRight}
\newcommand{\Psb}{\Tsb{\XCoord}{\YCoord}\MoveRight}
\newcommand{\PSb}{\TSb{\XCoord}{\YCoord}\MoveRight}
\newcommand{\PSr}{\TSr{\XCoord}{\YCoord}\MoveRight}

\newcommand{\SQb}{{\color{Black}\blacksquare}}
\newcommand{\SQw}{{\square}}

\newcommand{\Dr}{{{}_{\SQb}^{\SQw}}}
\newcommand{\Dg}{{{}_{\SQw}^{\SQb}}}
\newcommand{\Db}{{{}_{\SQb}^{\SQb}}}

\newcommand*{\RightArrow}[1][]{\mathbin{\tikz [baseline=-0.25ex,->,thick,#1] \draw (0,0.5ex) -- (.75em,0.5ex);}}
\newcommand*{\UpArrow}[1][]{\phantom{\SQw}\tikz [overlay,baseline=-0.25ex,->,thick,#1]{\draw (-4pt,0) -- (-4pt,.75em);}}

\newcommand{\SIs}{\RightArrow[Blue]}
\newcommand{\SIr}{\RightArrow[dotted,Red]}
\newcommand{\SIu}{\mathbin{\mathrm{\unc}}}
\newcommand{\SIb}{\RightArrow[Blue]}

\newcommand{\VGg}{\UpArrow[Green]}
\newcommand{\VGr}{\UpArrow[dotted,Red]}
\newcommand{\VGb}{\phantom{\SQw}\tikz [overlay,baseline=-0.25ex,->,thick]{\node at (-4pt,.375em) {\bul};}}

\newcommand{\sett}[2]{\left\{#1\mathrel{\left|\vphantom{#1}\vphantom{#2}\right.}#2\right\}}
\newcommand{\set}[1]{\left\{\mathinner{#1}\right\}}

\newcommand{\abs}[1]{\left|\mathinner{#1}\right|}

\newcommand{\N}{\mathbb{N}}

\newcommand{\Oh}{\mathcal{O}}

\newcommand{\NP}{\ensuremath{\mathbf{NP}}\xspace}

\newcommand{\pats}{{\sc Pats}\xspace}

\newcommand{\SAT}{{\sc Sat}\xspace}
\newcommand{\modSAT}{{\sc M-Sat}\xspace}

\newcommand\kpats{{\sc \mbox{$k$-{\pats}}}\xspace}

\newcommand{\dom}[1]{{\rm dom}(#1)}

\begin{document}

\title{Binary pattern tile set synthesis is NP-hard}

\author{
	Lila Kari\thanks{Department of Computer Science, University of Western Ontario, London ON N6A 1Z8, Canada. {\tt \{lila,steffen\}@csd.uwo.ca}.
These authors' research was supported by the NSERC Discovery Grant R2824A01 and UWO Faculty of Science grant to L.~K.}\and
	Steffen Kopecki\footnotemark[1]\and
	Pierre-\'Etienne Meunier\thanks{Aix Marseille Universit\'e, CNRS, LIF UMR 7279, 13288, Marseille, France, {\tt pierre-etienne.meunier@lif.univ-mrs.fr}. Supported in part by National Science Foundation Grant CCF-1219274.}\and
	Matthew J. Patitz\thanks{Department of Computer Science and Computer Engineering, University of Arkansas, Fayetteville, AR, USA. {\tt mpatitz@self-assembly.net}. This author's research was supported in part by National Science Foundation Grant CCF-1117672.}\and
	Shinnosuke Seki\thanks{Department of Information and Computer Science, Aalto University, P.~O.~Box 15400, FI-00076, Aalto, Finland. {\tt shinnosuke.seki@aalto.fi}. This author's research was supported in part by Academy of Finland, Postdoctoral Researcher Grant 13266670/T30606.}
}

\maketitle

\begin{abstract}

In the field of algorithmic self-assembly, a long-standing unproven conjecture
has been that of the \NP-hardness of binary pattern tile set synthesis
(2-\pats).  The {$k$-\pats} problem is that of designing a tile assembly
system with the smallest number of tile types which will self-assemble an input
pattern of $k$ colors.  Of both theoretical and practical significance,
{$k$-\pats} has
been studied in a series of papers which have shown {$k$-\pats} to be \NP-hard
for $k=60$, $k=29$, and then $k=11$.  In this paper, we close the fundamental
conjecture that $2$-\pats is \NP-hard, concluding this line of study.

While most of our proof relies on standard mathematical proof techniques, one crucial lemma
makes use of a computer-assisted proof, which is a relatively novel but increasingly utilized paradigm for
deriving proofs for complex mathematical problems.  This tool is especially powerful
for attacking combinatorial problems, as exemplified by the proof of the
four color theorem by Appel and Haken (simplified later by Robertson, Sanders,
Seymour, and Thomas) or the recent important advance on the Erd\H{o}s
discrepancy problem by Konev and Lisitsa using computer programs.
We utilize a massively parallel algorithm and thus turn an otherwise intractable portion of our proof into a program which requires approximately a year of computation time, bringing the use of computer-assisted proofs to a new scale.  We fully detail the algorithm employed by our code, and make the code freely available online.

\end{abstract}

	\section{Introduction}
Self-assembly is the process through which
disorganized, relatively simple components au\-ton\-o\-mous\-ly coalesce according to simple local
rules to form more complex target structures. Despite sounding simple, self-assembly can
produce extraordinary results. For example, and beyond the many examples occurring in nature, researchers have been able to
self-assemble a wide variety of nanoscale structures experimentally, such
as regular arrays~\cite{WinLiuWenSee98}, fractal
structures~\cite{RoPaWi04,FujHarParWinMur07}, smiley
faces~\cite{rothemund2006folding,wei2012complex}, DNA
tweezers~\cite{yurke2000dna}, logic
circuits~\cite{seelig2006enzyme,qian2011scaling}, neural
networks~\cite{qian2011neural}, and molecular
robots\cite{DNARobotNature2010}. These examples are fundamental because they
demonstrate that self-assembly can, in principle, be used to manufacture
specialized geometrical, mechanical and computational objects at the
nanoscale. Potential future applications of nanoscale self-assembly include the
production of new materials with specifically tailored properties (electronic, photonic, etc.)
and medical technologies which are capable of diagnosing and even treating diseases in vivo and at the cellular level.

Controlling nanoscale self-assembly for the purposes of manufacturing atomically
precise components will require a bottom-up, hands-off strategy. In other words,
the self-assembling units themselves will have to be ``programmed'' to direct
themselves to assemble efficiently and correctly.  Molecular self-assembly is
rapidly becoming a ubiquitous engineering paradigm, and robust
theory is necessary to inform us of its algorithmic capabilities and ultimate limitations.

In 1998, Erik Winfree \cite{Winfree_PhDthesis} introduced the abstract Tile Assembly Model
(aTAM), a simplified discrete mathematical model of algorithmic DNA nanoscale
self-assembly pioneered by Seeman \cite{Seem82}. The aTAM is essentially an asynchronous
nondeterministic cellular automaton that models crystal growth processes. Put
another way, the aTAM augments classical Wang tiling \cite{Wang61}
with a mechanism for sequential growth of a tiling. This contrasts with classical Wang
tiling in which only the existence of a valid mismatch-free tiling is
considered, and \emph{not} the order of tile placement. In the aTAM, the
fundamental components are translatable but un-rotatable square \emph{tiles}
whose sides are labeled with colored \emph{glues}, each with an integer
\emph{strength}.  Two tiles that are placed next to each other \emph{interact}
if the glue colors on their abutting sides match, and they \emph{bind} if the
strengths on their abutting sides match and sum to at least a certain (integer)
\emph{temperature}. Self-assembly starts from an initial \emph{seed assembly}
and proceeds nondeterministically and asynchronously as tiles bind to the
seed-containing-assembly. Despite its deliberate simplification, the aTAM is a
computationally expressive model in which simulations of arbitrary Turing
computations have been built \cite{Winfree_PhDthesis} and complex series of computations performed \cite{jSADS,jCCSA}. It has even been shown recently to be intrinsically universal
\cite{IUSA,USA,2HAMIU,Meunier-2014,woods2013intrinsic,Demaine-2012}.



The problem we study in this paper is the optimization of the design of tile assembly systems in the aTAM which self-assemble to form input colored patterns.  The input for this problem is a rectangular pattern consisting of $k$ colors, and the output is a tile set in the aTAM which self-assembles the pattern.  Essentially, each type of tile is assigned a ``color'', and the goal is to design a system consisting of the minimal number of tile types such that they deterministically self-assemble to form a rectangular assembly in which each tile is assigned the same color as the corresponding location in the pattern.  This problem was introduced in \cite{MaLombardi2008}, and has since then been extensively studied \cite{GoosOrponen10,LempCzeilOrp11,CzeizlerPopaPATS,Seki2013,JohnsenKaoSeki2013,JohnsenKaoSeki2014}.  The interest is both theoretical, to determine the computational complexity of designing efficient tile assembly systems, and practical, as the goal of self-assembling patterned substrates onto which a potentially wide variety of molecular components could be attached is a major experimental goal.  Known as $k$-\pats, where $k$ is the number of unique colors in the input pattern, previous work has steadily decreased the value of $k$ for which $k$-\pats has been shown to be \NP-hard, from $60$ \cite{Seki2013} to $29$ \cite{JohnsenKaoSeki2013} to $11$ \cite{JohnsenKaoSeki2014}. (Additionally, in a variant of $k$-\pats where the number of tile types of certain colors is restricted, is has been proven to be \NP-hard for $3$ colors \cite{kari2013}.) However, the foundational and previously unproven conjecture has been that for $k=2$, i.e. $2$-\pats, the problem is also \NP-hard.  This is our main result, which is thus the terminus of this line of research and a fundamental result in algorithmic self-assembly.




Our proof of the $2$-\pats conjecture requires the solution of a massive combinatorial problem, and thus one of the lemmas upon which it relies has a computer-assisted proof.  That is, it relies upon the output of a computer program which exhaustively searches through a combinatorially explosive set of all possible tile sets which can self-assemble a particular, carefully designed input pattern.  (In fact, during each run our, verification program generated over $66\cdot 10^{12}$ partial tile assemblies to be inspected.)  Although not yet commonplace in theoretical results, computer-assisted proofs are becoming more widespread, driven by the exponential growth of computing power which can be applied to previously and otherwise intractable problems.
Nonetheless, as early as 1976, Appel and Haken proved the four color theorem
\cite{AppelHaken1977-I,AppelHaken1977-II} by using a computer program to check whether
each of the thousands of possible candidates for the smallest-sized counter example
to this theorem were actually four-colorable or not (simplified in \cite{RoSaSeTh1996}).
Since then, important problems
in various fields have been solved (fully or partially)
with the assistance of computers: the discovery of Mersenne primes
\cite{Tuckerman1971},
the 17-point case of the happy ending problem
\cite{SzekeresPeters2006}, the \NP-hardness of minimum-weight triangulation
\cite{MulzerRote2008}, a special case of Erd\H{o}s' discrepancy conjecture
\cite{KonevLisitsa2014}, the ternary Goldbach conjecture
\cite{Helfgott2013}, and Kepler's conjecture \cite{Hales2000,Marchal2011}, among others.
However, to the best of our knowledge, the scale of our computation is much greater than all of those and others which have been published.  Our program required approximately one year of computation time on very modern, high-end machines (as a sum total over several hundred distributed cores) to complete and verify the correctness of the lemma.
Thus, in this paper, we take this approach to a new order of magnitude.  Such techniques
create new possibilities, spanning beyond the aesthetic concerns usually associated
with the domain (see quotations from Paul Erd\H{o}s in
\cite{Hoffman1998}). Indeed, our proof takes several days to complete on a
massively parallel computer, which makes it essentially impossible for a human to verify.
This kind of method is likely to become
common when it comes to proving correctness of biological or chemical systems, due to the
complexity of these objects. Moreover, the ``natural proofs'' line of research
\cite{razborov94,rudich97,Allender2010,Chow2011} also suggests that our ability
to produce and verify large proofs is likely to become fundamental in complexity
theory.

The computer program that we used in our proof was written in C++, and we have made the code freely available online\footnote{\url{http://self-assembly.net/wiki/index.php?title=2PATS-tileset-search}}.  Furthermore, we provide a full technical description and justification of the main algorithm utilized by the code in Section~\ref{programmatic}. \footnote{We have also implemented the same algorithm in a client-server system written in Haskell and included a rigorous proof of its correctness as the appendix, since as a functional programming language Haskell lends itself more readily to formal proofs of program correctness.  However, at the time of submission, due to decreased efficiencies as compared to C++, the Haskell program has not yet completed.}

\subsection{Main result}

Our result solves a long-standing open problem in the field of DNA self-assembly, the
so-called {\it binary pattern tile set synthesis} (2-{\pats}) problem
\cite{MaLombardi2008,Seki2013}.  In the general $k$-{\pats} for $k \ge 2$, given a placement of $k$ different kinds of nanoparticles, represented in the model as a $k$-colored rectangular pattern, we are asked to design a tile assembly system with colored tile types that self-assembles the pattern.
The periodic placement of Au (gold) nanoparticles on 2D DNA nanogrid \cite{ZhLiKeYa2006} can be considered a 2-colored (i.e., binary) rectangular pattern on which the two colors specify the presence/absence of an Au nanoparticle at the position.
Another example, shown in Figure~\ref{fig:SA_counter}, is a binary counter pattern that is self-assembled using 4 tile types with two colors.
2-\pats has been conjectured to be \NP-hard\footnote{This problem was claimed to be so in the succeeding paper by the authors of \cite{MaLombardi2008} but what they proved was the {\bf NP}-hardness of a related but different problem.}.
In \cite{Seki2013}, Seki proved for the first time that the \NP-hardness of 60-{\pats}, whose inputs are allowed to have 60 colors, and the result has been strengthened to that of 29-{\pats} \cite{JohnsenKaoSeki2013}, and further to that of 11-{\pats} \cite{JohnsenKaoSeki2014} recently.

Our main theorem closes this line of research by lowering the number of colors allowed for input patterns to only two.
Although a number of terms it contains have not yet been defined, we state it now:

\begin{theorem}
\label{thm:2pats_NPhard}
The $2$-\pats optimization problem of finding, given a $2$ colored rectangular
pattern $P$, the minimal colored tileset that produces a single terminal
assembly where the color arrangement is exactly the same as in $P$, is NP-hard.
\end{theorem}

\subsection{Overview of our proof technique}

The main idea of our proof is similar to the strategies adopted by
\cite{Seki2013,JohnsenKaoSeki2013,JohnsenKaoSeki2014}.  We embed the computation
of a verifier of solutions for an NP-complete problem (in our case, a variant of
\sat, which we call \modSAT) in an assembly, which is relatively straightforward in Winfree's aTAM. One can indeed engineer a tile assembly system (TAS) in this model, with
colored tiles, implementing a verifier of solutions of the variant of \sat, in which a
formula $F$ and a variable assignment $\phi \in \{0,1\}^n$ are encoded in the seed
assembly, and a tile of a special color appears after some time if and only if
$F(\phi)=1$.
In our actual proof, reported in Section~\ref{sec:2pats}, we design a set $T$ of 13 tile types and a reduction of a given instance $\phi$ of \modSAT to a rectangular pattern $P_F$ such that

\begin{enumerate}[Property 1.]
\setlength{\parskip}{1mm}
\setlength{\itemsep}{0mm}
\item	A TAS using tile types in $T$ self-assembles $P_F$ if and only if $F$ is satisfiable.
\item	Any TAS of at most 13 tile types that self-assembles $P_F$ is isomorphic to $T$.\label{toughTask}
\end{enumerate}

Therefore, $F$ is solvable if and only if $P_F$ can be self-assembled
using at most 13 tile types.  In previous works
\cite{Seki2013,JohnsenKaoSeki2013,JohnsenKaoSeki2014}, significant portions of the
proofs were dedicated to ensuring their analog of Property~\ref{toughTask}, and many colors were
wasted to make the property ``manually'' checkable (for reference, 33 out of 60
colors just served this purpose for the proof of \NP-hardness of 60-{\pats} \cite{Seki2013} and 2 out of 11 did that for 11-{\pats}
\cite{JohnsenKaoSeki2014}).  Cutting this ``waste'' causes a combinatorial explosion of cases to test and motivates us to use
a computer program to do the verification instead.


Apart from the verification of Property~\ref{toughTask} (in Lemma~\ref{prop:gadget}), the rest of our proof can be verified as done in traditional mathematical proofs; our proof is in Section~\ref{sec:2pats}.
The verification of Property~\ref{toughTask} is done by an algorithm described in Section~\ref{subsect:algorithm}, which, given a pattern and an integer $n$, searches exhaustively for all possible sets of $n$ tile types that self-assemble the pattern.  The correctness of the algorithm is proven, and the C++ code implementing the algorithm is made freely available for inspection.




	\section{Preliminaries}
	\label{sec:preliminaries}

This section is divided into three subsections:
necessary notions and notation on rectilinear tile
assembly system, definitions related to the \pats problem, and then a
preliminary lemma about the variant of \SAT that we use.

\subsection{Pattern assembly}

Let $\mathbb{N}$ be the set of nonnegative integers, and for $n \in \mathbb{N}$, let $[n] = \{0, 1, 2, \ldots, n{-}2, n{-}1\}$.
For $k \ge 1$, a {\it $k$-colored pattern} is a partial function from $\mathbb{N}^2$ to the set of (color) indices $[k]$, and a {\it $k$-colored rectangular pattern} (of width $w$ and height $h$) is a pattern whose domain is $[w] \times [h]$.

\begin{figure}[tb]
\begin{center}
\begin{minipage}{0.4\linewidth}
\includegraphics[scale=0.45]{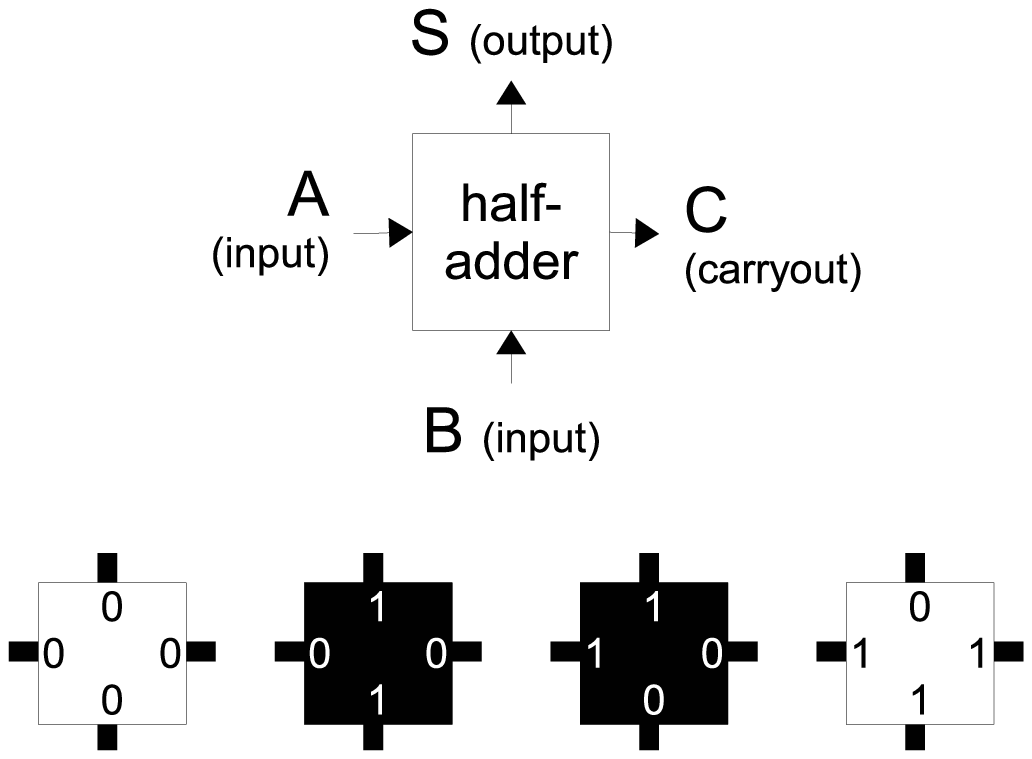}
\end{minipage}
\begin{minipage}{0.05\linewidth}
\hspace*{5mm}
\end{minipage}
\begin{minipage}{0.5\linewidth}
\includegraphics[scale=0.3]{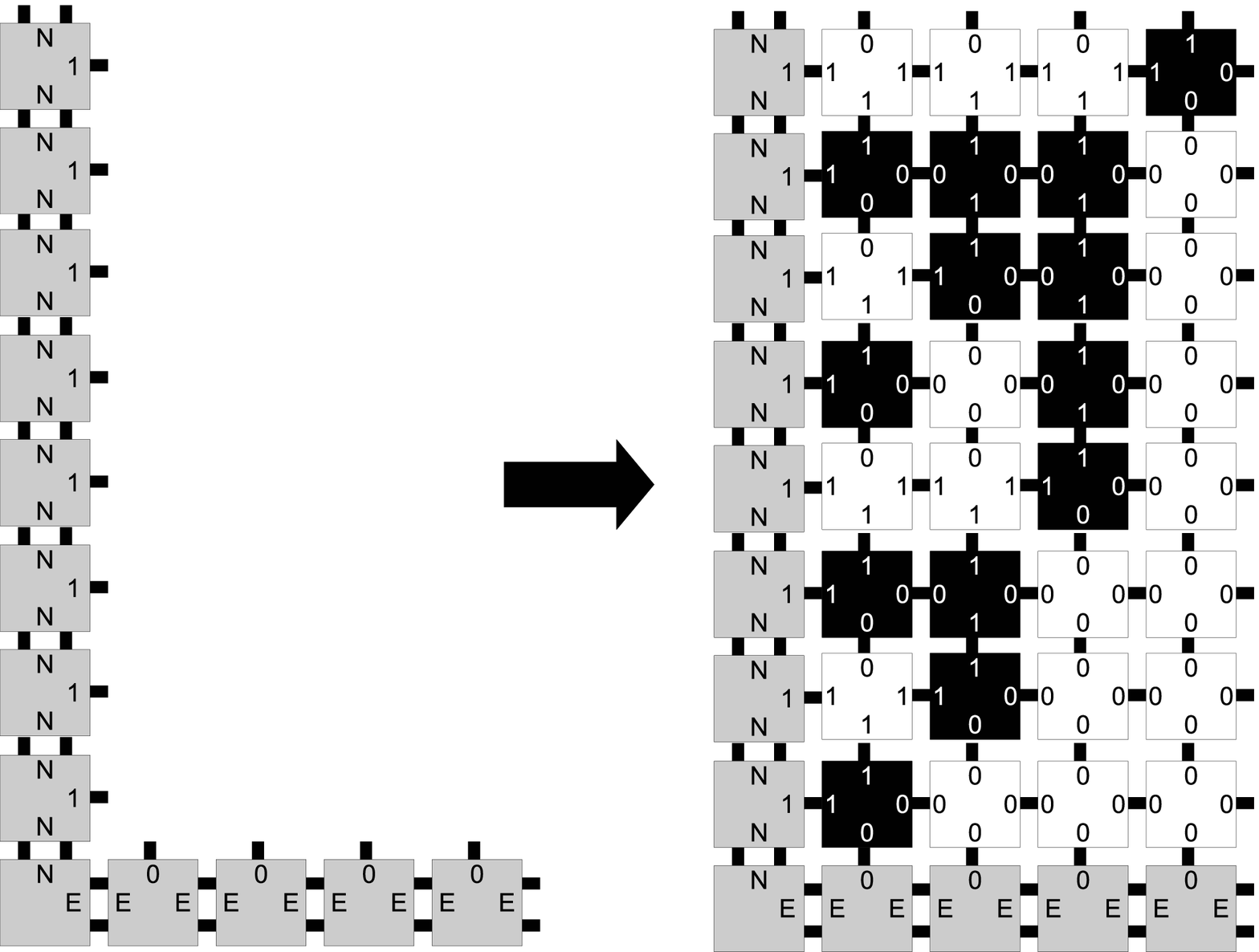}
\end{minipage}
\end{center}
\caption{
	(Left) Four tile types implement together the half-adder with two inputs A, B from the west and south, the output S to the north, and the carryout C to the east.
	(Right) Copies of the ``half-adder'' tile types turn the L-shape seed into the binary counter pattern.
}
\label{fig:SA_counter}
\end{figure}

Let $\Sigma$ be a glue alphabet.
A {\it (colored) tile type} $t$ is a tuple $(g_\north, g_\west, g_\south, g_\east, c)$, where $g_\north, g_\west, g_\south, g_\east \in \Sigma$ represent the respective north, west, south, and east glue of $t$, and $c \in \mathbb{N}$ is a color (index) of $t$.
For instance, the right black tile type in Figure~\ref{fig:SA_counter} (Left) is (1, 1, 0, 0, black).
We refer to $g_\north, g_\west, g_\south, g_\east$ as $t(\north), t(\west), t(\south), t(\east)$, respectively, and by $c(t)$ we denote the color of $t$.
For a set $T$ of tile types, an {\it assembly} $\alpha$ over $T$ is a partial function from $\mathbb{N}^2$ to $T$.
(When algorithms and computer programs will be explained in Section~\ref{programmatic}, it is convenient for the tile types in $T$ to be indexed as $t_0, t_1, \ldots, t_{\ell{-}1}$ and consider the assembly rather as a partial function to $[\ell]$.)
Its pattern, denoted by $P(\alpha)$, is such that ${\rm dom}(P(\alpha)) = \dom\alpha$ and $P(\alpha)(x, y) = c(\alpha(x, y))$ for any $(x, y) \in \dom\alpha$.
Given another assembly $\beta$, we say $\alpha$ is a {\it subassembly} of $\beta$ if $\dom\alpha \subseteq \dom\beta$ and, for any $(x, y) \in \dom\alpha$, $\beta(x, y) = \alpha(x, y)$.

A {\it rectilinear tile assembly system} (RTAS) is a pair $\mathcal{T} = (T, \sigma_L)$ of a set $T$ of tile types and an L-shape seed $\sigma_L$, which is an assembly over another set of tile types disjoint from $T$ such that $\dom\sigma_L = \{(-1, -1)\} \cup ([w] \times \{-1\}) \cup (\{-1\} \times [h])$ for some $w, h \in \mathbb{N}$.
Its {\it size} is measured by the number of tile types employed, that is, $|T|$.
According to the following general rule all RTASs obey, it tiles the first quadrant delimited by the seed:

\vspace*{3mm}
\noindent
{\bf RTAS's tiling rule:}
To a given assembly $\alpha$, a tile $t \in T$ can attach at position $(x, y)$ if
\begin{enumerate}
\setlength{\itemsep}{0mm}
\item	$\alpha(x, y)$ is undefined,
\item	both $\alpha(x{-}1, y)$ and $\alpha(x, y{-}1)$ are defined,
\item	$t(\west) = \alpha(x{-}1, y)[\east]$ and $t(\south) = \alpha(x, y{-}1)[\north]$.
\end{enumerate}
The attachment results in a larger assembly $\beta$ whose domain is $\dom\alpha \cup \{(x, y)\}$ such that for any $(x', y') \in \dom\alpha$, $\beta(x', y') = \alpha(x, y)$, and $\beta(x, y) = t$.
When this attachment takes place in the RTAS $\mathcal{T}$, we write $\alpha \to_1^{\mathcal{T}} \beta$.
Informally speaking, the tile $t$ can attach to the assembly $\alpha$ at $(x, y)$ if on $\alpha$, both $(x{-}1, y)$ and $(x, y{-}1)$ are tiled while $(x, y)$ is not yet, and the west and south glues of $t$ match the east glue of the tile at $(x{-}1, y)$ and the north glue of the tile at $(x, y{-}1)$, respectively.
This implies that, at the outset, (0,~0) is the sole position where a tile may attach.

For those who are familiar with the Winfree's aTAM \cite{Winfree_PhDthesis}, it should be straightforward that an RTAS is a temperature-2 tile assembly system all of whose glues are of strength 1.

\begin{example}\label{ex:SA_counter}
	See Figure~\ref{fig:SA_counter} for an RTAS with 4 tile types, aiming at self-assembling the binary counter pattern.
	To its L-shape seed shown there, a black tile of type (1,~1,~0,~0,~black) can attach at (0, 0), while no tile of other types can due to glue mismatches.
	The attachment makes the two positions (0, 1) and (1, 0) attachable.
	Tiling in RTASs thus proceeds from south-west to north-east {\it rectilinearly} until no attachable position is left.
\end{example}

The set $\mathcal{A}[\mathcal{T}]$ of {\it producible} assemblies by $\mathcal{T}$ is defined recursively as follows: (1) $\sigma_L \in \mathcal{A}[\mathcal{T}]$, and (2) for $\alpha \in \mathcal{A}[\mathcal{T}]$, if $\alpha \to_1^{\mathcal{T}} \beta$, then $\beta \in \mathcal{A}[\mathcal{T}]$.
A producible assembly $\alpha \in \mathcal{A}[\mathcal{T}]$ is called {\it terminal} if there is no assembly $\beta$ such that $\alpha \to_1^{\mathcal{T}} \beta$.
The set of terminal assemblies is denoted by $\mathcal{A}_\Box[\mathcal{T}]$.
Note that the domain of any producible assembly is a subset of $(\{-1\} \cup [w]) \times (\{-1\} \cup [h])$, starting from the seed $\sigma_L$ whose domain is $\{(-1, -1)\} \cup ([w] \times \{-1\}) \cup (\{-1\} \times [h])$.

A tile set $T$ is {\em directed} if for any distinct tile types $t_1, t_2 \in T$, $t_1(\west) \neq t_2(\west)$ or $t_1(\south) \neq t_2(\south)$ holds.
An RTAS $\mathcal{T} = (T, \sigma_L)$ is {\it directed} if its tile set $T$ is directed (the directedness of RTAS was originally defined in a different but equivalent way).
It is clear from the RTAS tiling rule that if $\mathcal{T}$ is directed, then it has the exactly one terminal assembly, which we let $\gamma$.
Let $\gamma$ be the subassembly of the terminal assembly such that $\dom\gamma \subseteq \mathbb{N}^2$, that is, the tiles on $\gamma$ did not originate from the seed $\sigma_L$ but were tiled by the RTAS.
Then we say that {\it $\mathcal{T}$ uniquely self-assembles the pattern $P(\gamma)$}.

\subsection{The \pats problem}

The {\em pattern self-assembly tile set synthesis} ({\pats}), proposed by Ma and
Lombardi \cite{MaLombardi2008}, aims at computing the minimum size directed RTAS
that uniquely self-assembles a given rectangular pattern. The solution to
{\pats} is required to be directed here, but not originally. 
However, in \cite{GoosOrponenDNA16}, it was proved that among all the RTASs that uniquely self-assemble the pattern, the minimum one is directed. 

To study the algorithmic complexity of this problem on ``real size'' particle
placement problems, a first
restriction that can be placed is on the number of colors allowed for the input
patterns, thereby defining the $k$-{\pats} problem:

\problem{A $k$-colored pattern $P$}
	{a smallest directed RTAS that uniquely self-assembles $P$}
	{{\sc $k$-colored Pats} (\kpats)}

The {\bf NP}-hardness of {\kpats} follows from the \NP-hardness of its decision variant.

\problem{A $k$-colored pattern $P$ and an integer $m$}
	{``yes'' if $P$ can be uniquely self-assembled by an RTAS whose tileset contains at most $m$ tile types}
	{decision variant of \kpats}

In the rest of this paper, we use the terminology $k$-{\pats} to refer to the decision problem, unless otherwise noted.

	\subsection{Monotone Satisfiability Problem}

We now prove the \NP-completeness of a modified version of \SAT, which we will
use in our proof.  In the {\em monotone satisfiability with few true variables}
(\modSAT), we consider a number $k$ and a boolean formula $F$ in conjunctive
normal form \emph{without negations} and ask whether or not $F$ can be satisfied
by only allowing $k$ variables to be true.

Note that formulae with no negations are easily satisfied by setting all the
variables to $1$.

\begin{lemma}
\modSAT is \NP-complete.
\begin{proof}

First, for any formula $F$ and assignment $x$ of its variables, it can be
checked in polynomial time that $F(x)=1$. Therefore, \modSAT is in \NP.
\NP-hardness follows from an easy reduction from {\sc Vertex Cover}: let
$G=(V,E)$ be a graph, and for any $e=(v,v')\in E$, let $x_v$ and $x_{v'}$ be two
variables, and $C_e$ be the clause defined by $C_e=x_v\vee x_{v'}$.
Finding an assignment with a minimal number of true variables in
$\bigwedge_{e\in E} C_e$
is then equivalent to finding a minimal vertex cover of $G$.
\end{proof}
\end{lemma}

\section{\texorpdfstring{{2-\pats} is \NP-hard}{2-PATS is NP-hard}}
\label{sec:2pats}

We will prove that \pats is \NP-hard for binary patterns ($2$-colored patterns).
First, we present a binary gadget pattern $G$ such that among all tilesets of
size at most $13$, exactly one self-assembles $G$. Let $T$ be this tileset.

Then, we present a reduction from \modSAT to the problem of deciding whether an
input pattern can be assembled with $T$.
The \NP-hardness of 2-{\pats} follows, since we ensure that the input pattern $P$ contains the gadget pattern $G$ as a subpattern, hence forcing the use of $T$ to assemble $P$.

The reduction is proved in Theorem \ref{thm:2pats_NPhard}, and the claim that it
is the only tileset of size at most 13 that can assemble $G$ is proved in Lemma
\ref{prop:gadget}, using a programmatic proof. $T$ is shown on
Figure~\ref{fig:tiles}; it contains eleven white tile types and two black tile
types.

\begin{figure}[ht]
\centering
\begin{tikzpicture}
	\Tub{0}{0}
	\Tug{1.25cm}{0cm}
	\Tur{2.5cm}{0cm}

	\Tsb{0cm}{1.25cm}	
	\Tsg{1.25cm}{1.25cm}	
	\Tsr{2.5cm}{1.25cm}
	\TSb{3.75cm}{1.25cm}
	\TSr{5cm}{1.25cm}

	\Tbb{0}{2.5cm}
	\Tbg{1.25cm}{2.5cm}
	\Tbr{2.5cm}{2.5cm}	
	\Trg{3.75cm}{2.5cm}	
	\Trb{5cm}{2.5cm}
\end{tikzpicture}
\caption{The tile set $T$, where the background depicts the color of each tile type and the labels and signals depict the glues.
We refer to the tile types with a gray background as the black tile types.}
\label{fig:tiles}
\end{figure}
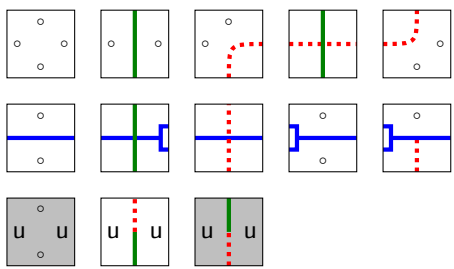

We interpret the glues in the tile set $T$ as follows.
Ten of the white tile types (first and second rows in Figure~\ref{fig:tiles}) simulate three types of signals which interact with each other:
\begin{compactenum}
	\item blue signals running from left to right,
	\item green signals running from bottom to top, and
	\item red signals running diagonally from bottom left to top right in a wavelike line.
\end{compactenum}
For better visibility in printouts the red signals are dotted; blue and green signals can easily be distinguished as blue signals run only horizontally while green signals run only vertically.
When any two of the signals meet, they simply cross over each other, where the red signal is displaced upwards or rightwards when crossing a blue or green signal, respectively.
Only when the blue signal crosses a green signal immediately before crossing a red signal, the red signal is destroyed.
In order to recognize this configuration, the blue signal is {\em tagged} when it crosses a green signal; in Figure~\ref{fig:tiles}, the tagging is displayed by the fork in the blue signal.
Let us stress that the signals are encoded in the glues of the tiles, they are not visible.
The other three tile types are used to uncover the green and red signals in two consecutive rows: two black tiles above each other stand for no signal, a white tile below a black tile stands for a green signal, a black tile below a white tile stands for a red signal, and blue signals are never uncovered.
Note that green (respectively, red) signals switch to red (respectively, green) in the first uncover row, but they switch back to their original state in the second uncover row.
An example assembly is shown in Figure~\ref{fig:example}.

\begin{figure}[ht]
\centering
\begin{tikzpicture}
	\InitCoordinates
	\Pub\Pug\Pur\Pub\Pug\Pub\Pur\Pug\Pub\Pub
	\LineUp
	\Pub\Pur\Pug\Pub\Pur\Pub\Pug\Pur\Pub\Pub
	\LineUp
	\Pbb\Pbg\Pbr\Prb\Pbg\Pbb\Pbr\Prg\Prb\Pbb
	\LineUp
	\Psb\Psg\PSb\Psr\Psg\PSb\Psb\Psg\PSr\Psb
	\LineUp
	\Pbb\Pbg\Pbb\Pbr\Prg\Prb\Pbb\Pbg\Pbb\Pbb
	\LineUp
	\Pbb\Pbg\Pbb\Pbb\Pbg\Pbr\Prb\Pbg\Pbb\Pbb
	\LineUp
	\Pub\Pug\Pub\Pub\Pug\Pub\Pur\Pug\Pub\Pub
	\LineUp
	\Pub\Pur\Pub\Pub\Pur\Pub\Pug\Pur\Pub\Pub
\end{tikzpicture}
\caption{Example interactions of the signals in the tile set $T$ with uncovering of the configurations.}
\label{fig:example}
\end{figure}

Clearly, the green signals always progress upwards from one glue to the next and blue signals progress rightwards from one glue to the next.
Let us formalize the progression of the red signal under the assumption that it is not destroyed.

\begin{lemma}\label{lem:red:signal}
Let the south glue of a tile in position $(x,y)$ be a red signal.
If this red signal progresses up-/rightwards to the south glue of a tile in position $(x',y')$ while crossing $i$ green signals, and $j$ blue signals, and no uncovering rows, then
\[
	x'-x-i = y'-y-j.
\]
\end{lemma}

\begin{proof}
In every row where the red signal crosses a blue signal, the red signal remains at its horizontal position.
Thus, only in the $y'-y-j$ rows without blue signal the red signal moves rightwards.
In each of these rows, we move one positions rightwards plus one position for every green signal that is crossed on the total way.
We conclude that $x' = x + i + (y'-y-j)$.
\end{proof}

By using a computer-aided search through all tile assignments with 13 tile types of the pattern shown in Figure~\ref{fig:gadget}, we obtain the following lemma:

\begin{lemma}\label{prop:gadget}
The gadget pattern $G$, shown in Figure~\ref{fig:gadget}, can
only be self-assembled with 13 tile types if a tile set is used which is
isomorphic to $T$.  No smaller tile set can self-assemble $G$.
\begin{proof}
Details about the program that we used for the computer-aided search are
presented in Section~\ref{programmatic}.
\end{proof}
\end{lemma}

\begin{figure}[ht]
\begin{align*}
\SIu\ &\SQb\SQb\SQw\SQb\SQw\SQb\SQb\SQb\SQb\SQb\SQb\SQb\SQw\SQb\SQb\SQb\SQb\SQb\SQb\SQb\SQb\\[-.7em]
\SIu\ &\SQw\SQb\SQb\SQb\SQb\SQw\SQb\SQw\SQb\SQb\SQw\SQb\SQb\SQw\SQb\SQb\SQb\SQb\SQw\SQb\SQb\\[-.7em]
\SIb\, &\SQw\SQw\SQw\SQw\SQw\SQw\SQw\SQw\SQw\SQw\SQw\SQw\SQw\SQw\SQw\SQw\SQw\SQw\SQw\SQw\SQw\\[-.7em]
\SIu\  &\SQb\SQb\SQw\SQb\SQw\SQb\SQb\SQb\SQb\SQb\SQb\SQb\SQw\SQb\SQb\SQb\SQb\SQb\SQb\SQw\SQb\\[-.7em]
\SIu\ &\SQw\SQb\SQb\SQb\SQb\SQw\SQb\SQw\SQb\SQb\SQw\SQb\SQb\SQw\SQb\SQb\SQb\SQb\SQw\SQb\SQb\\[-.7em]
\bul\ &\SQw\SQw\SQw\SQw\SQw\SQw\SQw\SQw\SQw\SQw\SQw\SQw\SQw\SQw\SQw\SQw\SQw\SQw\SQw\SQw\SQw\\[-.7em]
\SIu\ &\SQb\SQw\SQb\SQw\SQb\SQb\SQb\SQb\SQb\SQb\SQb\SQw\SQb\SQb\SQb\SQb\SQb\SQw\SQb\SQb\SQb\\[-.7em]
\SIu\ &\SQw\SQb\SQb\SQb\SQb\SQw\SQb\SQw\SQb\SQb\SQw\SQb\SQb\SQw\SQb\SQb\SQb\SQb\SQw\SQb\SQb\\[-.7em]
\SIr\, &\SQw\SQw\SQw\SQw\SQw\SQw\SQw\SQw\SQw\SQw\SQw\SQw\SQw\SQw\SQw\SQw\SQw\SQw\SQw\SQw\SQw\\[-.7em]
\SIu\ &\SQb\SQb\SQw\SQb\SQb\SQb\SQb\SQb\SQb\SQw\SQb\SQb\SQb\SQb\SQb\SQb\SQw\SQb\SQb\SQb\SQb\\[-.7em]
\SIu\ &\SQw\SQb\SQb\SQb\SQb\SQw\SQb\SQw\SQb\SQb\SQw\SQb\SQb\SQw\SQb\SQb\SQb\SQb\SQw\SQb\SQb\\[-.7em]
\SIu\ &\SQb\SQb\SQw\SQb\SQb\SQb\SQb\SQb\SQb\SQw\SQb\SQb\SQb\SQb\SQb\SQb\SQw\SQb\SQb\SQb\SQb\\[-.7em]
\SIu\ &\SQw\SQb\SQb\SQb\SQb\SQw\SQb\SQw\SQb\SQb\SQw\SQb\SQb\SQw\SQb\SQb\SQb\SQb\SQw\SQb\SQb\\[-.7em]
\SIb\, &\SQw\SQw\SQw\SQw\SQw\SQw\SQw\SQw\SQw\SQw\SQw\SQw\SQw\SQw\SQw\SQw\SQw\SQw\SQw\SQw\SQw\\[-.7em]
\SIu\ &\SQb\SQb\SQw\SQb\SQb\SQb\SQb\SQb\SQb\SQw\SQb\SQb\SQb\SQb\SQb\SQb\SQw\SQb\SQb\SQb\SQb\\[-.7em]
\SIu\ &\SQw\SQb\SQb\SQb\SQb\SQw\SQb\SQw\SQb\SQb\SQw\SQb\SQb\SQw\SQb\SQb\SQb\SQb\SQw\SQb\SQb\\[-.7em]
\SIu\ &\SQb\SQb\SQw\SQb\SQb\SQb\SQb\SQb\SQb\SQw\SQb\SQb\SQb\SQb\SQb\SQb\SQw\SQb\SQb\SQb\SQb\\[-.7em]
\SIu\ &\SQw\SQb\SQb\SQb\SQb\SQw\SQb\SQw\SQb\SQb\SQw\SQb\SQb\SQw\SQb\SQb\SQb\SQb\SQw\SQb\SQb\\[-.7em]
\SIb\, &\SQw\SQw\SQw\SQw\SQw\SQw\SQw\SQw\SQw\SQw\SQw\SQw\SQw\SQw\SQw\SQw\SQw\SQw\SQw\SQw\SQw\\[-.7em]
\SIu\ &\SQb\SQb\SQw\SQb\SQb\SQb\SQw\SQb\SQb\SQw\SQb\SQb\SQb\SQb\SQb\SQb\SQw\SQb\SQb\SQb\SQb\\[-.7em]
\SIu\ &\SQw\SQb\SQb\SQb\SQb\SQw\SQb\SQw\SQb\SQb\SQw\SQb\SQb\SQw\SQb\SQb\SQb\SQb\SQw\SQb\SQb\\[-.7em]
\bul\ &\SQw\SQw\SQw\SQw\SQw\SQw\SQw\SQw\SQw\SQw\SQw\SQw\SQw\SQw\SQw\SQw\SQw\SQw\SQw\SQw\SQw\\[-.7em]
\SIu\ &\SQb\SQw\SQb\SQb\SQw\SQb\SQb\SQb\SQw\SQb\SQb\SQb\SQb\SQb\SQb\SQw\SQb\SQb\SQb\SQb\SQb\\[-.7em]
\SIu\  &\SQw\SQb\SQb\SQb\SQb\SQw\SQb\SQw\SQb\SQb\SQw\SQb\SQb\SQw\SQb\SQb\SQb\SQb\SQw\SQb\SQb\\[-.45em]
&\VGg\VGr\VGb\VGb\VGr\VGg\VGb\VGg\VGr\VGb\VGg\VGb\VGb\VGg\VGb\VGr\VGb\VGb\VGg\VGb\VGb
\end{align*}
\caption{The binary gadget pattern $G$: the pattern can only be self-assembled by 13 tile types if we use the tile set $T$ and the input glues as shown in the pattern.
The bottom row in the pattern was not actually included in the search.}
\label{fig:gadget}
\end{figure}

In order to prove the \NP-hardness of $2$-\pats we encode an instance of \modSAT into the problem of deciding whether or not an initial configuration of red and green signals can be transformed into a target configuration of red and green signals after a certain number of steps.
The clue is that the hidden positions of blue signals refer to variables which are true in a satisfying variable assignment of the \modSAT instance.

\begin{reptheorem}{thm:2pats_NPhard}
	$2$-\pats is \NP-hard.
\end{reptheorem}

\begin{proof}
Let $k\in \N$ and $F$ be a set of $m$ clauses which is an instance of \modSAT.
For convenience, we assume that $F$ is defined over the $n$ variables $V = \set{1,2,\ldots, n}$.
We design a pattern $P$ based on $k$ and $F$ such that $P$ can be self-assembled with no more than $13$ tile types if and only if $F$ is satisfiable with only $k$ positive variables.
The pattern $P$, schematically presented in Figure~\ref{fig:pattern:scheme}, contains two rows defining the initial configuration $c_0$ and two rows containing the target configuration $c_t$.
These two rows are separated by $k + n$ completely white rows.
The gadget pattern $G$ is attached in the top left corner of the pattern and is separated by 11 white rows from the target configuration.
The area right of the gadget pattern $G$ is composed of the progression of the red and green signals which appear in the target configuration and where we consider the 11 white rows to have no blue signal and the signals in the following 24 rows are as required by the gadget pattern.

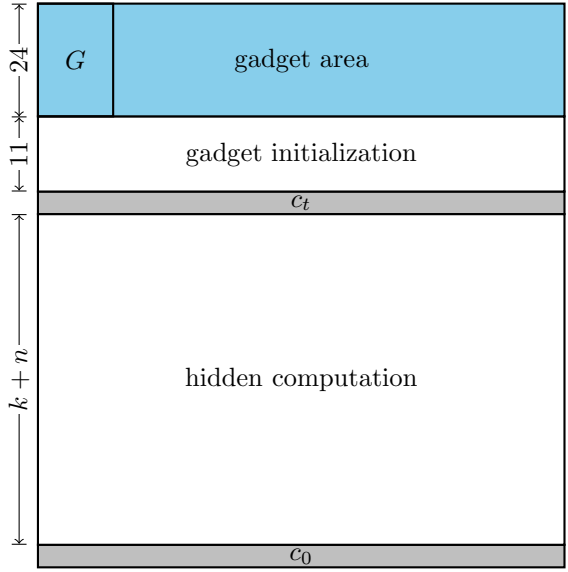
\begin{figure}[ht]
\centering
\begin{tikzpicture}

\draw [thick,fill=lightgray] (0,1) rectangle (7,1.3);
\draw [thick,fill=white] (0,1.3) rectangle (7,5.7);
\draw [thick,fill=lightgray] (0,5.7) rectangle (7,6);
\draw [thick,fill=white] (0,6) rectangle (7,7);
\draw [thick,fill=SkyBlue] (0,7) rectangle (7,8.5);
\draw [thick] (0,7) rectangle (1,8.5);

\node at (3.5,1.15) {$c_0$};
\node at (3.5,5.85) {$c_t$};
\node at (3.5,7.75) {gadget area};
\node at (3.5,6.5) {gadget initialization};
\node at (3.5,3.5) {hidden computation};
\node at (.5,7.75) {$G$};

\draw[|<->|] (-.25,1.3) --
	node [rotate=90,inner sep=1pt,fill=white] {$k+n$} (-.25,5.7);
\draw[|<->|] (-.25,6) --
	node [rotate=90,inner sep=1pt,fill=white] {$11$} (-.25,7);
\draw[<->|] (-.25,7) --
	node [rotate=90,inner sep=1pt,fill=white] {$24$} (-.25,8.5);

\end{tikzpicture}
\caption{The pattern $P$, consisting of $k+n+39$ rows.}
\label{fig:pattern:scheme}
\end{figure}

Since the target and initial configurations are represented by two rows of black or white pixels and cannot contain a pair of white pixels above each other, they can be interpreted as words over the three-letter alphabet $\set{\Db,\Dg,\Dr}$, where $\Db$ stands for no signal, $\Dg$ encodes a green signal, and $\Dr$ encodes a red signal.
The target and initial configuration rows are
\begin{align*}
	c_t &= w_{G} (\Db^{k+n}\Dr) (\Db\Db w_1) (\Db\Db w_2) \cdots (\Db\Db w_m) \\
	c_0 &= w_{G} (\Db^k\Dr\Db^n) (\Db\Dr w_1) (\Db\Dr w_2) \cdots (\Db\Dr w_m)
\end{align*}
where
\[
	w_{G} = \Dg\Db\Db\Db\Db\Dg\Db\Dg\Db\Db\Dg\Db\Db\Dg\Db\Db\Db\Db\Dg\Db\Db
\]
encodes the green signals which appear in the gadget pattern;
and for the $i$-th clause $C_i$ in $F$, the sequence $w_i$ contains $n$ pairs of black pixels $\Db$ and $\abs{C_i}$ pixel pairs $\Dg$, which represent green signals, such that for variable $x\in C$ we find a green signal in $w_i$ which is preceded by a total of $x-1$ black pixel pairs (and possibly some green signals) in $w_i$:
\[
	w_i = u \Dg v \quad \iff \quad \abs{u}_\Db + 1\in C
\]
(for a sequence $u$ and a symbol $s$, the integer $\abs{u}_s$ equals the number of occurrences of $s$ in $u$ ignoring all other symbols in $u$).
Note that the width of the pattern $P$ is in $\Oh(m\cdot n)$.
The pattern $P$ can be computed from $k$ and $F$ in polynomial time.
See Example~\ref{ex:formula:to:pattern} at the end of this section for the conversion of a short formula into a pattern.

Suppose that the pattern $P$ can be self-assembled with $13$ tile types.
As $P$ includes the subpattern $G$, we only have to focus on the question of whether or not $P$ can be self-assembled by the tile set $T$, shown in Figure~\ref{fig:tiles}; see Lemma~\ref{prop:gadget}.
Since we only have two black tile types in $T$ and these can only occur in an uncover row, it is clear that all tiles in $c_0$ and $c_t$ have the glue \unc on all of their east and west edges, and that they are the only tiles with that glue on their east and west edges.
Obviously, these three uncover tile types, uniquely define the assembly of the initial and target configuration.
Also, the white uncover tile type cannot be used in any position in the white area between the initial and target configuration; otherwise, there had to be a complete white row with only this tile type which implied that all rows above had to have red signals as south input, and some of these red signals ultimately had to appear in the target configuration.
Also the glue sequence on the south part of the seed is determined by $c_0$.

Observe that the following sequence of vertical glues can be used as part of the west seed which lies above the target configuration from bottom to top
\[
	{}\SIr\bul^4\SIr\bul\SIr\bul^2\SIr
	\SIu^2\bul\SIu^2\SIs\SIu^4\SIb\SIu^4\SIr\SIu^2\bul\SIu^2\SIs\SIu^2{}.
\]
The first 11 glues in this sequence together with the green signals in $w_G$ form the south input of the gadget pattern $G$, thusly, allowing the gadget pattern to self-assemble in the top left corner of $P$. (This is because the pattern grows from the bottom left to the top right, and red signals which are input into the south of the gadget pattern grow diagonally up and right, starting from this sequence, with some of them moved further to the right when crossing green signals.)

After we showed that the part of the pattern $P$ which lies above the two rows $c_t$ can be self-assembled, we will prove that the lower part of the pattern up to and including the two rows $c_t$ can be self-assembled if and only if $F$ is satisfiable.
Let the sequence $z$ denote the $k + n$ west glues in between the rows $c_0$ and $c_t$ from bottom to top.
Note that this is the only part of the pattern that is not fixed yet; it will carry the information of how to choose the variable assignment in order to satisfy $F$.
We are interested only in the blue signals which are hidden in $z$ and we do not care about any red signals which might be included in $z$; hence, when convenient we assume $z\in \set{\wild,\SIs}^*$ where $\wild$ represents no signal or a red signal.
Consider the sequence $v_0 = \Db^k\Dr \Db^n$ in $c_0$ which turns into $v_t = \Db^{k+n}\Dr$ in $c_t$.
Note that the red signal in $v_t$, encoded by $\Dr$, is is the only red signal in the entire final configuration $c_t$.
Since there are no green signals encoded in $v_0$ and by Lemma~\ref{lem:red:signal}, this red signal has to be initiated at most $k+n$ positions to the left in $c_0$.
Therefore, the origin of the red signal in $v_t$ is indeed the red signal encoded in $v_0$.
Furthermore, as this red signal moves by $n$ positions to the right from $c_0$ to $c_t$, we conclude that exactly $k$ of the white rows contain a hidden blue signal, by Lemma~\ref{lem:red:signal}.

Now, we will show that if $F$ is satisfiable, then the pattern $P$ can be self-assembled by $T$.
Let $\phi\colon V\to \set{0,1}$ be a satisfying variable assignment for $F$ with $k  = \abs{\sett{x\in V}{\phi(x) = 1}}$.
We let $z$ contain $n$ non-signal glues $\bul$ and for each variable $x$ with $\phi(x) = 1$ we add a blue signal in $z$ which is preceded by $x$ \bul glues and possibly other blue signals:
\[
	z = u \SIs v \quad \iff \quad \phi(\abs{u}_\bul) = 1.
\]
Note that $\abs z = n+k$ and that $z$ contains $k$ blue signals as $k$ variables have to be true in $\phi$.
Consider $\Db\Dr w_i$ in $c_0$ which represents the clause $C_i\in F$.
In $c_t$ this sequence turns into $\Db\Db w_i$ where all the green signals appear in the same positions, but the red signal does not appear anymore.
Note that if this red signal in $c_0$ were not destroyed, then it had show up as the last pixel of $w_i$ in $c_t$, by Lemma~\ref{lem:red:signal}.
As $\phi$ satisfies $F$ there is a variable $x\in C_i$ with $\phi(x) = 1$.
By definition of $w_i$, we have $w_i = u' \Dg v'$ with $\abs{u'}_\Db + 1 = x$ and we have $z = u \SIs v$ with $\abs{u}_\bul = x$.
Suppose that the red signal does not get destroyed before it passes the green signal representing $x$.
After passing this green signal, the red signal has moved $x$ columns without signals to the right (the last of these columns lies right of the green signal), therefore, the signal also traveled $x$ rows without signals upwards and it has to cross the blue signal which corresponds to $\phi(x) = 1$ in the next step, see Figure~\ref{fig:destroy:red}.
We see that the red signal is destroyed on the way from $c_0$ to $c_t$, and hence, this sequence in $c_0$ can successfully be transformed into the corresponding sequence in $c_t$.
As this argument holds for every clause and we already showed how the remainder of the pattern $P$ can be self-assembled, we conclude that if $F$ is satisfiable with $k$ true variables, then $P$ can be self-assembled with $13$ tile types.

\begin{figure}[ht]
\centering
\begin{tikzpicture}
	\InitCoordinates	
	\Pub\Pur\Pub\MoveRight\Pub\Pug\Pub
	\LineUp	
	\Pub\Pug\Pub\MoveRight\Pub\Pur\Pub
	\LineUp
	\Pbb\Pbr\Prb\MoveRight\Pbb\Pbg\Pbb
	\LineUp
	\LineUp
	\Pbb\Pbb\Pbb\MoveRight\Pbr\Prg\Prb
	\LineUp
	\Psb\Psb\Psb\MoveRight\Psb\Psg\PSr

	\node at (3,1) {$\cdots$};
	\node at (3,4.5) {$\cdots$};
	\node at (1,3) {$\vdots$};
	\node at (5,3) {$\vdots$};

\draw[|<->|] (1.5,-.75) --
	node [inner sep=1pt,fill=white] {$x$ plain columns} (6.5,-.75);
\draw[|<->|] (-.75,1.5) --
	node [rotate=90,inner sep=1pt,fill=white] {$x$ plain rows} (-.75,4.5);

\end{tikzpicture}
\caption{The red signal is destroyed by the signals representing the variable $x$ with $\phi(x) = 1$. By a plain row or column we mean a row or column without signal.}
\label{fig:destroy:red}
\end{figure}

Vice versa, suppose $P$ can be self-assembled by the tile set $T$ and an appropriately chosen seed.
We define a variable assignment $\phi$ based on the sequence of west glues $z$, as defined above.
We let $\phi(x) =1$ if and only if there is a blue signal in $z$ which is preceded  by $x$ glues $\wild$ (recall that, \wild stands for everything other than a blue signal):
\[
	z = u \SIs v \quad \iff \quad \phi(\abs{u}_\wild) = 1.
\]
As there are $k$ blue signals in $z$ we have that the number of positive variable assignments in $\phi$ is at most $k$.
Due to the monotone nature of the formula $F$ (it does not contain negations), if we prove that $F$ can be satisfied with at most $k$ positive variables in the assignment, then $F$ can also be satisfied with exactly $k$ positive variables in the assignment.
As we already observed above, the pixel sequence $\Db\Dr w_i$ in $c_0$ which represents the clause $C_i\in F$ is transformed into $\Db\Db w_i$ in $c_t$, and hence, the red signal has to be destroyed within the white area of the pattern.
Note that a blue signal in the first position of $z$ will not destroy the red signal as the red signal does not have a green signal as left neighbor.
Thus, the red signal is destroyed when it hits one signal which stands for a positively assigned variable $x$ in $\phi$.
The green signal which is also involved in the destruction of the red signal is separated by $x-1$ black pixel pairs $\Db$ from the red signal in $\Db\Dr w_i$; see Figure~\ref{fig:destroy:red} again.
The existence of this green signal implies that $x\in C_i$ and the clause is satisfied by $\phi$.
We conclude that $\phi$ satisfies every clause in $F$, and hence, $F$ itself.
\end{proof}

\begin{example}\label{ex:formula:to:pattern}
Figure~\ref{fig:formula:to:pattern} shows a tile assignment for the pattern which is generated from the formula $(x\lor y)\land (y\lor z)$ with $k = 1$.
The part of the pattern which is needed for the gadget attachment is ignored in this example.
\begin{figure}[ht]
\resizebox{\linewidth}{!}{
\begin{tikzpicture}
	\InitCoordinates
	\Pub\Pur\Pub\Pub\Pub\Pub\Pur\Pug\Pub\Pug\Pub\Pub\Pub\Pur\Pub\Pug\Pub\Pug\Pub
	\LineUp
	\Pub\Pug\Pub\Pub\Pub\Pub\Pug\Pur\Pub\Pur\Pub\Pub\Pub\Pug\Pub\Pur\Pub\Pur\Pub
	\LineUp
	\Pbb\Pbr\Prb\Pbb\Pbb\Pbb\Pbr\Prg\Prb\Pbg\Pbb\Pbb\Pbb\Pbr\Prb\Pbg\Pbb\Pbg\Pbb
	\LineUp
	\Pbb\Pbb\Pbr\Prb\Pbb\Pbb\Pbb\Pbg\Pbr\Prg\Prb\Pbb\Pbb\Pbb\Pbr\Prg\Prb\Pbg\Pbb
	\LineUp
	\Psb\Psb\Psb\Psr\Psb\Psb\Psb\Psg\PSb\Psg\PSr\Psb\Psb\Psb\Psb\Psg\PSr\Psg\PSb
	\LineUp
	\Pbb\Pbb\Pbb\Pbr\Prb\Pbb\Pbb\Pbg\Pbb\Pbg\Pbb\Pbb\Pbb\Pbb\Pbb\Pbg\Pbb\Pbg\Pbb
	\LineUp
	\Pub\Pub\Pub\Pub\Pur\Pub\Pub\Pug\Pub\Pug\Pub\Pub\Pub\Pub\Pub\Pug\Pub\Pug\Pub
	\LineUp
	\Pub\Pub\Pub\Pub\Pug\Pub\Pub\Pur\Pub\Pur\Pub\Pub\Pub\Pub\Pub\Pur\Pub\Pur\Pub
	
	\draw [decorate,decoration={brace,mirror}] (-.45,-.6) -- node [below=2pt] {blue signal counter (here, $k=1$)} (4.45,-.6);
	\draw [decorate,decoration={brace,mirror}] (4.55,-.6) -- node [below=2pt] {$(x\lor y)$} (11.45,-.6);
	\draw [decorate,decoration={brace,mirror}] (11.55,-.6) -- node [below=2pt] {$(y\lor z)$} (18.45,-.6);
\end{tikzpicture}
}
\caption{Pattern for the formula $(x\lor y)\land (y\lor z)$ with $k = 1$.
The position of the blue signal represents the satisfying variable assignment $\phi(y) = 1$.}
\label{fig:formula:to:pattern}
\end{figure}
\end{example}

\section{Programmatic search for the minimal tile set}
\label{programmatic}

We present our algorithm for finding all tile sets of a given size $\ell$ which self-assemble a given $k$-colored pattern on a rectangle (in our case $\ell = 13$ and $k = 2$).
In Section~\ref{subsect:assemblies}, we show that it is sufficient to generate all {\em valid tile assemblies} for the given pattern, which use at most $\ell$ tile types, rather than generating all tile sets with all possible $L$-shaped seeds.
We present our algorithm which generates these valid tile assemblies and the corresponding tile sets in Section~\ref{subsect:algorithm} and discuss several methods which we implemented to speed up the algorithm.
Lastly, we discuss the parallel implementation and the performance of the algorithm in the two programming languages C++ (Section~\ref{subsect:cpp}) and Haskell (Section~\ref{subsect:haskell}).

\subsection{RTASs defined by assemblies}
\label{subsect:assemblies}

Consider the $k$-colored pattern $P\colon [m]\times[n] \to [k]$.
Recall that for a tile set $T = \set{t_0,\ldots,t_{\ell-1}}$ a terminal assembly (without the seed structure) of $P$ is a mapping $\alpha\colon [m]\times [n] \to [\ell]$ such that
\begin{enumerate}[\ (a)\ ]
\item\label{color:match}
if $\alpha(x,y) = \alpha(x',y')$, then $P(x,y) = P(x',y')$ for all $x,x'\in [m]$ and $y,y'\in[n]$,
\item\label{vertical:match}
$t_{\alpha(x,y)}(\south) = t_{\alpha(x,y-1)}(\north)$ for all $x\in[m]$ and $y\in\set{1,\ldots,n-1}$, and
\item\label{horizontal:match}
$t_{\alpha(x,y)}(\west) = t_{\alpha(x-1,y)}(\east)$ for all $x\in\set{1,\ldots,m-1}$ and $y\in[n]$.
\end{enumerate}
Condition~(\ref{color:match}) implies that every tile type can only have one color and conditions~(\ref{vertical:match}) and~(\ref{horizontal:match}) ensure that there are no vertical or horizontal glue mismatches in the assembly.
An (partial) assembly is a partial mapping $\alpha\colon [m]\times [n] \to_p [\ell]$ which satisfies the three conditions for all positions which are defined in $\alpha$.
Every RTAS $\cT = (T,\sigma_L)$ that self-assembles $P$ yields a terminal assembly $\alpha\colon [m]\times [n] \to [\abs{T}]$ by enumerating the tiles in $T$ where the seed is implicitly constituted by the south glues of the tiles $\alpha(x,0)$ for $x\in[m]$ and the west glues of the tiles $\alpha(0,y)$ for $y\in[n]$.

Conversely, every mapping $\alpha\colon [m]\times[n] \to [\ell]$ which satisfies condition~(\ref{color:match}) yields a (not necessarily directed) tile set $T_\alpha = \set{t_0,\ldots,t_{\ell-1}}$ where condition~(\ref{vertical:match}) imposes equivalence classes on the vertical glues and condition~(\ref{horizontal:match}) imposes equivalence classes on the horizontal glues (e.g.\ the glue $t_{\alpha(0,0)}(\east)$ belongs to the same equivalence class as the glue $t_{\alpha(1,0)}(\west)$).
For each of these equivalence classes we reserve one unique glue label in $T_\alpha$; in particular, no vertical glue gets the same label as a horizontal glue.
Thus, $\alpha$ is a terminal assembly of $P$ for $T_\alpha$.
Next, we show that if $\alpha$ is a terminal assembly of $P$ for a tile set $T$, then $T$ is a morphic image of $T_\alpha$; that is, there exists a bijection of tile types $h\colon T_\alpha\to T$, and a morphisms $g$ from the glues of $T_\alpha$ to the glues of $T$ such that for all $t\in T_A$ and $d\in\set{\north,\east,\south,\west}$ we have $c(t) = c(h(t))$ and $g(t(d)) = h(t(d))$.
Let $T = \set{t_0,\ldots,t_{\ell-1}}$ and $T_\alpha = \set{s_0,\ldots,s_{\ell-1}}$ be the chosen tile enumerations with respect to the assembly $\alpha$, then the bijection $h$ is chosen such that $h(s_i) = t_i$ for all $i\in[\ell]$.
Since both, $T$ and $T_\alpha$, have to satisfy (\ref{color:match}), we obtain that $c(s_i) = c(h(s_i)) = c(t_i)$ as desired.
Furthermore, $T_\alpha$ was defined such that it satisfies the minimal requirements for $\alpha$ to be an assignment according to conditions~(\ref{vertical:match}) and~(\ref{horizontal:match}).
Because $T$ must also satisfy these two conditions, it is clear that the morphism $g$ can be defined.

Note that the fact that $T$ is a morphic image of $T_\alpha$ implies that if $T$ is a directed tile set, then $T_\alpha$ is directed as well (though, the converse does not necessarily hold).
Henceforth, we call an assembly $\alpha$ {\em valid} if it is terminal and its corresponding tile set $T_\alpha$ is directed.
The algorithm that we present next lists all valid assemblies of $P$ together with their corresponding directed tile sets with at most $\ell$ tile types.
Therefore, up to morphic images of these solution tile sets, it lists all directed tile sets which can self-assemble $P$.
Also note that if a directed tile set $S$ is a morphic image of our tile set $T$ shown in Fig.~\ref{fig:tiles}, then $T$ and $S$ are isomorphic.
This can easily be verified as every tile set which is obtained by combining any two horizontal glues or any two vertical glues in $T$ is an undirected tile set.

\subsection{The algorithm}
\label{subsect:algorithm}

Instead of fully generating every terminal assembly $\alpha$ of the pattern $P$ and then checking whether or not the corresponding tile set $T_\alpha$ is directed, we generate partial assemblies tile by tile while adapting a generic tile set in each step such that it satisfies conditions~(\ref{color:match}) through~(\ref{horizontal:match}) from Section~\ref{subsect:assemblies}.
If a tile set $T_\alpha$ which corresponds to an assembly $\alpha$ is not directed, then we do not have to place any further tiles into this assembly because any larger assembly $\beta$ which contains $\alpha$ as subassembly has a corresponding undirected tile set $T_\beta$ and, hence, $\alpha$ cannot be completed to become a valid assembly.
This procedure can be illustrated in a tree spanning the search space where every node is a partial assembly with corresponding tile set.
Its root is the empty assembly (no tiles are placed) whose corresponding tile set consists of $\ell = 13$ tile types with every glue of every tile type unique and all tiles un-colored.
Leaves in this tree are either solutions, valid assemblies of $P$ with a corresponding directed tile set, or breakpoints, nodes whose tile sets are not directed.

The tiles are placed according to a tile placing strategy; that is, each position in $\alpha$ has a successor position where the next tile is placed.
The correctness of the algorithm does not depend on the tile placing strategy, however, the performance of the algorithm highly depends on this strategy.
Our strategy is to keep the area that is covered by tiles as compact as possible.
Performance tests on small patterns confirmed that the average depth of paths in the tree spanning the search space is smaller when using our strategy as compared to the naive row-by-row or column-by-column approaches.
The ordering of positions is illustrated in Fig.~\ref{fig:order}, and is intuitively
defined by \emph{``the alternative addition of a row and a column''}, starting as shown in the figure.
Formally, this amounts to defining a sequence of coordinates
$(x_i,y_i)_{n\in\mathbb{N}}$ inductively by $(x_0,y_0) = (0,0)$ and
\[
	(x_{i+1},y_{i+1}) =
	\begin{cases}
		(0,y_i +1) & \text{if } x_i = y_i, \\
		(x_i +1,0) & \text{if } x_i = y_i -1, \\
		(x_i,y_i +1) & \text{if } x_i > y_i, \\
		(x_i +1,y_i) & \text{otherwise}.
	\end{cases}
\]
The cases in the formula can be interpreted as follows, from top to bottom:
start a new row, start a new column, add one tile to an existing column, add one tile to an existing row.
This simplified formula suggests that the pattern has to be a square, but it can also be interpreted as ordering on all positions in the rectangular pattern $P$ by simply skipping positions which lie outside of $P$.

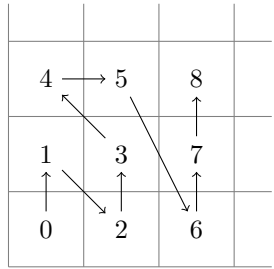
\begin{figure}[ht]
  \begin{center}
    \begin{tikzpicture}

      \foreach \x in {0,...,3}
        {\draw[gray](\x,0)--(\x,3.5);};
      \foreach \y in {0,...,3}
        {\draw[gray](0,\y)--(3.5,\y);};
      \draw(0.5,0.5)node(a){$0$}
      (0.5,1.5)node(b){$1$}
      (1.5,0.5)node(c){$2$}
      (1.5,1.5)node(d){$3$}
      (0.5,2.5)node(e){$4$}
      (1.5,2.5)node(f){$5$}
      (2.5,0.5)node(g){$6$}
      (2.5,1.5)node(h){$7$}
      (2.5,2.5)node(i){$8$};
      \draw[->](a)--(b);
      \draw[->](b)--(c);
      \draw[->](c)--(d);
      \draw[->](d)--(e);
      \draw[->](e)--(f);
      \draw[->](f)--(g);
      \draw[->](g)--(h);
      \draw[->](h)--(i);
    \end{tikzpicture}
  \end{center}
  \caption{Order on positions}
  \label{fig:order}
\end{figure}

Let $\alpha$ be a partial assembly in which exactly the first $i \in [n\cdot m-1]$ positions, according to the tile placing strategy, are covered with tiles and let $T_\alpha$ be the corresponding tile set which we assume to be directed.
Therefore, $(\alpha,T_\alpha)$ can be viewed as a node in the tree spanning the search space which is not a leave.
We try out all possible tile types in the empty position $\alpha(x_{i+1},y_{i+1})$ as follows:
\begin{enumerate}
\item
If there is a tile type $t$ in the current tile set $T$ which fits (i.e.\ its glues match those adjacent to the position and its color matches $c(t) = P(x_{i+1},y_{i+1})$), a tile of that type is placed in $\alpha(x_{i+1},x_{i+1})$.
Note that due to the ordering of tile placements, the adjacent glues (if any) will always be to the west and/or south, ensuring that those are the input sides.
If the location is on the bottom (left) edge of the pattern, there will not be an input glue on the south (west).\label{already}
\item
Else if, for each tile type $t$ which has already been placed somewhere in the assembly and which has color $c(t) = P(x_{i+1},y_{i+1})$, the west and south glues of $t$ can be changed to match those adjacent to the current position, wherever they occur throughout the current tiles of the assembly (modifying additional tile types as necessary).
If the tile set remains directed (i.e.\ no tile types have the same input glues), then the glue changes are made and $t$ is placed in the current location.\label{merge}
\item Else, if the number of tile types which are used in the partial assembly are less than $\ell = 13$, change the glues of one unused type so that it matches those adjacent to the current location, assign the color $c(t) \gets P(x_{i+1},y_{i+1})$, and place a tile of that type in position $\alpha(x_{i+1},y_{i+1})$.\label{newtile}
\end{enumerate}

Note that this procedure is optimized such that it will not generate two assemblies which are permutations of each other because we do not try several {\em unused tile types in the same position}.
The tree spanning the search space which is defined through this procedure is recursively traversed in a depth-first manner.

If this procedure finds a valid assignment $\alpha$ of $P$ with a corresponding directed tile set $T_\alpha$, then we output $(\alpha,T_\alpha)$ as solution.
As discussed in Section~\ref{subsect:assemblies}, this algorithm will output all directed tile sets which can self-assemble $P$ up to morphic images.
Both, the Haskell and C++ version of our program, are parallelized implementations of the algorithm described here.

\subsection{Implementation in C++}
\label{subsect:cpp}

The C++ code uses MPI \footnote{The implementation is Open MPI: \url{http://www.open-mpi.org}} as its communication protocol and a simple strategy for sharing the work among the cores.
The master process generates a list containing all partial assemblies in which exactly 14 positions are covered by tiles and whose corresponding tile sets are directed.
This list contains 271,835 partial assemblies, or {\em jobs}, in the case of our gadget pattern $G$ from Section~\ref{sec:2pats}.
The master sends out one of these jobs to each of the client processes.
Afterwards, the master process only gathers the results of jobs that were finished by clients and assigns new jobs to clients on request.
When the list is empty the master sends a kill signal to each client process that requests a new job.

A client process which got assigned the partial assembly $\alpha$ generates all valid assemblies of $P$ which contain $\alpha$ as a subassembly with corresponding directed tile set, using the algorithm described in Section~\ref{subsect:algorithm}.
When one job is finished, possible solutions are transmitted to the master and a new job is requested by the client.
In this implementation we did not address the computational bottleneck that emerges when client processes finish the last jobs and then have to idle until the last client process is finished.
There is no concept of sharing a job after it has been assigned to a client.

The C++ implementation of our algorithm\footnote{Freely available for download here: \url{http://self-assembly.net/wiki/index.php?title=2PATS-tileset-search}} was run on the cluster {\tt saw.sharcnet.ca} of \mbox{Sharcnet}\footnote{\url{https://www.sharcnet.ca/my/systems/show/41}}.
The cluster allowed us to utilize the processing power of 256 cores of Intel Xeon 2.83GHz (out of the total 2712 cores) for our computation.
In order to minimize the chances of the already unlikely event that undetected network errors influenced the outcome of the computation, our program was run twice on this system, with both runs yielding the same result, namely that the tile set $T$ from Fig.~\ref{fig:tiles} is the only tile set (up to isomorphism) with 13 or less tile types capable of generating the gadget pattern $G$, thus proving Lemma~\ref{prop:gadget}.
Each of the computations finished after almost 35 hours using a total CPU time of approximately 342 days.
Note that this implies a combined CPU idle time of about 30 days for the clients which we assume to be chiefly caused by the computational bottleneck at the end of the computation.
During one computation all the cores together generated over $66\cdot 10^{12}$ partial tile assemblies.

\subsection{Implementation in Haskell}
\label{subsect:haskell}

The kind of intensive proofs our approach uses has traditionally been proven
``rigorously'', with consensual proofs, several years after their first
publication. This means that some of the latest proofs rely on the simplicity of
their implementation to make them checkable. Moreover, proof assistants like Coq
are not yet able to provide a fast enough alternative, to verify really large
proofs in a reasonable amount of time.

Things are beginning to change, however, and the gap between rigorous and
algorithmic proofs is being progressively bridged.  The ultimate goal of this
research direction is to get rigorous proofs as the first proof of a theorem,
even in the case of explorations run on large parallel computers.

In order to reach this goal for Lemma \ref{prop:gadget}, we wrote a library to
be used for additional algorithmic proofs whose size requires a parallel
implementation. It also allows to work with different computing platforms,
including grids, clusters and desktop computers.

This library, called \emph{Parry}, is available at
\url{http://parry.lif.univ-mrs.fr}.
Its version 0.1, with SHA1 sum {\tt d572cbb7189d1d8232982913c70b3cfcb72ec44f},
is rigorously proven in the appendix of this paper.

Remarkably enough for a parallel program, its proof makes no hypotheses on the
network, and only relies on the equivalence of Haskell semantics and its compiled
assembly version, as well as on the security of cryptographic primitives.

Unfortunately, this implementation is relatively slower than the C++ one
described in section \ref{subsect:cpp}, and is still running at the time of
submitting this paper.

\section{Acknowledgements}

We would like to thank Manuel Bertrand for his infinite patience and helpful
assistance with setting up the server and helping debug our network and system
problems. We also thank C\'ecile Barbier, Eric Fede and Kai Poutrain for their
assistance with software setup.

	\bibliographystyle{abbrv}
	\bibliography{2pats}

\pagebreak

\appendix

We now discuss the Haskell implementation of the proof of Lemma \ref{prop:gadget},
and give a proof of this implementation.

\section{Global overview of the architecture}
\label{subsec:architecture}

Our system is composed of two main components, a ``\texttt{server}'' and a ``\texttt{client}''.  The \texttt{server} orchestrates the work done by a collection of \texttt{clients} by assigning \emph{jobs} (where a job is a current tile set and partial assembly) to each, monitoring their progress, and recording all discovered solutions.  The \texttt{clients} are assigned jobs by the server and perform the actual testing of all possible tile sets within the fixed size bound (i.e. $13$ tile types) to see if they can self-assemble the input pattern.  To prove the correctness of the system, we will individually prove the correctness of the \texttt{server} and \texttt{client}.  The main result to be proven for the system is the following:

\begin{lemma}
The \texttt{server} completes its search if and only if all tilesets of size
$\leq 13$ (up to isomorphism) which can self-assemble the input pattern have
been discovered.
\end{lemma}

The task of the \texttt{server} is to assign and keep track of all jobs which are being explored by the \texttt{clients}. Each \texttt{client} connects to it to ask for a job assignment. The \texttt{server} then replies with an assignment and keeps track of that job in case the \texttt{client} crashes, in which case the \texttt{server} will be able to detect that (in a way to be discussed) and reassign the job to another \texttt{client}. Along with that job, the \texttt{server} sends a boolean indicating whether it expects the job to be re-shared.

The \texttt{clients'} messages to the \texttt{server} can be of three kinds:
``get job'' messages, new jobs (in our case, new tilesets and new partial
assemblies to be explored), or a ``job done'' message, to tell the
\texttt{server} that the job has been completed.

Formally, we can represent all possible states of the \texttt{clients} with the
graph of Figure \ref{fig:clientloop}. In this figure, nodes drawn in solid lines
represent states in which the client sends a message to the \texttt{server},
and edges represent state transitions of the \texttt{client} based on messages
received from the \texttt{server}. Red edges are followed when the
\texttt{server} sends an unexpected message (which is normally a ``re-share''
message, or one caused by the detection of an attacker), or when thread $T_2$
enters the ``Die'' red state.  Dashed nodes and edges are ``silent''
\texttt{client} states and transitions, where no messages are sent or received.

Finally, the ``heartbeat'' state causes the creation of thread $T_2$,
and the ``Die'' and ``Stop'' states cause that thread to terminate.
The green transition on that thread (from the ``Alive'' state to the ``Stop'' state)
is triggered whenever thread $T_1$ enters the ``Get job'' state.
The heartbeat thread otherwise exists in the ``Alive'' state.
Both in the case of the green and red states, synchronization mechanisms ensure
that the corresponding colored edges are followed in the other thread before
the node is exited.

\begin{figure}[ht]
\begin{center}\includegraphics[scale=0.5]{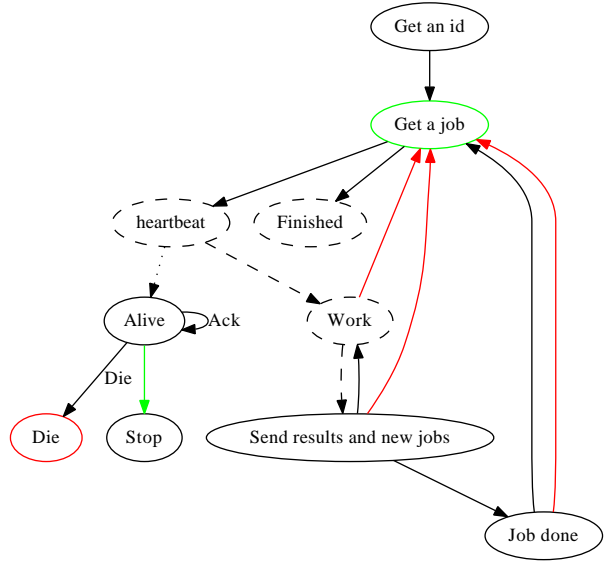}
\end{center}
\caption{State graph of the protocol, from the client's perspective. Clients
  send messages on each node drawn in full line, and receive messages on each
  edge drawn in full line. Red and green nodes trigger the red and green edges
  in the other thread, respectively. The dotted edge creates a new thread. }
\label{fig:clientloop}
\end{figure}

\section{The implementation}
\label{implementation}
Our strategy to prove the whole system is the following:

\begin{itemize}
\item First prove, in Section \ref{subsect:server}, an invariant on the server's
  state, conditioned on hypotheses called \emph{validity} and \emph{fluency} on
  the clients.
\item Then prove, in Section \ref{subsect:client}, that our clients respect the
  fluency and validity condition, if their worker function (which is the actual
  implementation of the algorithm described in Section \ref{subsect:algorithm}),
  shares the work properly.
\item Finally, prove, in Section \ref{subsect:pats} that our worker function
  shares the work properly.
\end{itemize}

The reason for this organization is to make the proof for the server and client
reusable in other applications.
We begin by defining what we mean by \emph{tasks}, and how they can be represented
in the server in an efficient way (by \emph{jobs}).

\begin{definition}
\label{def:tasks}

Let $T$ and $R$ be two sets, and $f:T\rightarrow 2^R$ be any function.  A
\emph{task} is an element $t\in T$, and a \emph{result} is an element $r\in R$.
If there is a $r\in 2^R$, and a set $\{t_1,\ldots,t_n\}$ such that
$f(t)=r\cup \bigcup_{1\leq i\leq n} f(t_i)$, we call $t_1,\ldots,t_n$ the
\emph{subtasks} of $t$.

We say that a task $t$ has been \emph{explored} when $f(t)$ is known.

\end{definition}

In our code, most tasks are computed by computing their subtasks. For performance
reasons, we cannot simply represent tasks in a direct way, and we need ``jobs''
instead:

\begin{definition}
\label{def:jobs}
Let $T$ be a set (of tasks), $R$ be a set of results, and $f:T\rightarrow 2^R$
be any function. For any set $J$, we say that $T$ is \emph{represented} by $J$
if there is an onto map $\rho:J\rightarrow T$. In this case, for any $j\in J$,
we say that $j$ is a \emph{job representation} of $\rho(j)$.

By extension, for any job $j$ representing some task $t$, jobs representing the
subtasks of $t$ are called \emph{subjobs} of $j$.

\end{definition}

\subsection{How to read the code, and what we prove on it}

The language we used to implement this architecture is the functional language
Haskell. Although the syntax of this language may be somewhat surprising at
first, the essential points that make our program easier to read, and easy to
prove, are:

\begin{itemize}

\item Our program makes only one use of mutable variables, in the server's
  global state.  This is necessary to synchronize state variables between
  different threads.

  All other variables that we use are \emph{non-mutable}, meaning that \emph{new
    variables} are created whenever a change is needed. For instance, in the
  client's {\tt placeTile} function, we will see in Section \ref{subsect:pats}
  that new vectors are allocated when we need to change them.

\item Global mutable variables are not allowed by the grammar and types of
  Haskell, hence all our functions depend \emph{on their arguments only} (and,
  of course, on global \emph{constants}).

\item The type of functions tells whether or not they have side effects: a
  function has side effects (that is, modifies one of its arguments, or sends a
  message on the network) if and only if its type ends with ``{\tt IO a}'' for some
  {\tt a} (for example ``{\tt IO ()}'').

\item Surprisingly, we do not need to prove anything about the messages
  \emph{sent by the server}: any message sent by the protocol can be interpreted
  as a full predicate, sufficient to build a part of the global proof.

  This makes the following kind of attacks possible: an attacker intercepts a
  ``job mission'' sent to a client and changes it. The client then starts to
  work on that job. However, when it sends ``predicates'' back to the server, these
  parts of the proof cannot be used, because they do not correspond to any
  ``proof goal'' in the server.

\item Parts of this paper were generated directly from actual Haskell code using
  tools for literate programming. Therefore, the line numbers are actual line
  numbers, and might not be contiguous, in particular if some lines contain
  annotations for literate programming.
\end{itemize}

Moreover, many details of our functions need not be proven. In fact the only
thing we need to prove is that the server does not halt before exploring
all the tilesets it needs to explore.

In particular, we will not prove the efficiency or complexity of our protocol,
nor the fact that the server will eventually halt on all runs: the fact
that it halts on at least one run is sufficient for Lemma
\ref{prop:gadget} to hold.

\newcounter{line}
\subsection{Protocol datatypes}
We need to introduce a few datatypes to represent results, jobs, and client and
server messages. Only representations of {\tt ClientMessage} and {\tt
  ServerMessage} via the {\tt encode} function, along with signed hashes, will
be passed on the network.

The first module we need is {\tt Parry.Protocol}, that defines the type of
messages exchanged between the client and server. We assume, in the rest of the
proof, that the {\tt encode} and {\tt decode} functions generated by the {\tt
Binary} instances of these datatypes verify {\tt decode.encode} is the identity.

\hfill

{\tt{}\small{}\setcounter{line}{14}
\noindent\stepcounter{line}\makebox[2em][l]{\theline}\hbox{\phantom{}{\color{RoyalBlue}module} {\color{Green}Parry}.{\color{Green}Protocol} {\color{RoyalBlue}where}}

\noindent\stepcounter{line}\makebox[2em][l]{\theline}\hbox{\phantom{}{\color{RoyalBlue}import} {\color{Green}Data}.{\color{Green}Binary}}

\noindent\stepcounter{line}\makebox[2em][l]{\theline}\hbox{\phantom{}{\color{RoyalBlue}import} {\color{Green}GHC}.{\color{Green}Generics} ({\color{Green}Generic})}

\noindent\stepcounter{line}\makebox[2em][l]{\theline}\hbox{\phantom{}{\color{RoyalBlue}import} {\color{Green}Codec}.{\color{Green}Crypto}.{\color{Green}RSA}.{\color{Green}Pure}}

\noindent\stepcounter{line}\makebox[2em][l]{\theline}\hbox{\phantom{}}

\noindent\stepcounter{line}\makebox[2em][l]{\theline}\hbox{\phantom{}{\color{Red}-- | The type of messages sent by the client, exposed here for full}}

\noindent\stepcounter{line}\makebox[2em][l]{\theline}\hbox{\phantom{}{\color{Red}-- disclosure of the protocol's proof.}}

\noindent\stepcounter{line}\makebox[2em][l]{\theline}\hbox{\phantom{}{\color{RoyalBlue}data} {\color{Green}ClientMessage} j=}

\noindent\stepcounter{line}\makebox[2em][l]{\theline}\hbox{\phantom{xx}{\color{Green}GetJob} {\color{Green}Integer} {\color{Green}PublicKey}}

\noindent\stepcounter{line}\makebox[2em][l]{\theline}\hbox{\phantom{xx}$|$ {\color{Green}JobDone} \{ clientId::{\color{Green}Integer}, jobResults::[j], currentJob::j \}}

\noindent\stepcounter{line}\makebox[2em][l]{\theline}\hbox{\phantom{xx}$|$ {\color{Green}NewJobs} \{ clientId::{\color{Green}Integer}, jobResults::[j], currentJob::j,}

\noindent\stepcounter{line}\makebox[2em][l]{\theline}\hbox{\phantom{xxxxxxxxxxxxxx}nextJob::j, newJobs::[j] \}}

\noindent\stepcounter{line}\makebox[2em][l]{\theline}\hbox{\phantom{xx}$|$ {\color{Green}Alive} {\color{Green}Integer}}

\noindent\stepcounter{line}\makebox[2em][l]{\theline}\hbox{\phantom{xx}{\color{RoyalBlue}deriving} ({\color{Green}Generic}, {\color{Green}Show})}

\noindent\stepcounter{line}\makebox[2em][l]{\theline}\hbox{\phantom{}}

\noindent\stepcounter{line}\makebox[2em][l]{\theline}\hbox{\phantom{}{\color{Red}-- | The type of messages sent by the server, exposed here for full}}

\noindent\stepcounter{line}\makebox[2em][l]{\theline}\hbox{\phantom{}{\color{Red}-- disclosure of the protocol's proof.}}

\noindent\stepcounter{line}\makebox[2em][l]{\theline}\hbox{\phantom{}{\color{RoyalBlue}data} {\color{Green}ServerMessage} j=}

\noindent\stepcounter{line}\makebox[2em][l]{\theline}\hbox{\phantom{xx}{\color{Green}Job} {\color{Green}Bool} j}

\noindent\stepcounter{line}\makebox[2em][l]{\theline}\hbox{\phantom{xx}$|$ {\color{Green}Finished}}

\noindent\stepcounter{line}\makebox[2em][l]{\theline}\hbox{\phantom{xx}$|$ {\color{Green}Ack}}

\noindent\stepcounter{line}\makebox[2em][l]{\theline}\hbox{\phantom{xx}$|$ {\color{Green}Die}}

\noindent\stepcounter{line}\makebox[2em][l]{\theline}\hbox{\phantom{xx}{\color{RoyalBlue}deriving} ({\color{Green}Generic},{\color{Green}Show})}

\noindent\stepcounter{line}\makebox[2em][l]{\theline}\hbox{\phantom{}}

\noindent\stepcounter{line}\makebox[2em][l]{\theline}\hbox{\phantom{}{\color{RoyalBlue}instance} ({\color{Green}Binary} j)$\Rightarrow${\color{Green}Binary} ({\color{Green}ClientMessage} j)}

\noindent\stepcounter{line}\makebox[2em][l]{\theline}\hbox{\phantom{}{\color{RoyalBlue}instance} ({\color{Green}Binary} j)$\Rightarrow${\color{Green}Binary} ({\color{Green}ServerMessage} j)}

}\vspace{1em}

\subsection{Proof of the server}
\label{subsect:server}

A good way to think of this code is the following: we are trying to build a
giant proof tree, so big that no memory or even hard drive can handle it all.
However, we can construct it ``lazily'' and in parallel, and at the same time
verify it. This is why all the messages sent and received, that modify the
server's state, can be thought of as parts of the proof, of the form \emph{``I, valid
client, hereby RSA-certify that the following is a valid part of the tree''}.

The design of the protocol is such that every message that modifies the server's
state contains a complete predicate (for instance \emph{``the results in job {\tt j}
are exactly {\tt r}''}, or
\emph{``While exploring {\tt j}'s subjobs, I found results {\tt r}, and the
remaining subjobs are
$\mathtt{j_0},\mathtt{j_1},\ldots,\mathtt{j_n}$''}.
The server's task is then to assemble these predicates, which are summaries of
subtrees of the proof, into a complete proof.

We now proceed to the proof of the server. The file is included for the sake of
completeness. In particular, remark that the whole state of the server is defined as
a single data type called {\tt State}. We will use this fact to prove
invariants on the whole server state.
\vspace{1em}

{\tt{}\small{}\setcounter{line}{24}
\noindent\stepcounter{line}\makebox[2em][l]{\theline}\hbox{\phantom{}{\color{RoyalBlue}module} {\color{Green}Parry}.{\color{Green}Server} (}

\noindent\stepcounter{line}\makebox[2em][l]{\theline}\hbox{\phantom{xx}{\color{Red}-- * Jobs on the server side}}

\noindent\stepcounter{line}\makebox[2em][l]{\theline}\hbox{\phantom{xx}{\color{Green}Exhaustive}(..),}

\noindent\stepcounter{line}\makebox[2em][l]{\theline}\hbox{\phantom{xx}{\color{Green}Result}(..),}

\noindent\stepcounter{line}\makebox[2em][l]{\theline}\hbox{\phantom{xx}{\color{Red}-- * Server's internal state}}

\noindent\stepcounter{line}\makebox[2em][l]{\theline}\hbox{\phantom{xx}initState,}

\noindent\stepcounter{line}\makebox[2em][l]{\theline}\hbox{\phantom{xx}stateFromFile,}

\noindent\stepcounter{line}\makebox[2em][l]{\theline}\hbox{\phantom{xx}saveThread,}

\noindent\stepcounter{line}\makebox[2em][l]{\theline}\hbox{\phantom{xx}{\color{Green}State}(..),}

\noindent\stepcounter{line}\makebox[2em][l]{\theline}\hbox{\phantom{xx}{\color{Red}-- * Server configuration and functions}}

\noindent\stepcounter{line}\makebox[2em][l]{\theline}\hbox{\phantom{xx}{\color{Green}Config}(..),defaultConfig,server}

\noindent\stepcounter{line}\makebox[2em][l]{\theline}\hbox{\phantom{xx}) {\color{RoyalBlue}where}}

\noindent\stepcounter{line}\makebox[2em][l]{\theline}\hbox{\phantom{}{\color{RoyalBlue}import} {\color{Green}Control}.{\color{Green}Concurrent}}

\noindent\stepcounter{line}\makebox[2em][l]{\theline}\hbox{\phantom{}{\color{RoyalBlue}import} {\color{Green}Control}.{\color{Green}Exception} {\color{RoyalBlue}as} {\color{Green}E}}

\noindent\stepcounter{line}\makebox[2em][l]{\theline}\hbox{\phantom{}{\color{RoyalBlue}import} {\color{Green}Control}.{\color{Green}Monad}}

\noindent\stepcounter{line}\makebox[2em][l]{\theline}\hbox{\phantom{}{\color{RoyalBlue}import} {\color{Green}Control}.{\color{Green}Concurrent}.{\color{Green}MSem} {\color{RoyalBlue}as} {\color{Green}Sem}}

\noindent\stepcounter{line}\makebox[2em][l]{\theline}\hbox{\phantom{}{\color{RoyalBlue}import} {\color{Green}Network}}

\noindent\stepcounter{line}\makebox[2em][l]{\theline}\hbox{\phantom{}{\color{RoyalBlue}import} {\color{Green}System}.{\color{Green}IO}}

\noindent\stepcounter{line}\makebox[2em][l]{\theline}\hbox{\phantom{}{\color{RoyalBlue}import} {\color{Green}System}.{\color{Green}Directory}}

\noindent\stepcounter{line}\makebox[2em][l]{\theline}\hbox{\phantom{}{\color{RoyalBlue}import} {\color{Green}Data}.{\color{Green}List}}

\noindent\stepcounter{line}\makebox[2em][l]{\theline}\hbox{\phantom{}{\color{RoyalBlue}import} {\color{Green}Data}.{\color{Green}Time}.{\color{Green}Format}()}

\noindent\stepcounter{line}\makebox[2em][l]{\theline}\hbox{\phantom{}{\color{RoyalBlue}import} {\color{Green}Data}.{\color{Green}Time}.{\color{Green}Clock}.{\color{Green}POSIX}}

\noindent\stepcounter{line}\makebox[2em][l]{\theline}\hbox{\phantom{}{\color{Purple}\#}ifdef {\color{Green}UNIX}}

\noindent\stepcounter{line}\makebox[2em][l]{\theline}\hbox{\phantom{}{\color{RoyalBlue}import} {\color{Green}System}.{\color{Green}Posix}.{\color{Green}Signals}}

\noindent\stepcounter{line}\makebox[2em][l]{\theline}\hbox{\phantom{}{\color{Purple}\#}endif}

\noindent\stepcounter{line}\makebox[2em][l]{\theline}\hbox{\phantom{}}

\noindent\stepcounter{line}\makebox[2em][l]{\theline}\hbox{\phantom{}{\color{RoyalBlue}import} {\color{RoyalBlue}qualified} {\color{Green}Data}.{\color{Green}ByteString}.{\color{Green}Char8} {\color{RoyalBlue}as} {\color{Green}B}}

\noindent\stepcounter{line}\makebox[2em][l]{\theline}\hbox{\phantom{}{\color{RoyalBlue}import} {\color{RoyalBlue}qualified} {\color{Green}Data}.{\color{Green}ByteString}.{\color{Green}Lazy}.{\color{Green}Char8} {\color{RoyalBlue}as} {\color{Green}LB}}

\noindent\stepcounter{line}\makebox[2em][l]{\theline}\hbox{\phantom{}{\color{RoyalBlue}import} {\color{RoyalBlue}qualified} {\color{Green}Data}.{\color{Green}Map} {\color{RoyalBlue}as} {\color{Green}M}}

\noindent\stepcounter{line}\makebox[2em][l]{\theline}\hbox{\phantom{}{\color{RoyalBlue}import} {\color{RoyalBlue}qualified} {\color{Green}Data}.{\color{Green}Set} {\color{RoyalBlue}as} {\color{Green}S}}

\noindent\stepcounter{line}\makebox[2em][l]{\theline}\hbox{\phantom{}{\color{RoyalBlue}import} {\color{Green}GHC}.{\color{Green}Generics}}

\noindent\stepcounter{line}\makebox[2em][l]{\theline}\hbox{\phantom{}{\color{RoyalBlue}import} {\color{Green}Data}.{\color{Green}Binary}}

\noindent\stepcounter{line}\makebox[2em][l]{\theline}\hbox{\phantom{}}

\noindent\stepcounter{line}\makebox[2em][l]{\theline}\hbox{\phantom{}{\color{RoyalBlue}import} {\color{Green}Codec}.{\color{Green}Crypto}.{\color{Green}RSA}.{\color{Green}Pure}}

\noindent\stepcounter{line}\makebox[2em][l]{\theline}\hbox{\phantom{}}

\noindent\stepcounter{line}\makebox[2em][l]{\theline}\hbox{\phantom{}{\color{RoyalBlue}import} {\color{Green}Parry}.{\color{Green}Protocol}}

\noindent\stepcounter{line}\makebox[2em][l]{\theline}\hbox{\phantom{}{\color{RoyalBlue}import} {\color{Green}Parry}.{\color{Green}Util}}

\noindent\stepcounter{line}\makebox[2em][l]{\theline}\hbox{\phantom{}}

\noindent\stepcounter{line}\makebox[2em][l]{\theline}\hbox{\phantom{}{\color{Red}-- | The class of jobs and job results that Parry can deal with. For}}

\noindent\stepcounter{line}\makebox[2em][l]{\theline}\hbox{\phantom{}{\color{Red}-- efficiency and to keep types simple, jobs and results are stored in}}

\noindent\stepcounter{line}\makebox[2em][l]{\theline}\hbox{\phantom{}{\color{Red}-- a single type.}}

\noindent\stepcounter{line}\makebox[2em][l]{\theline}\hbox{\phantom{}{\color{RoyalBlue}class} {\color{Green}Exhaustive} j {\color{RoyalBlue}where}}

\noindent\stepcounter{line}\makebox[2em][l]{\theline}\hbox{\phantom{xx}{\color{Red}-- | Indication of the depth of a job in the explored tree. The server sends}}

\noindent\stepcounter{line}\makebox[2em][l]{\theline}\hbox{\phantom{xx}{\color{Red}-- the least deep jobs first, as an optimization of network use.}}

\noindent\stepcounter{line}\makebox[2em][l]{\theline}\hbox{\phantom{xx}depth :: j$\rightarrow${\color{Green}Int}}

\noindent\stepcounter{line}\makebox[2em][l]{\theline}\hbox{\phantom{xx}{\color{Red}-- | Number of times a job has been killed. When a job is killed,}}

\noindent\stepcounter{line}\makebox[2em][l]{\theline}\hbox{\phantom{xx}{\color{Red}-- either because it must be reshared, or because the client itself}}

\noindent\stepcounter{line}\makebox[2em][l]{\theline}\hbox{\phantom{xx}{\color{Red}-- was killed, it is scheduled to be re-executed by the server.}}

\noindent\stepcounter{line}\makebox[2em][l]{\theline}\hbox{\phantom{xx}killed::j$\rightarrow${\color{Green}Int}}

\noindent\stepcounter{line}\makebox[2em][l]{\theline}\hbox{\phantom{xx}{\color{Red}-- | Called each time a job needs to be killed. For better resharing,}}

\noindent\stepcounter{line}\makebox[2em][l]{\theline}\hbox{\phantom{xx}{\color{Red}-- this function must verify @killed (kill j) >= killed j@.}}

\noindent\stepcounter{line}\makebox[2em][l]{\theline}\hbox{\phantom{xx}kill::j$\rightarrow$j}

\noindent\stepcounter{line}\makebox[2em][l]{\theline}\hbox{\phantom{}}

\noindent\stepcounter{line}\makebox[2em][l]{\theline}\hbox{\phantom{}{\color{Red}-- | The class of results, and how to combine them in the server state.}}

\noindent\stepcounter{line}\makebox[2em][l]{\theline}\hbox{\phantom{}{\color{RoyalBlue}class} {\color{Green}Result} j r {\color{RoyalBlue}where}}

\noindent\stepcounter{line}\makebox[2em][l]{\theline}\hbox{\phantom{xx}{\color{Red}-- | A function to tell how to combine job results. That function will be}}

\noindent\stepcounter{line}\makebox[2em][l]{\theline}\hbox{\phantom{xx}{\color{Red}-- called on the hostname of the reporting client, with the finished job it}}

\noindent\stepcounter{line}\makebox[2em][l]{\theline}\hbox{\phantom{xx}{\color{Red}-- sent, and the current result from the server state.}}

\noindent\stepcounter{line}\makebox[2em][l]{\theline}\hbox{\phantom{xx}addResult::{\color{Green}HostName}$\rightarrow$r$\rightarrow$j$\rightarrow$r}

\noindent\stepcounter{line}\makebox[2em][l]{\theline}\hbox{\phantom{}}

\noindent\stepcounter{line}\makebox[2em][l]{\theline}\hbox{\phantom{}{\color{Red}-- | This type is exposed mostly for writing alternative user interfaces.}}

\noindent\stepcounter{line}\makebox[2em][l]{\theline}\hbox{\phantom{}{\color{Red}-- Other operations must be done using the functions in this module, or}}

\noindent\stepcounter{line}\makebox[2em][l]{\theline}\hbox{\phantom{}{\color{Red}-- the correction of the protocol can be lost.}}

\noindent\stepcounter{line}\makebox[2em][l]{\theline}\hbox{\phantom{}{\color{RoyalBlue}data} {\color{Green}State} j r={\color{Green}State} \{}

\noindent\stepcounter{line}\makebox[2em][l]{\theline}\hbox{\phantom{xx}{\color{Red}-- | Available jobs}}

\noindent\stepcounter{line}\makebox[2em][l]{\theline}\hbox{\phantom{xx}jobs::{\color{Green}S}.{\color{Green}Set} ({\color{Green}Int},j),}

\noindent\stepcounter{line}\makebox[2em][l]{\theline}\hbox{\phantom{xx}{\color{Red}-- | Map from the machine id to its hostname, its current job, its}}

\noindent\stepcounter{line}\makebox[2em][l]{\theline}\hbox{\phantom{xx}{\color{Red}-- starting time, the last time we heard from it.}}

\noindent\stepcounter{line}\makebox[2em][l]{\theline}\hbox{\phantom{xx}ongoing::{\color{Green}M}.{\color{Green}Map} {\color{Green}Integer} ({\color{Green}HostName},{\color{Green}PublicKey},j,{\color{Green}Double},{\color{Green}Double}),}

\noindent\stepcounter{line}\makebox[2em][l]{\theline}\hbox{\phantom{xx}{\color{Red}-- | Set of unemployed machines}}

\noindent\stepcounter{line}\makebox[2em][l]{\theline}\hbox{\phantom{xx}unemployed::{\color{Green}S}.{\color{Green}Set} {\color{Green}Integer},}

\noindent\stepcounter{line}\makebox[2em][l]{\theline}\hbox{\phantom{xx}{\color{Red}-- | The results.}}

\noindent\stepcounter{line}\makebox[2em][l]{\theline}\hbox{\phantom{xx}results::r,}

\noindent\stepcounter{line}\makebox[2em][l]{\theline}\hbox{\phantom{xx}{\color{Red}-- | The smallest available machine id. In a run of the server, it}}

\noindent\stepcounter{line}\makebox[2em][l]{\theline}\hbox{\phantom{xx}{\color{Red}-- is guaranteed that are never assigned the same.}}

\noindent\stepcounter{line}\makebox[2em][l]{\theline}\hbox{\phantom{xx}newId::{\color{Green}Integer},}

\noindent\stepcounter{line}\makebox[2em][l]{\theline}\hbox{\phantom{xx}{\color{Red}-- | Total number of jobs killed from the beginning (for benchmarking purposes).}}

\noindent\stepcounter{line}\makebox[2em][l]{\theline}\hbox{\phantom{xx}killings::{\color{Green}Int},}

\noindent\stepcounter{line}\makebox[2em][l]{\theline}\hbox{\phantom{xx}{\color{Red}-- | Number of jobs finished (for benchmarking purposes).}}

\noindent\stepcounter{line}\makebox[2em][l]{\theline}\hbox{\phantom{xx}solved::{\color{Green}Integer},}

\noindent\stepcounter{line}\makebox[2em][l]{\theline}\hbox{\phantom{xx}{\color{Red}-- | The list of authorized RSA public keys.}}

\noindent\stepcounter{line}\makebox[2em][l]{\theline}\hbox{\phantom{xx}authorizedKeys::[{\color{Green}PublicKey}]}

\noindent\stepcounter{line}\makebox[2em][l]{\theline}\hbox{\phantom{xx}\} {\color{RoyalBlue}deriving} ({\color{Green}Show},{\color{Green}Read},{\color{Green}Generic})}

\noindent\stepcounter{line}\makebox[2em][l]{\theline}\hbox{\phantom{}}

\noindent\stepcounter{line}\makebox[2em][l]{\theline}\hbox{\phantom{}{\color{RoyalBlue}instance} ({\color{Green}Binary} j,{\color{Green}Binary} r)$\Rightarrow${\color{Green}Binary} ({\color{Green}State} j r)}

}

\hfill

\begin{definition}
In a server state {\tt st}, the \emph{current job} of a client is the job
registered in the {\tt ongoing} field of {\tt st}.
\end{definition}

\begin{definition}
\label{def:valid}
We call a client \emph{valid} if, at the same time:
\begin{enumerate}

\item \label{valid:newjobs}
Its {\tt NewJobs} messages contain all the results in subjobs of its
current job that have been completely explored, and the subjobs of its current
job that have not been completely explored, divided into three fields: the
results it has found, its next current job, and other subjobs.

\item \label{valid:jobdone}
It does not send a {\tt JobDone} message before the task representing its
current job is completely explored.

\end{enumerate}
\end{definition}

The main function, {\tt answer}, keeps track of the clients. We now prove the
following Lemma:

\begin{lemma}
\label{lem:answer}
If {\tt st} is a state of the server containing (in the union of {\tt job st}
and {\tt ongoing st}) jobs representing all the tasks that have not yet been explored,
and for any job {\tt j}, {\tt j} and {\tt kill j} represent the same task,
then for any message {\tt m} sent by a valid client, all values of {\tt host}
and {\tt time}, {\tt answer host time st m} (the Haskell syntax for ``the value
of function {\tt answer}, called with arguments {\tt t}, {\tt host}, {\tt st} and
{\tt m}'')
is a couple
 $(\mathtt{st'},\mathtt{m'})$, where {\tt st'} is a state of the server
containing the roots of all subtrees that have not yet been explored ({\tt m'}
is the message to be sent to the client).

Moreover, all results sent by the clients are added to the server state using
the {\tt addResult} function.

\begin{proof}

We prove it for all the cases.

{\tt{}\small{}\setcounter{line}{117}
\noindent\stepcounter{line}\makebox[2em][l]{\theline}\hbox{\phantom{}{\color{Purple}answer}::({\color{Green}Exhaustive} j,{\color{Green}Result} j r,{\color{Green}Eq} j,{\color{Green}Ord} j,{\color{Green}Binary} j)$\Rightarrow$}

\noindent\stepcounter{line}\makebox[2em][l]{\theline}\hbox{\phantom{xxxxxxxx}{\color{Green}Double}$\rightarrow${\color{Green}String}$\rightarrow${\color{Green}State} j r$\rightarrow${\color{Green}ClientMessage} j}

\noindent\stepcounter{line}\makebox[2em][l]{\theline}\hbox{\phantom{xxxxxxxx}$\rightarrow$({\color{Green}State} j r,{\color{Green}ServerMessage} j)}

\noindent\stepcounter{line}\makebox[2em][l]{\theline}\hbox{\phantom{}{\color{Purple}answer} t host st ({\color{Green}GetJob} num key)=}

\noindent\stepcounter{line}\makebox[2em][l]{\theline}\hbox{\phantom{xx}{\color{RoyalBlue}case} {\color{Green}M}.lookup num (ongoing st) {\color{RoyalBlue}of}}

\noindent\stepcounter{line}\makebox[2em][l]{\theline}\hbox{\phantom{xxxx}{\color{Green}Just} (ho,key0,j0,t0,\_)$\rightarrow$}

}
If the client is registered as an ``ongoing'' job, we can simply send it the job
it is supposed to be working on. In this case, the invariant is still
maintained, as we do not change its recorded current job (here, we only update the
time at which we last saw this client).

{\tt{}\small{}\setcounter{line}{176}
\noindent\stepcounter{line}\makebox[2em][l]{\theline}\hbox{\phantom{xxxxxx}{\color{RoyalBlue}if} ho==host $\wedge$ key0==key {\color{RoyalBlue}then}}

\noindent\stepcounter{line}\makebox[2em][l]{\theline}\hbox{\phantom{xxxxxxxx}(st \{ ongoing={\color{Green}M}.insert num (ho,key,j0,t0,t) (ongoing st) \},}

\noindent\stepcounter{line}\makebox[2em][l]{\theline}\hbox{\phantom{xxxxxxxxx}{\color{Green}Job} (not \$ {\color{Green}S}.null \$ unemployed st) j0)}

\noindent\stepcounter{line}\makebox[2em][l]{\theline}\hbox{\phantom{xxxxxx}{\color{RoyalBlue}else}}

\noindent\stepcounter{line}\makebox[2em][l]{\theline}\hbox{\phantom{xxxxxxxx}(st,{\color{Green}Die})}

\noindent\stepcounter{line}\makebox[2em][l]{\theline}\hbox{\phantom{xxxx}{\color{Green}Nothing}$\rightarrow$}

}Else, client {\tt num} is not in the map of ongoing jobs.
If there are no more jobs to be done:

\begin{itemize}

\item if there are no more jobs being worked on, we do not modify the state, and
we tell the client to stop (with a {\tt Finished} message).

\item else, we simply record that job as ``unemployed''. The next time a client
reports its state, it will be asked to share its current job. This does not
change the jobs registered in the server's state anyway.
\end{itemize}

{\tt{}\small{}\setcounter{line}{188}
\noindent\stepcounter{line}\makebox[2em][l]{\theline}\hbox{\phantom{xxxxxx}{\color{RoyalBlue}if} {\color{Green}S}.null (jobs st) {\color{RoyalBlue}then}}

\noindent\stepcounter{line}\makebox[2em][l]{\theline}\hbox{\phantom{xxxxxxxx}{\color{RoyalBlue}if} {\color{Green}M}.null (ongoing st) {\color{RoyalBlue}then}}

\noindent\stepcounter{line}\makebox[2em][l]{\theline}\hbox{\phantom{xxxxxxxxxx}(st,{\color{Green}Finished})}

\noindent\stepcounter{line}\makebox[2em][l]{\theline}\hbox{\phantom{xxxxxxxx}{\color{RoyalBlue}else}}

\noindent\stepcounter{line}\makebox[2em][l]{\theline}\hbox{\phantom{xxxxxxxxxx}(st \{ unemployed={\color{Green}S}.insert num (unemployed st) \},{\color{Green}Die})}

}
Else, if there are still jobs to be done, we pick any such job (using {\tt
S.deleteFindMin}). According to the documentation of Haskell's {\tt Data.Map}
module, $\mathtt{jobs\ st}$ is equal to
$\{\mathtt{h}\}\cup\mathtt{nextJobs}$. Therefore, since {\tt num} is not a
member of {\tt ongoing st}, the returned state contains, in the union of its
{\tt ongoing} and {\tt jobs} fields, exactly the same jobs as in {\tt st}.

{\tt{}\small{}\setcounter{line}{206}
\noindent\stepcounter{line}\makebox[2em][l]{\theline}\hbox{\phantom{xxxxxx}{\color{RoyalBlue}else}}

\noindent\stepcounter{line}\makebox[2em][l]{\theline}\hbox{\phantom{xxxxxxxx}{\color{RoyalBlue}let} ((\_,h),nextJobs)={\color{Green}S}.deleteFindMin (jobs st)}

\noindent\stepcounter{line}\makebox[2em][l]{\theline}\hbox{\phantom{xxxxxxxxxxxx}shareIt=killed h$>$0}

\noindent\stepcounter{line}\makebox[2em][l]{\theline}\hbox{\phantom{xxxxxxxx}{\color{RoyalBlue}in}}

\noindent\stepcounter{line}\makebox[2em][l]{\theline}\hbox{\phantom{xxxxxxxxx}(st \{ jobs=nextJobs,}

\noindent\stepcounter{line}\makebox[2em][l]{\theline}\hbox{\phantom{xxxxxxxxxxxxxxx}unemployed={\color{Green}S}.delete num (unemployed st),}

\noindent\stepcounter{line}\makebox[2em][l]{\theline}\hbox{\phantom{xxxxxxxxxxxxxxx}ongoing={\color{Green}M}.insert num (host,key,h,t,t) (ongoing st) \},}

\noindent\stepcounter{line}\makebox[2em][l]{\theline}\hbox{\phantom{xxxxxxxxxx}{\color{Green}Job} shareIt h)}

}\hfill

Another message the server can receive is the {\tt NewJobs} message, when
clients reshare their work:
In this case, the client sends its number {\tt num}, the initial
job {\tt initialJob} it was given, the new job {\tt job} that it will now work
on, a list {\tt js} of jobs that need to be shared, and a list of results.
We can think of this message as equivalent to \emph{``I, valid client {\tt num},
hereby RSA-certify that job {\tt currentJob j} you gave me has subjobs
{\tt newJobs j}, and results {\tt results j}''}.

\hfill

{\tt{}\small{}\setcounter{line}{222}
\noindent\stepcounter{line}\makebox[2em][l]{\theline}\hbox{\phantom{}{\color{Purple}answer} t host st j@({\color{Green}NewJobs} \{\})=}

\noindent\stepcounter{line}\makebox[2em][l]{\theline}\hbox{\phantom{xx}{\color{RoyalBlue}case} {\color{Green}M}.lookup (clientId j) (ongoing st) {\color{RoyalBlue}of}}

\noindent\stepcounter{line}\makebox[2em][l]{\theline}\hbox{\phantom{xxxx}{\color{Green}Nothing}$\rightarrow$(st,{\color{Green}Die})}

\noindent\stepcounter{line}\makebox[2em][l]{\theline}\hbox{\phantom{xxxx}{\color{Green}Just} (ho,key,j0,t0,\_)$\rightarrow$}

\noindent\stepcounter{line}\makebox[2em][l]{\theline}\hbox{\phantom{xxxxxx}{\color{RoyalBlue}if} host==ho $\wedge$ j0==currentJob j {\color{RoyalBlue}then}}

\noindent\stepcounter{line}\makebox[2em][l]{\theline}\hbox{\phantom{xxxxxxxx}(st \{ jobs=foldl ($\lambda\,$s x$\rightarrow${\color{Green}S}.insert (depth x,x) s) (jobs st) (newJobs j),}

\noindent\stepcounter{line}\makebox[2em][l]{\theline}\hbox{\phantom{xxxxxxxxxxxxxx}ongoing={\color{Green}M}.insert (clientId j)}

\noindent\stepcounter{line}\makebox[2em][l]{\theline}\hbox{\phantom{xxxxxxxxxxxxxxxxxxxxxx}(host,key,nextJob j,t0,t)}

\noindent\stepcounter{line}\makebox[2em][l]{\theline}\hbox{\phantom{xxxxxxxxxxxxxxxxxxxxxx}(ongoing st),}

\noindent\stepcounter{line}\makebox[2em][l]{\theline}\hbox{\phantom{xxxxxxxxxxxxxx}results=foldl (addResult host) (results st) (jobResults j) \}, {\color{Green}Ack})}

\noindent\stepcounter{line}\makebox[2em][l]{\theline}\hbox{\phantom{xxxxxx}{\color{RoyalBlue}else}}

\noindent\stepcounter{line}\makebox[2em][l]{\theline}\hbox{\phantom{xxxxxxxx}(st, {\color{Green}Die})}

}\hfill

If the client is not registered as an ``ongoing job'', this message is ignored,
the state is not modified, and the client is sent the {\tt Die} message.

Else, we assumed that this {\tt NewJobs} message can only be sent by a valid
client.  Therefore, it contains all subjobs of its current job that have not
been explored, along with the job it will start working on, and the list of all
results that have been found during the exploration of the other subjobs of its
current job. Since all these subjobs are stored in the {\tt jobs} field of the
state, and the {\tt ongoing} field is updated with the client's new current job,
our claim still holds.

\hfill

{\tt{}\small{}\setcounter{line}{246}
\noindent\stepcounter{line}\makebox[2em][l]{\theline}\hbox{\phantom{}{\color{Purple}answer} \_ \_ st j@({\color{Green}JobDone} \{\})=}

\noindent\stepcounter{line}\makebox[2em][l]{\theline}\hbox{\phantom{xx}{\color{RoyalBlue}case} {\color{Green}M}.lookup (clientId j) (ongoing st) {\color{RoyalBlue}of}}

\noindent\stepcounter{line}\makebox[2em][l]{\theline}\hbox{\phantom{xxxx}{\color{Green}Nothing}$\rightarrow$(st,{\color{Green}Die});}

\noindent\stepcounter{line}\makebox[2em][l]{\theline}\hbox{\phantom{xxxx}{\color{Green}Just} (host,\_,j0,\_,\_)$\rightarrow$}

\noindent\stepcounter{line}\makebox[2em][l]{\theline}\hbox{\phantom{xxxxxx}{\color{RoyalBlue}if} j0==currentJob j {\color{RoyalBlue}then}}

\noindent\stepcounter{line}\makebox[2em][l]{\theline}\hbox{\phantom{xxxxxxxx}(st \{ ongoing={\color{Green}M}.delete (clientId j) (ongoing st),}

\noindent\stepcounter{line}\makebox[2em][l]{\theline}\hbox{\phantom{xxxxxxxxxxxxxx}results=foldl (addResult host) (results st) (jobResults j),}

\noindent\stepcounter{line}\makebox[2em][l]{\theline}\hbox{\phantom{xxxxxxxxxxxxxx}solved=solved st+1 \}, {\color{Green}Ack})}

\noindent\stepcounter{line}\makebox[2em][l]{\theline}\hbox{\phantom{xxxxxx}{\color{RoyalBlue}else}}

\noindent\stepcounter{line}\makebox[2em][l]{\theline}\hbox{\phantom{xxxxxxxx}(st,{\color{Green}Die})}

}
\hfill

In this case, if the client is not registered as an ongoing job, we do not
modify the state. Else, we can safely delete the corresponding job from the
state, and add its results to the state's results field: indeed, since we
assumed that this message is sent by a valid client, that job has been explored
completely. The intuitive version of this message is \emph{``I, valid client
{\tt num}, hereby RSA-certify that I have explored job {\tt currentJob j}
completely, and that it contains exactly results {\tt results j}''}.

\hfill

The last case of {\tt answer} is when the client sends an ``Alive'' message:

{\tt{}\small{}\setcounter{line}{272}
\noindent\stepcounter{line}\makebox[2em][l]{\theline}\hbox{\phantom{}{\color{Purple}answer} t host st ({\color{Green}Alive} num)=}

\noindent\stepcounter{line}\makebox[2em][l]{\theline}\hbox{\phantom{xx}{\color{RoyalBlue}case} {\color{Green}M}.lookup num (ongoing st) {\color{RoyalBlue}of}}

\noindent\stepcounter{line}\makebox[2em][l]{\theline}\hbox{\phantom{xxxx}{\color{Green}Nothing}$\rightarrow$(st,{\color{Green}Die});}

\noindent\stepcounter{line}\makebox[2em][l]{\theline}\hbox{\phantom{xxxx}{\color{Green}Just} (ho,key,j,t0,\_)$\rightarrow$}

\noindent\stepcounter{line}\makebox[2em][l]{\theline}\hbox{\phantom{xxxxxx}{\color{RoyalBlue}if} ho==host $\wedge$ ({\color{Green}S}.null (unemployed st) $\vee$ (not \$ {\color{Green}S}.null \$ jobs st)) {\color{RoyalBlue}then}}

\noindent\stepcounter{line}\makebox[2em][l]{\theline}\hbox{\phantom{xxxxxxxx}(st \{ ongoing={\color{Green}M}.insert num (ho,key,j,t0,t) (ongoing st) \},{\color{Green}Ack})}

\noindent\stepcounter{line}\makebox[2em][l]{\theline}\hbox{\phantom{xxxxxx}{\color{RoyalBlue}else}}

\noindent\stepcounter{line}\makebox[2em][l]{\theline}\hbox{\phantom{xxxxxxxx}(st,{\color{Green}Die})}

}
\hfill

In this case, the set of jobs is not modified, and hence our claim holds.
\end{proof}
\end{lemma}

Our next task is to prove {\tt reply}, the network interface to the {\tt answer}
function. We first need hypotheses on how this interface works, and especially
how the messages are written and read at the ends of the connection.

\begin{definition}

A client is \emph{fluent} if the messages it sends on the network are of exactly
two kinds:

\begin{itemize}

\item Messages with a single line containing exactly {\tt Hello}.

\item Messages with two lines:
\begin{itemize}

\item the first line is the encoding via {\tt encode16l} of $m$, where $m$ is
the encoding via {\tt encode} of a constructor of the {\tt ClientMessage} type.

\item the second line is the RSA signature, using the client's private key,
of $m$.
\end{itemize}

\end{itemize}

\end{definition}

\begin{lemma} \label{lem:reply} If all the clients that have their public key in
{\tt authorizedKeys st}, where {\tt st} is the state of the server, are valid
and fluent, and {\tt st} contains all the jobs that have not been completely
explored (in the {\tt ongoing} and {\tt jobs} fields), then so does it after one
run of {\tt reply}, assuming that $\mathtt{decode}\circ\mathtt{encode}$ (from
Haskell's {\tt Data.Binary} module) is the identity, and
$\mathtt{decode16}\circ\mathtt{encode16l}$ (from module {\tt Parry.Util}) is the
identity.

\begin{proof}
We prove this invariant on the code of the {\tt reply} function, which handles
every connection to our server.\vspace{1em}

{\tt{}\small{}\setcounter{line}{294}
\noindent\stepcounter{line}\makebox[2em][l]{\theline}\hbox{\phantom{}{\color{Purple}reply}::({\color{Green}Binary} j,{\color{Green}Exhaustive} j,{\color{Green}Result} j r,{\color{Green}Ord} j)$\Rightarrow$}

\noindent\stepcounter{line}\makebox[2em][l]{\theline}\hbox{\phantom{xxxxxxx}{\color{Green}MVar} ({\color{Green}State} j r) $\rightarrow$ {\color{Green}Handle} $\rightarrow$ {\color{Green}HostName} $\rightarrow$ {\color{Green}IO} ()}

\noindent\stepcounter{line}\makebox[2em][l]{\theline}\hbox{\phantom{}{\color{Purple}reply} state rhandle host=({\color{RoyalBlue}do}}

\noindent\stepcounter{line}\makebox[2em][l]{\theline}\hbox{\phantom{xx}l$\leftarrow${\color{Green}B}.hGetLine rhandle}

\noindent\stepcounter{line}\makebox[2em][l]{\theline}\hbox{\phantom{xx}{\color{RoyalBlue}if} l=={\color{Green}B}.pack {\color{Brown}"Hello"} {\color{RoyalBlue}then}}

\noindent\stepcounter{line}\makebox[2em][l]{\theline}\hbox{\phantom{xxxx}modifyMVar\_ state \$ $\lambda\,$st$\rightarrow${\color{RoyalBlue}do}}

\noindent\stepcounter{line}\makebox[2em][l]{\theline}\hbox{\phantom{xxxxxx}{\color{Green}LB}.hPutStrLn rhandle \$ encode16l \$ encode \$ newId st}

\noindent\stepcounter{line}\makebox[2em][l]{\theline}\hbox{\phantom{xxxxxx}return \$ st \{ newId=newId st+1 \}}

}
\hfill

When the first line is the initial {\tt Hello} message, the claim holds: indeed,
the only field of the server state that is modified is the {\tt newId} one,
which represents the first unused client number.

In all other cases, we do the following:

\hfill

{\tt{}\small{}\setcounter{line}{349}
\noindent\stepcounter{line}\makebox[2em][l]{\theline}\hbox{\phantom{xxxx}{\color{RoyalBlue}else} {\color{RoyalBlue}do}}

\noindent\stepcounter{line}\makebox[2em][l]{\theline}\hbox{\phantom{xxxx}st$\leftarrow$withMVar state return}

\noindent\stepcounter{line}\makebox[2em][l]{\theline}\hbox{\phantom{xxxx}sig$\leftarrow${\color{Green}B}.hGetLine rhandle}

\noindent\stepcounter{line}\makebox[2em][l]{\theline}\hbox{\phantom{xxxx}{\color{RoyalBlue}let} dec={\color{Green}LB}.fromStrict \$ decode16 l}

\noindent\stepcounter{line}\makebox[2em][l]{\theline}\hbox{\phantom{xxxxxxxx}msg=decode dec}

\noindent\stepcounter{line}\makebox[2em][l]{\theline}\hbox{\phantom{xxxxxxxx}num={\color{RoyalBlue}case} msg {\color{RoyalBlue}of}}

\noindent\stepcounter{line}\makebox[2em][l]{\theline}\hbox{\phantom{xxxxxxxxxx}{\color{Green}GetJob} x \_$\rightarrow$x}

\noindent\stepcounter{line}\makebox[2em][l]{\theline}\hbox{\phantom{xxxxxxxxxx}{\color{Green}JobDone} x \_ \_$\rightarrow$x}

\noindent\stepcounter{line}\makebox[2em][l]{\theline}\hbox{\phantom{xxxxxxxxxx}{\color{Green}NewJobs} x \_ \_ \_ \_$\rightarrow$x}

\noindent\stepcounter{line}\makebox[2em][l]{\theline}\hbox{\phantom{xxxxxxxxxx}{\color{Green}Alive} x$\rightarrow$x}

\noindent\stepcounter{line}\makebox[2em][l]{\theline}\hbox{\phantom{xxxxxxxx}key={\color{RoyalBlue}case} msg {\color{RoyalBlue}of}}

\noindent\stepcounter{line}\makebox[2em][l]{\theline}\hbox{\phantom{xxxxxxxxxx}{\color{Green}GetJob} \_ pub$\rightarrow$}

\noindent\stepcounter{line}\makebox[2em][l]{\theline}\hbox{\phantom{xxxxxxxxxxxx}{\color{RoyalBlue}if} any (==pub) (authorizedKeys st) {\color{RoyalBlue}then}}

\noindent\stepcounter{line}\makebox[2em][l]{\theline}\hbox{\phantom{xxxxxxxxxxxxxx}{\color{Green}Just} pub}

\noindent\stepcounter{line}\makebox[2em][l]{\theline}\hbox{\phantom{xxxxxxxxxxxx}{\color{RoyalBlue}else}}

\noindent\stepcounter{line}\makebox[2em][l]{\theline}\hbox{\phantom{xxxxxxxxxxxxxx}{\color{Green}Nothing}}

\noindent\stepcounter{line}\makebox[2em][l]{\theline}\hbox{\phantom{xxxxxxxxxx}\_$\rightarrow$({\color{RoyalBlue}case} {\color{Green}M}.lookup num \$ ongoing st {\color{RoyalBlue}of}}

\noindent\stepcounter{line}\makebox[2em][l]{\theline}\hbox{\phantom{xxxxxxxxxxxxxxxxx}{\color{Green}Just} (\_,pub,\_,\_,\_)$\rightarrow${\color{Green}Just} pub}

\noindent\stepcounter{line}\makebox[2em][l]{\theline}\hbox{\phantom{xxxxxxxxxxxxxxxxx}{\color{Green}Nothing}$\rightarrow${\color{Green}Nothing})}

}
\hfill

We will now verify the message signature, using either the public key registered
for this client in the {\tt ongoing} field of the server's state, or the public
key sent by the client itself, in the case of the {\tt GetJob} message (if that
key is registered in the {\tt authorizedKeys} field of the server state):

\hfill

{\tt{}\small{}\setcounter{line}{378}
\noindent\stepcounter{line}\makebox[2em][l]{\theline}\hbox{\phantom{xxxx}message$\leftarrow${\color{RoyalBlue}case} key {\color{RoyalBlue}of}}

\noindent\stepcounter{line}\makebox[2em][l]{\theline}\hbox{\phantom{xxxxxx}{\color{Green}Nothing}$\rightarrow$return {\color{Green}Die}}

\noindent\stepcounter{line}\makebox[2em][l]{\theline}\hbox{\phantom{xxxxxx}{\color{Green}Just} pub$\rightarrow$}

\noindent\stepcounter{line}\makebox[2em][l]{\theline}\hbox{\phantom{xxxxxxxx}{\color{RoyalBlue}case} verify pub dec ({\color{Green}LB}.fromStrict \$ decode16 sig) {\color{RoyalBlue}of}}

\noindent\stepcounter{line}\makebox[2em][l]{\theline}\hbox{\phantom{xxxxxxxxxx}{\color{Green}Right} {\color{Green}True}$\rightarrow${\color{RoyalBlue}do}}

\noindent\stepcounter{line}\makebox[2em][l]{\theline}\hbox{\phantom{xxxxxxxxxxxx}t$\leftarrow$getPOSIXTime}

\noindent\stepcounter{line}\makebox[2em][l]{\theline}\hbox{\phantom{xxxxxxxxxxxx}modifyMVar state \$ $\lambda\,$st0$\rightarrow$}

\noindent\stepcounter{line}\makebox[2em][l]{\theline}\hbox{\phantom{xxxxxxxxxxxxxx}{\color{RoyalBlue}let} (!a,!b)=answer (realToFrac t) host st0 msg {\color{RoyalBlue}in}}

\noindent\stepcounter{line}\makebox[2em][l]{\theline}\hbox{\phantom{xxxxxxxxxxxxxx}return (a,b)}

\noindent\stepcounter{line}\makebox[2em][l]{\theline}\hbox{\phantom{xxxxxxxxxx}\_$\rightarrow$return {\color{Green}Die}}

\noindent\stepcounter{line}\makebox[2em][l]{\theline}\hbox{\phantom{xxxx}{\color{Green}LB}.hPutStrLn rhandle (encode16l \$ encode message)}

\noindent\stepcounter{line}\makebox[2em][l]{\theline}\hbox{\phantom{xx})}

}
\hfill

Since we assumed that $\mathtt{decode}\circ\mathtt{encode}$ and
 $\mathtt{decode16}\circ\mathtt{encode16l}$ are both the identity function,
variable {\tt msg} contains the message sent by the client. Because the client
is valid (because its public key is in the {\tt authorizedKeys} field of the
server state), we can conclude using Lemma \ref{lem:answer} that the invariant
is maintained by the {\tt reply} function, because the only call modifying the
state is a call to {\tt answer}.

\end{proof}
\end{lemma}

The last piece of server code that we need to prove is the {\tt cleanupThread}
function, whose aim is to collect all dead machines. We do need this function,
especially on standard clusters with small walltimes compared to the task.

\begin{lemma}
\label{lem:cleanupThread}
If the {\tt kill} function, defined on jobs, does not change the task
represented by the job, and {\tt state} is a server state containing (in the
{\tt ongoing} and {\tt jobs} fields) all the jobs that have not been explored,
then so is it after one run of {\tt cleanupThread state}.

\begin{proof}

In the following function: the state is only modified by partitionning the {\tt
ongoing st} map into two maps {\tt a} and {\tt b}, and adding all the jobs of
{\tt a} to the {\tt jobs st} set, possibly calling {\tt kill} on some of
them. Therefore, the tasks represented by jobs in the union of {\tt jobs st} and
{\tt ongoing st} is not modified.

\hfill

{\tt{}\small{}\setcounter{line}{400}
\noindent\stepcounter{line}\makebox[2em][l]{\theline}\hbox{\phantom{}{\color{Purple}cleanupThread}::({\color{Green}Ord} j,{\color{Green}Exhaustive} j)$\Rightarrow${\color{Green}MVar} ({\color{Green}State} j r)$\rightarrow${\color{Green}IO} ()}

\noindent\stepcounter{line}\makebox[2em][l]{\theline}\hbox{\phantom{}{\color{Purple}cleanupThread} state={\color{RoyalBlue}do} \{}

\noindent\stepcounter{line}\makebox[2em][l]{\theline}\hbox{\phantom{xx}t\_$\leftarrow$getPOSIXTime;}

\noindent\stepcounter{line}\makebox[2em][l]{\theline}\hbox{\phantom{xx}{\color{RoyalBlue}let} \{ t=realToFrac t\_ \};}

\noindent\stepcounter{line}\makebox[2em][l]{\theline}\hbox{\phantom{xx}modifyMVar\_ state \$ $\lambda\,$st0$\rightarrow${\color{RoyalBlue}do} \{}

\noindent\stepcounter{line}\makebox[2em][l]{\theline}\hbox{\phantom{xxxx}{\color{Red}-- Find machines that have not given news for more than 10 minutes.}}

\noindent\stepcounter{line}\makebox[2em][l]{\theline}\hbox{\phantom{xxxx}{\color{RoyalBlue}let} \{ (a,b)={\color{Green}M}.partition ($\lambda\,$(\_,\_,\_,\_,t1)$\rightarrow$(t-t1) $>$ 600) (ongoing st0);}

\noindent\stepcounter{line}\makebox[2em][l]{\theline}\hbox{\phantom{xxxxxxxxxx}st=st0 \{ jobs=}

\noindent\stepcounter{line}\makebox[2em][l]{\theline}\hbox{\phantom{xxxxxxxxxxxxxxxxxxxxxx}{\color{Green}M}.foldl ($\lambda\,$set (\_,\_,job,t0,\_)$\rightarrow$}

\noindent\stepcounter{line}\makebox[2em][l]{\theline}\hbox{\phantom{xxxxxxxxxxxxxxxxxxxxxxxxxxxxxxxxx}{\color{Green}S}.insert (depth job,}

\noindent\stepcounter{line}\makebox[2em][l]{\theline}\hbox{\phantom{xxxxxxxxxxxxxxxxxxxxxxxxxxxxxxxxxxxxxxxxxxx}{\color{RoyalBlue}if} t-t0 $>$ 3600 {\color{RoyalBlue}then}}

\noindent\stepcounter{line}\makebox[2em][l]{\theline}\hbox{\phantom{xxxxxxxxxxxxxxxxxxxxxxxxxxxxxxxxxxxxxxxxxxxxx}kill job}

\noindent\stepcounter{line}\makebox[2em][l]{\theline}\hbox{\phantom{xxxxxxxxxxxxxxxxxxxxxxxxxxxxxxxxxxxxxxxxxxx}{\color{RoyalBlue}else} job) set)}

\noindent\stepcounter{line}\makebox[2em][l]{\theline}\hbox{\phantom{xxxxxxxxxxxxxxxxxxxxxx}(jobs st0)}

\noindent\stepcounter{line}\makebox[2em][l]{\theline}\hbox{\phantom{xxxxxxxxxxxxxxxxxxxxxx}a,}

\noindent\stepcounter{line}\makebox[2em][l]{\theline}\hbox{\phantom{xxxxxxxxxxxxxxxxxxx}ongoing=b \}}

\noindent\stepcounter{line}\makebox[2em][l]{\theline}\hbox{\phantom{xxxxxxxx}\};}

\noindent\stepcounter{line}\makebox[2em][l]{\theline}\hbox{\phantom{xxxx}return st}

\noindent\stepcounter{line}\makebox[2em][l]{\theline}\hbox{\phantom{xxxx}\};}

\noindent\stepcounter{line}\makebox[2em][l]{\theline}\hbox{\phantom{xx}{\color{Red}-- Sleep 30 seconds, and clean again.}}

\noindent\stepcounter{line}\makebox[2em][l]{\theline}\hbox{\phantom{xx}threadDelay 30000000;}

\noindent\stepcounter{line}\makebox[2em][l]{\theline}\hbox{\phantom{xx}cleanupThread state}

\noindent\stepcounter{line}\makebox[2em][l]{\theline}\hbox{\phantom{xx}\}}

}\end{proof}
\end{lemma}

Finally, the entry point to our server library is the {\tt server} function:

\begin{lemma}
\label{lem:server}
If:
\begin{itemize}

\item all tasks that have not been completely explored have job representants
in the {\tt ongoing} and {\tt jobs} fields of the {\tt state} argument to {\tt server},

\item all clients that sign their messages with a private RSA key whose corresponding
public key is in the {\tt state} variable are valid and fluent,

\item $\mathtt{decode}\circ\mathtt{encode}$ and
 $\mathtt{decode16}\circ\mathtt{encode16l}$ are both the identity function,
\end{itemize}

then after any number of messages received by the server, variable {\tt state}
also contains jobs representing tasks that have not been completely explored,
in the union of its {\tt ongoing} and {\tt jobs} fields.

\begin{proof}

Clearly, everything {\tt server} does is calling functions that maintain this
invariant, by Lemmas \ref{lem:reply} for {\tt reply} and \ref{lem:cleanupThread}
for {\tt cleanupThread}.

{\tt{}\small{}\setcounter{line}{457}
\noindent\stepcounter{line}\makebox[2em][l]{\theline}\hbox{\phantom{}{\color{Red}-- | Starts the synchronization server.}}

\noindent\stepcounter{line}\makebox[2em][l]{\theline}\hbox{\phantom{}{\color{Purple}server}::({\color{Green}Ord} j, {\color{Green}Binary} j, {\color{Green}Exhaustive} j, {\color{Green}Result} j r)$\Rightarrow$}

\noindent\stepcounter{line}\makebox[2em][l]{\theline}\hbox{\phantom{xxxxxxxx}{\color{Green}Config}$\rightarrow${\color{Green}MVar} ({\color{Green}State} j r)$\rightarrow${\color{Green}IO} ()}

\noindent\stepcounter{line}\makebox[2em][l]{\theline}\hbox{\phantom{}{\color{Purple}server} config state=withSocketsDo \$ {\color{RoyalBlue}do} \{}

\noindent\stepcounter{line}\makebox[2em][l]{\theline}\hbox{\phantom{}{\color{Purple}\#}ifdef {\color{Green}UNIX}}

\noindent\stepcounter{line}\makebox[2em][l]{\theline}\hbox{\phantom{xx}installHandler sigPIPE {\color{Green}Ignore} {\color{Green}Nothing};}

\noindent\stepcounter{line}\makebox[2em][l]{\theline}\hbox{\phantom{}{\color{Purple}\#}endif}

\noindent\stepcounter{line}\makebox[2em][l]{\theline}\hbox{\phantom{xx}threads$\leftarrow${\color{Green}Sem}.new \$ maxThreads config;}

\noindent\stepcounter{line}\makebox[2em][l]{\theline}\hbox{\phantom{xx}\_$\leftarrow$forkIO \$ cleanupThread state;}

\noindent\stepcounter{line}\makebox[2em][l]{\theline}\hbox{\phantom{xx}forever \$ {\color{RoyalBlue}do} \{}

\noindent\stepcounter{line}\makebox[2em][l]{\theline}\hbox{\phantom{xxxx}{\color{Green}E}.catch (bracket (listenOn (port config)) sClose \$}

\noindent\stepcounter{line}\makebox[2em][l]{\theline}\hbox{\phantom{xxxxxxxxxxxxx}$\lambda\,$sock$\rightarrow$forever \$ {\color{RoyalBlue}do} \{}

\noindent\stepcounter{line}\makebox[2em][l]{\theline}\hbox{\phantom{xxxxxxxxxxxxxxx}bracket ({\color{RoyalBlue}do} \{ (s,a,\_)$\leftarrow$accept sock; wait threads; return (s,a) \})}

\noindent\stepcounter{line}\makebox[2em][l]{\theline}\hbox{\phantom{xxxxxxxxxxxxxxx}($\lambda\,$(s,\_)$\rightarrow${\color{RoyalBlue}do} \{ signal threads; hClose s\})}

\noindent\stepcounter{line}\makebox[2em][l]{\theline}\hbox{\phantom{xxxxxxxxxxxxxxx}($\lambda\,$(s,a)$\rightarrow$reply state s a)}

\noindent\stepcounter{line}\makebox[2em][l]{\theline}\hbox{\phantom{xxxxxxxxxxxxxxx}\})}

\noindent\stepcounter{line}\makebox[2em][l]{\theline}\hbox{\phantom{xxxx}($\lambda\,$e$\rightarrow${\color{RoyalBlue}let} \_=e::{\color{Green}SomeException} {\color{RoyalBlue}in} appendFile (logFile config) (show e++{\color{Brown}"\textbackslash{}n"}));}

\noindent\stepcounter{line}\makebox[2em][l]{\theline}\hbox{\phantom{xxxx}threadDelay 100000;}

\noindent\stepcounter{line}\makebox[2em][l]{\theline}\hbox{\phantom{xxxx}\};}

\noindent\stepcounter{line}\makebox[2em][l]{\theline}\hbox{\phantom{xx}\}}

}\end{proof}
\end{lemma}

\subsection{Proof of the client}
\label{subsect:client}

Finally, we include the client code, which is just an interface to the
main function described in Section \ref{implementation}. The exact code
consists of two threads. Thread 1 periodically reports its activity to
the server (so that client crashes can be detected), while thread 2
performs the actual computation.

\ignore{

{\tt{}\small{}\setcounter{line}{0}
\noindent\stepcounter{line}\makebox[2em][l]{\theline}\hbox{\phantom{}{\color{Purple}\{-\#}{\color{Green}OPTIONS} -cpp \#-\}}

\noindent\stepcounter{line}\makebox[2em][l]{\theline}\hbox{\phantom{}{\color{Purple}\{-$|$}}

\noindent\stepcounter{line}\makebox[2em][l]{\theline}\hbox{\phantom{}{\color{Purple}Module}      :  {\color{Green}Parry}.{\color{Green}Client}}

\noindent\stepcounter{line}\makebox[2em][l]{\theline}\hbox{\phantom{}{\color{Purple}Copyright}   :  (c) {\color{Green}Pierre}-{\color{Green}Étienne} {\color{Green}Meunier} 2014}

\noindent\stepcounter{line}\makebox[2em][l]{\theline}\hbox{\phantom{}{\color{Purple}License}     :  {\color{Green}GPL}-3}

\noindent\stepcounter{line}\makebox[2em][l]{\theline}\hbox{\phantom{}{\color{Purple}Maintainer}  :  pierre-etienne.meunier@lif.univ-mrs.fr}

\noindent\stepcounter{line}\makebox[2em][l]{\theline}\hbox{\phantom{}{\color{Purple}Stability}   :  experimental}

\noindent\stepcounter{line}\makebox[2em][l]{\theline}\hbox{\phantom{}{\color{Purple}Portability} :  {\color{Green}All}}

\noindent\stepcounter{line}\makebox[2em][l]{\theline}\hbox{\phantom{}}

\noindent\stepcounter{line}\makebox[2em][l]{\theline}\hbox{\phantom{}{\color{Purple}Tool} to build clients for {\color{Green}Parry}.}

\noindent\stepcounter{line}\makebox[2em][l]{\theline}\hbox{\phantom{}{\color{Purple}-\}}}

}}

We now prove the client functions. More precisely, we prove that the clients
that can be built using the functions in this module are valid and fluent,
provided that their ``worker'' function explores the space correctly. We include
the whole source code of this module for the sake of completeness.

\hfill

{\tt{}\small{}\setcounter{line}{12}
\noindent\stepcounter{line}\makebox[2em][l]{\theline}\hbox{\phantom{}{\color{RoyalBlue}module} {\color{Green}Parry}.{\color{Green}Client}(}

\noindent\stepcounter{line}\makebox[2em][l]{\theline}\hbox{\phantom{xx}{\color{Red}-- * Jobs on the client side}}

\noindent\stepcounter{line}\makebox[2em][l]{\theline}\hbox{\phantom{xx}{\color{Green}Client}(..),}

\noindent\stepcounter{line}\makebox[2em][l]{\theline}\hbox{\phantom{xx}{\color{Red}-- * Writing clients}}

\noindent\stepcounter{line}\makebox[2em][l]{\theline}\hbox{\phantom{xx}client,{\color{Green}Config}(..),defaultConfig}

\noindent\stepcounter{line}\makebox[2em][l]{\theline}\hbox{\phantom{xx}) {\color{RoyalBlue}where}}

\noindent\stepcounter{line}\makebox[2em][l]{\theline}\hbox{\phantom{}{\color{RoyalBlue}import} {\color{Green}Network}}

\noindent\stepcounter{line}\makebox[2em][l]{\theline}\hbox{\phantom{}{\color{RoyalBlue}import} {\color{Green}System}.{\color{Green}IO}}

\noindent\stepcounter{line}\makebox[2em][l]{\theline}\hbox{\phantom{}{\color{RoyalBlue}import} {\color{Green}System}.{\color{Green}Exit}}

\noindent\stepcounter{line}\makebox[2em][l]{\theline}\hbox{\phantom{}{\color{Purple}\#}ifdef {\color{Green}UNIX}}

\noindent\stepcounter{line}\makebox[2em][l]{\theline}\hbox{\phantom{}{\color{RoyalBlue}import} {\color{Green}System}.{\color{Green}Posix}.{\color{Green}Signals}}

\noindent\stepcounter{line}\makebox[2em][l]{\theline}\hbox{\phantom{}{\color{Purple}\#}endif}

\noindent\stepcounter{line}\makebox[2em][l]{\theline}\hbox{\phantom{}{\color{RoyalBlue}import} {\color{Green}Control}.{\color{Green}Concurrent}}

\noindent\stepcounter{line}\makebox[2em][l]{\theline}\hbox{\phantom{}{\color{RoyalBlue}import} {\color{Green}Control}.{\color{Green}Exception} {\color{RoyalBlue}as} {\color{Green}E}}

\noindent\stepcounter{line}\makebox[2em][l]{\theline}\hbox{\phantom{}{\color{RoyalBlue}import} {\color{Green}Data}.{\color{Green}List}}

\noindent\stepcounter{line}\makebox[2em][l]{\theline}\hbox{\phantom{}{\color{RoyalBlue}import} {\color{Green}Codec}.{\color{Green}Crypto}.{\color{Green}RSA}.{\color{Green}Pure}}

\noindent\stepcounter{line}\makebox[2em][l]{\theline}\hbox{\phantom{}{\color{RoyalBlue}import} {\color{RoyalBlue}qualified} {\color{Green}Data}.{\color{Green}ByteString}.{\color{Green}Lazy}.{\color{Green}Char8} {\color{RoyalBlue}as} {\color{Green}B8}}

\noindent\stepcounter{line}\makebox[2em][l]{\theline}\hbox{\phantom{}{\color{RoyalBlue}import} {\color{RoyalBlue}qualified} {\color{Green}Data}.{\color{Green}ByteString}.{\color{Green}Char8} {\color{RoyalBlue}as} {\color{Green}S}}

\noindent\stepcounter{line}\makebox[2em][l]{\theline}\hbox{\phantom{}{\color{RoyalBlue}import} {\color{Green}Data}.{\color{Green}Binary}}

\noindent\stepcounter{line}\makebox[2em][l]{\theline}\hbox{\phantom{}{\color{RoyalBlue}import} {\color{Green}Parry}.{\color{Green}Util}}

\noindent\stepcounter{line}\makebox[2em][l]{\theline}\hbox{\phantom{}{\color{RoyalBlue}import} {\color{Green}Parry}.{\color{Green}Protocol}}

\noindent\stepcounter{line}\makebox[2em][l]{\theline}\hbox{\phantom{}}

\noindent\stepcounter{line}\makebox[2em][l]{\theline}\hbox{\phantom{}}

\noindent\stepcounter{line}\makebox[2em][l]{\theline}\hbox{\phantom{}{\color{Red}-- | For a job to be usable by Parry's clients, it must be a member of this class.}}

\noindent\stepcounter{line}\makebox[2em][l]{\theline}\hbox{\phantom{}{\color{RoyalBlue}class} {\color{Green}Client} j {\color{RoyalBlue}where}}

\noindent\stepcounter{line}\makebox[2em][l]{\theline}\hbox{\phantom{xx}{\color{Red}-- | This function is used by the client to distinguish results from regular jobs.}}

\noindent\stepcounter{line}\makebox[2em][l]{\theline}\hbox{\phantom{xx}isResult::j$\rightarrow${\color{Green}Bool}}

}
\begin{lemma}
\label{lem:signAndSend}
The {\tt signAndSend\_} and {\tt signAndSend} functions send only one kind of
messages on the network, with two lines: the first line is an encoding
via {\tt encode16l} of $m$, the encoding via {\tt encode} of a constructor
of the {\tt ClientMessage} type.
The second line is the RSA signature of $m$ using the private key in the
{\tt conf} argument.
\begin{proof}

The results follows from the fact that this function uses only two lines to send
messages on the network, and these two lines have the claimed form.

\hfill

{\tt{}\small{}\setcounter{line}{73}
\noindent\stepcounter{line}\makebox[2em][l]{\theline}\hbox{\phantom{}{\color{Red}-- | A wrapper around signAndSend, to force the type of @j@, which is}}

\noindent\stepcounter{line}\makebox[2em][l]{\theline}\hbox{\phantom{}{\color{Red}-- needed for the protocol's @Alive@ messages.}}

\noindent\stepcounter{line}\makebox[2em][l]{\theline}\hbox{\phantom{}{\color{Purple}signAndSend\_}::({\color{Green}Binary} j)$\Rightarrow${\color{Green}Config}$\rightarrow$j$\rightarrow${\color{Green}ClientMessage} j$\rightarrow${\color{Green}IO} ({\color{Green}ServerMessage} j)}

\noindent\stepcounter{line}\makebox[2em][l]{\theline}\hbox{\phantom{}{\color{Purple}signAndSend\_} conf \_ a=signAndSend conf a}

\noindent\stepcounter{line}\makebox[2em][l]{\theline}\hbox{\phantom{}}

\noindent\stepcounter{line}\makebox[2em][l]{\theline}\hbox{\phantom{}{\color{Red}-- | Sign a message and send it.}}

\noindent\stepcounter{line}\makebox[2em][l]{\theline}\hbox{\phantom{}{\color{Purple}signAndSend}::({\color{Green}Binary} j)$\Rightarrow${\color{Green}Config}$\rightarrow${\color{Green}ClientMessage} j$\rightarrow${\color{Green}IO} ({\color{Green}ServerMessage} j)}

\noindent\stepcounter{line}\makebox[2em][l]{\theline}\hbox{\phantom{}{\color{Purple}signAndSend} conf m={\color{RoyalBlue}do} \{}

\noindent\stepcounter{line}\makebox[2em][l]{\theline}\hbox{\phantom{xx}{\color{RoyalBlue}let} \{ msg=encode m \};}

\noindent\stepcounter{line}\makebox[2em][l]{\theline}\hbox{\phantom{xx}{\color{RoyalBlue}case} sign (privateKey conf) msg {\color{RoyalBlue}of} \{}

\noindent\stepcounter{line}\makebox[2em][l]{\theline}\hbox{\phantom{xxxx}{\color{Green}Right} b$\rightarrow${\color{Green}E}.catch ({\color{RoyalBlue}do} \{}

\noindent\stepcounter{line}\makebox[2em][l]{\theline}\hbox{\phantom{xxxxxxxxxxxxxxxxxxxxxxxxx}l$\leftarrow$bracket (connectTo (server conf) (port conf)) (hClose)}

\noindent\stepcounter{line}\makebox[2em][l]{\theline}\hbox{\phantom{xxxxxxxxxxxxxxxxxxxxxxxxxxxx}($\lambda\,$h$\rightarrow${\color{RoyalBlue}do} \{}

\noindent\stepcounter{line}\makebox[2em][l]{\theline}\hbox{\phantom{xxxxxxxxxxxxxxxxxxxxxxxxxxxxxxxx}{\color{Green}B8}.hPutStrLn h \$ encode16l msg;}

\noindent\stepcounter{line}\makebox[2em][l]{\theline}\hbox{\phantom{xxxxxxxxxxxxxxxxxxxxxxxxxxxxxxxx}{\color{Green}B8}.hPutStrLn h \$ encode16l b;}

\noindent\stepcounter{line}\makebox[2em][l]{\theline}\hbox{\phantom{xxxxxxxxxxxxxxxxxxxxxxxxxxxxxxxx}hFlush h;}

\noindent\stepcounter{line}\makebox[2em][l]{\theline}\hbox{\phantom{xxxxxxxxxxxxxxxxxxxxxxxxxxxxxxxx}{\color{Green}S}.hGetLine h;}

\noindent\stepcounter{line}\makebox[2em][l]{\theline}\hbox{\phantom{xxxxxxxxxxxxxxxxxxxxxxxxxxxxxxxx}\});}

\noindent\stepcounter{line}\makebox[2em][l]{\theline}\hbox{\phantom{xxxxxxxxxxxxxxxxxxxxxxxxx}{\color{RoyalBlue}case} decodeOrFail \$ decode16l \$ {\color{Green}B8}.fromStrict l {\color{RoyalBlue}of} \{}

\noindent\stepcounter{line}\makebox[2em][l]{\theline}\hbox{\phantom{xxxxxxxxxxxxxxxxxxxxxxxxxxx}{\color{Green}Right} (\_,\_,a)$\rightarrow$return a;}

\noindent\stepcounter{line}\makebox[2em][l]{\theline}\hbox{\phantom{xxxxxxxxxxxxxxxxxxxxxxxxxxx}{\color{Green}Left} \_$\rightarrow${\color{RoyalBlue}do} \{}

\noindent\stepcounter{line}\makebox[2em][l]{\theline}\hbox{\phantom{xxxxxxxxxxxxxxxxxxxxxxxxxxxxx}threadDelay 1000000;}

\noindent\stepcounter{line}\makebox[2em][l]{\theline}\hbox{\phantom{xxxxxxxxxxxxxxxxxxxxxxxxxxxxx}signAndSend conf m}

\noindent\stepcounter{line}\makebox[2em][l]{\theline}\hbox{\phantom{xxxxxxxxxxxxxxxxxxxxxxxxxxxxx}\}}

\noindent\stepcounter{line}\makebox[2em][l]{\theline}\hbox{\phantom{xxxxxxxxxxxxxxxxxxxxxxxxxxx}\}\})}

\noindent\stepcounter{line}\makebox[2em][l]{\theline}\hbox{\phantom{xxxxxxxxxxxxx}($\lambda\,$e$\rightarrow${\color{RoyalBlue}do} \{}

\noindent\stepcounter{line}\makebox[2em][l]{\theline}\hbox{\phantom{xxxxxxxxxxxxxxxxx}{\color{RoyalBlue}let} \{ \_=e::{\color{Green}IOError} \};}

\noindent\stepcounter{line}\makebox[2em][l]{\theline}\hbox{\phantom{xxxxxxxxxxxxxxxxx}print e;}

\noindent\stepcounter{line}\makebox[2em][l]{\theline}\hbox{\phantom{xxxxxxxxxxxxxxxxx}threadDelay 1000000;signAndSend conf m}

\noindent\stepcounter{line}\makebox[2em][l]{\theline}\hbox{\phantom{xxxxxxxxxxxxxxxxx}\});}

\noindent\stepcounter{line}\makebox[2em][l]{\theline}\hbox{\phantom{xxxx}er$\rightarrow${\color{RoyalBlue}do} \{}

\noindent\stepcounter{line}\makebox[2em][l]{\theline}\hbox{\phantom{xxxxxx}print er;hFlush stdout;threadDelay 100000;signAndSend conf m}

\noindent\stepcounter{line}\makebox[2em][l]{\theline}\hbox{\phantom{xxxxxx}\}}

\noindent\stepcounter{line}\makebox[2em][l]{\theline}\hbox{\phantom{xxxx}\}\}}

}
\end{proof}
\end{lemma}

To prove the remaining functions, we need to introduce the following invariant
on their arguments:
\begin{invariant}
\label{h}
When the {\tt cur} variable is not {\tt Nothing},
the {\tt jobs} and {\tt results} variables contain,
respectively, the list of all jobs of the contents of {\tt cur} that have not
been completely explored, and the list of results found during the exploration
of all other subjobs of the job in {\tt cur}.
\end{invariant}

\begin{lemma}
\label{lem:dowork}
If there is a function {\tt doWork} such that, at the same time:
\begin{enumerate}
\item For all values of {\tt b}, {\tt save} and {\tt j}, {\tt doWork b save j}
only calls {\tt save} with arguments {\tt l} and {\tt r}, where {\tt r} is the
list of all results that have been found when the subjobs of {\tt j} that have
not been completely explored are all in list {\tt l}.

\item For all values of {\tt b}, {\tt save} and {\tt j}, {\tt doWork b save
j} returns the list of all subjobs of {\tt j} that have not been explored, and
all results that have been found in {\tt j}, in the remaining subjobs of {\tt j}.

\end{enumerate}
Then {\tt client conf doWork} is a valid and fluent client.
\begin{proof}

We first prove fluency: in the two functions below ({\tt work} and {\tt
client}), the only messages sent to the network are either sent using {\tt
signAndSend}, or else consist of a single line containing exactly {\tt Hello},
which is the definition of fluency (by Lemma \ref{lem:signAndSend}).

We now prove that {\tt client conf doWork} is valid. Indeed, condition
\ref{valid:newjobs} of validity ({\tt NewJobs} messages contain all subjobs of
the current job) is clearly respected by the {\tt work} function: indeed, by the
definition of a result, the {\tt todo} variable in {\tt work} is the list of all
subjobs that have not been explored. Condition \ref{valid:jobdone} is also
respected, because it is not applicable to the {\tt work} function.

Moreover, the {\tt work} function returns either {\tt Nothing}, or {\tt Just
(j,r)} where {\tt j} is a job, and {\tt r} is the list of all results found
during the exploration of {\tt j}. We prove this recursively on its code:

\hfill

{\tt{}\small{}\setcounter{line}{123}
\noindent\stepcounter{line}\makebox[2em][l]{\theline}\hbox{\phantom{}{\color{Purple}work}::({\color{Green}Client} j,{\color{Green}Binary} j)$\Rightarrow${\color{Green}Config}$\rightarrow${\color{Green}MVar} [j]$\rightarrow${\color{Green}MVar}[j]$\rightarrow${\color{Green}MVar} ({\color{Green}Maybe} j)$\rightarrow$}

\noindent\stepcounter{line}\makebox[2em][l]{\theline}\hbox{\phantom{xxxxxx}({\color{Green}Bool}$\rightarrow$([j]$\rightarrow$[j]$\rightarrow${\color{Green}IO} ())$\rightarrow$j$\rightarrow${\color{Green}IO} [j])$\rightarrow${\color{Green}Bool}$\rightarrow${\color{Green}Integer}$\rightarrow$j$\rightarrow${\color{Green}IO} ({\color{Green}Maybe} (j,[j]))}

\noindent\stepcounter{line}\makebox[2em][l]{\theline}\hbox{\phantom{}{\color{Purple}work} conf jobs results cur doWork shared num j={\color{RoyalBlue}do}}

\noindent\stepcounter{line}\makebox[2em][l]{\theline}\hbox{\phantom{xx}{\color{RoyalBlue}let} save jobs\_ results\_ jobs results={\color{RoyalBlue}do}}

\noindent\stepcounter{line}\makebox[2em][l]{\theline}\hbox{\phantom{xxxxxxxx}modifyMVar\_ jobs\_ \$ $\lambda\,$\_$\rightarrow$return jobs}

\noindent\stepcounter{line}\makebox[2em][l]{\theline}\hbox{\phantom{xxxxxxxx}modifyMVar\_ results\_ \$ $\lambda\,$\_$\rightarrow$return results}

\noindent\stepcounter{line}\makebox[2em][l]{\theline}\hbox{\phantom{xx}someWork$\leftarrow$doWork shared (save jobs results) j}

\noindent\stepcounter{line}\makebox[2em][l]{\theline}\hbox{\phantom{xx}{\color{RoyalBlue}let} (result,todo)=partition isResult someWork}

\noindent\stepcounter{line}\makebox[2em][l]{\theline}\hbox{\phantom{xx}{\color{RoyalBlue}case} todo {\color{RoyalBlue}of}}

}\hfill

Either {\tt doWork} has returned no new jobs, in which case the induction
hypothesis clearly holds, because the exploration of the current job is over.

\hfill

{\tt{}\small{}\setcounter{line}{182}
\noindent\stepcounter{line}\makebox[2em][l]{\theline}\hbox{\phantom{xxxx}[]$\rightarrow$return ({\color{Green}Just} (j,result))}

}\hfill

Or {\tt doWork} returns some new jobs. In this case, {\tt work} is called on
{\tt u} only after reply {\tt Ack} has been received from the server,
acknowledging that the current job registered for this client has been updated.
Therefore, {\tt work} is called recursively on the current job of this client,
and hence, the claim also holds, by recursion.

\hfill

{\tt{}\small{}\setcounter{line}{189}
\noindent\stepcounter{line}\makebox[2em][l]{\theline}\hbox{\phantom{xxxx}u:v$\rightarrow${\color{RoyalBlue}do}}

\noindent\stepcounter{line}\makebox[2em][l]{\theline}\hbox{\phantom{xxxxxx}x$\leftarrow$signAndSend conf ({\color{Green}NewJobs} \{clientId=num,jobResults=result,}

\noindent\stepcounter{line}\makebox[2em][l]{\theline}\hbox{\phantom{xxxxxxxxxxxxxxxxxxxxxxxxxxxxxxxxxxxx}currentJob=j,nextJob=u,newJobs=v \})}

\noindent\stepcounter{line}\makebox[2em][l]{\theline}\hbox{\phantom{xxxxxx}modifyMVar\_ cur ($\lambda\,$\_$\rightarrow${\color{RoyalBlue}do} \{}

\noindent\stepcounter{line}\makebox[2em][l]{\theline}\hbox{\phantom{xxxxxxxxxxxxxxxxxxxxxxxxxx}modifyMVar\_ results ($\lambda\,$\_$\rightarrow$return []);}

\noindent\stepcounter{line}\makebox[2em][l]{\theline}\hbox{\phantom{xxxxxxxxxxxxxxxxxxxxxxxxxx}modifyMVar\_ jobs ($\lambda\,$\_$\rightarrow$return [u]);}

\noindent\stepcounter{line}\makebox[2em][l]{\theline}\hbox{\phantom{xxxxxxxxxxxxxxxxxxxxxxxxxx}return ({\color{Green}Just} u)}

\noindent\stepcounter{line}\makebox[2em][l]{\theline}\hbox{\phantom{xxxxxxxxxxxxxxxxxxxxxxxxxx}\})}

\noindent\stepcounter{line}\makebox[2em][l]{\theline}\hbox{\phantom{xxxxxx}{\color{RoyalBlue}case} x {\color{RoyalBlue}of}}

\noindent\stepcounter{line}\makebox[2em][l]{\theline}\hbox{\phantom{xxxxxxxx}{\color{Green}Ack}$\rightarrow$work conf jobs results cur doWork (shared $\wedge$ (null v)) num u;}

\noindent\stepcounter{line}\makebox[2em][l]{\theline}\hbox{\phantom{xxxxxxxx}\_$\rightarrow$return {\color{Green}Nothing}}

}\hfill

We will now prove validity for the code of the {\tt client} function.
We have assumed that the {\tt save} function above is only called on
a list {\tt l} of subjobs, and a list {\tt r} of results, such that
{\tt l} contains all subjobs of the initial job {\tt j} that have
not been completely explored when results {\tt r} are found.

Therefore, in the {\tt work} function above, Invariant \ref{h} on
the {\tt jobs\_} and {\tt results\_} variables is clearly maintained:

\hfill
\ignore{

{\tt{}\small{}\setcounter{line}{209}
\noindent\stepcounter{line}\makebox[2em][l]{\theline}\hbox{\phantom{}{\color{Red}-- | @'client' config work@ calls its argument function @work@ on a}}

\noindent\stepcounter{line}\makebox[2em][l]{\theline}\hbox{\phantom{}{\color{Red}-- boolean @s@ (if the server asked to share jobs) and the actual job}}

\noindent\stepcounter{line}\makebox[2em][l]{\theline}\hbox{\phantom{}{\color{Red}-- @j@ that must be done. @work@ must return a list of resulting jobs,}}

\noindent\stepcounter{line}\makebox[2em][l]{\theline}\hbox{\phantom{}{\color{Red}-- that may include results (see 'Parry.Protocol.Client').}}

\noindent\stepcounter{line}\makebox[2em][l]{\theline}\hbox{\phantom{}{\color{Red}--}}

\noindent\stepcounter{line}\makebox[2em][l]{\theline}\hbox{\phantom{}{\color{Red}-- Workers asked to share their input job should share it as early as}}

\noindent\stepcounter{line}\makebox[2em][l]{\theline}\hbox{\phantom{}{\color{Red}-- possible: this usually means that the job already got killed in a}}

\noindent\stepcounter{line}\makebox[2em][l]{\theline}\hbox{\phantom{}{\color{Red}-- previous attempt, probably because of its length.}}

}}

{\tt{}\small{}\setcounter{line}{230}
\noindent\stepcounter{line}\makebox[2em][l]{\theline}\hbox{\phantom{}{\color{Purple}client}::({\color{Green}Binary} j,{\color{Green}Client} j)$\Rightarrow$}

\noindent\stepcounter{line}\makebox[2em][l]{\theline}\hbox{\phantom{xxxxxxxx}{\color{Green}Config}$\rightarrow$({\color{Green}Bool}$\rightarrow$([j]$\rightarrow$[j]$\rightarrow${\color{Green}IO}())$\rightarrow$j$\rightarrow${\color{Green}IO} [j])$\rightarrow${\color{Green}IO} ()}

\noindent\stepcounter{line}\makebox[2em][l]{\theline}\hbox{\phantom{}{\color{Purple}client} conf doWork=}

\noindent\stepcounter{line}\makebox[2em][l]{\theline}\hbox{\phantom{xx}{\color{RoyalBlue}let} startConnection={\color{RoyalBlue}do}}

\noindent\stepcounter{line}\makebox[2em][l]{\theline}\hbox{\phantom{xxxxxxxx}hPutStrLn stderr {\color{Brown}"Connecting..."}}

\noindent\stepcounter{line}\makebox[2em][l]{\theline}\hbox{\phantom{xxxxxxxx}hFlush stderr}

\noindent\stepcounter{line}\makebox[2em][l]{\theline}\hbox{\phantom{xxxxxxxx}l$\leftarrow${\color{Green}E}.catch (bracket (connectTo (server conf) (port conf)) hClose}

\noindent\stepcounter{line}\makebox[2em][l]{\theline}\hbox{\phantom{xxxxxxxxxxxxxxxxxxxx}($\lambda\,$h$\rightarrow${\color{RoyalBlue}do} \{}

\noindent\stepcounter{line}\makebox[2em][l]{\theline}\hbox{\phantom{xxxxxxxxxxxxxxxxxxxxxxxx}{\color{Green}B8}.hPutStrLn h ({\color{Green}B8}.pack {\color{Brown}"Hello"});}

\noindent\stepcounter{line}\makebox[2em][l]{\theline}\hbox{\phantom{xxxxxxxxxxxxxxxxxxxxxxxx}hFlush h;}

\noindent\stepcounter{line}\makebox[2em][l]{\theline}\hbox{\phantom{xxxxxxxxxxxxxxxxxxxxxxxx}l$\leftarrow${\color{Green}S}.hGetLine h;}

\noindent\stepcounter{line}\makebox[2em][l]{\theline}\hbox{\phantom{xxxxxxxxxxxxxxxxxxxxxxxx}{\color{RoyalBlue}case} decodeOrFail (decode16l ({\color{Green}B8}.fromStrict l)) {\color{RoyalBlue}of} \{}

\noindent\stepcounter{line}\makebox[2em][l]{\theline}\hbox{\phantom{xxxxxxxxxxxxxxxxxxxxxxxxxx}{\color{Green}Right} x$\rightarrow$return ({\color{Green}Right} x);}

\noindent\stepcounter{line}\makebox[2em][l]{\theline}\hbox{\phantom{xxxxxxxxxxxxxxxxxxxxxxxxxx}{\color{Green}Left} \_$\rightarrow$return ({\color{Green}Left} ())}

\noindent\stepcounter{line}\makebox[2em][l]{\theline}\hbox{\phantom{xxxxxxxxxxxxxxxxxxxxxxxxxx}\}\}))}

\noindent\stepcounter{line}\makebox[2em][l]{\theline}\hbox{\phantom{xxxxxxxxxxx}($\lambda\,$e$\rightarrow${\color{RoyalBlue}let} \_=e::{\color{Green}IOError} {\color{RoyalBlue}in} {\color{RoyalBlue}do} \{ hPutStrLn stderr \$ show e; return ({\color{Green}Left} ()) \})}

\noindent\stepcounter{line}\makebox[2em][l]{\theline}\hbox{\phantom{xxxxxxxx}{\color{RoyalBlue}case} l {\color{RoyalBlue}of}}

\noindent\stepcounter{line}\makebox[2em][l]{\theline}\hbox{\phantom{xxxxxxxxxx}{\color{Green}Right} (\_,\_,num)$\rightarrow${\color{RoyalBlue}do}}

\noindent\stepcounter{line}\makebox[2em][l]{\theline}\hbox{\phantom{xxxxxxxxxxxx}jobs$\leftarrow$newMVar []}

\noindent\stepcounter{line}\makebox[2em][l]{\theline}\hbox{\phantom{xxxxxxxxxxxx}results$\leftarrow$newMVar []}

\noindent\stepcounter{line}\makebox[2em][l]{\theline}\hbox{\phantom{xxxxxxxxxxxx}cur$\leftarrow$newMVar {\color{Green}Nothing}}

}\hfill

At this point, {\tt cur} contains {\tt Nothing}: Invariant \ref{h} is clearly
maintained.

Now remark that, when the {\tt saveAll} function below is called, and the {\tt
cur} variable contains {\tt Nothing}, no message is sent. Therefore, if
Invariant \ref{h} holds when {\tt saveAll} is called, the corresponding {\tt NewJobs}
message respects the validity condition.

\hfill

{\tt{}\small{}\setcounter{line}{252}
\noindent\stepcounter{line}\makebox[2em][l]{\theline}\hbox{\phantom{xxxxxxxxxxxx}{\color{RoyalBlue}let} saveAll=withMVar jobs \$ $\lambda\,$j$\rightarrow$}

\noindent\stepcounter{line}\makebox[2em][l]{\theline}\hbox{\phantom{xxxxxxxxxxxxxxxxxx}withMVar cur \$ $\lambda\,$curj$\rightarrow$}

\noindent\stepcounter{line}\makebox[2em][l]{\theline}\hbox{\phantom{xxxxxxxxxxxxxxxxxx}{\color{RoyalBlue}case} (j,curj) {\color{RoyalBlue}of} \{}

\noindent\stepcounter{line}\makebox[2em][l]{\theline}\hbox{\phantom{xxxxxxxxxxxxxxxxxxxx}(h:s,{\color{Green}Just} cu)$\rightarrow$}

\noindent\stepcounter{line}\makebox[2em][l]{\theline}\hbox{\phantom{xxxxxxxxxxxxxxxxxxxxxxx}withMVar results \$ $\lambda\,$res$\rightarrow${\color{RoyalBlue}do} \{}

\noindent\stepcounter{line}\makebox[2em][l]{\theline}\hbox{\phantom{xxxxxxxxxxxxxxxxxxxxxxxxx}\_$\leftarrow$signAndSend conf (}

\noindent\stepcounter{line}\makebox[2em][l]{\theline}\hbox{\phantom{xxxxxxxxxxxxxxxxxxxxxxxxxxxx}{\color{Green}NewJobs} \{ clientId=num,}

\noindent\stepcounter{line}\makebox[2em][l]{\theline}\hbox{\phantom{xxxxxxxxxxxxxxxxxxxxxxxxxxxxxxxxxxxxxx}jobResults=res,}

\noindent\stepcounter{line}\makebox[2em][l]{\theline}\hbox{\phantom{xxxxxxxxxxxxxxxxxxxxxxxxxxxxxxxxxxxxxx}currentJob=cu,}

\noindent\stepcounter{line}\makebox[2em][l]{\theline}\hbox{\phantom{xxxxxxxxxxxxxxxxxxxxxxxxxxxxxxxxxxxxxx}nextJob=h,}

\noindent\stepcounter{line}\makebox[2em][l]{\theline}\hbox{\phantom{xxxxxxxxxxxxxxxxxxxxxxxxxxxxxxxxxxxxxx}newJobs=s}

\noindent\stepcounter{line}\makebox[2em][l]{\theline}\hbox{\phantom{xxxxxxxxxxxxxxxxxxxxxxxxxxxxxxxxxxxx}\});}

\noindent\stepcounter{line}\makebox[2em][l]{\theline}\hbox{\phantom{xxxxxxxxxxxxxxxxxxxxxxxxx}return ()\};}

\noindent\stepcounter{line}\makebox[2em][l]{\theline}\hbox{\phantom{xxxxxxxxxxxxxxxxxxxxxxx}\_$\rightarrow$return ()}

\noindent\stepcounter{line}\makebox[2em][l]{\theline}\hbox{\phantom{xxxxxxxxxxxxxxxxxxxx}\}}

\noindent\stepcounter{line}\makebox[2em][l]{\theline}\hbox{\phantom{xxxxxxxxxxxx}myth$\leftarrow$myThreadId}

\noindent\stepcounter{line}\makebox[2em][l]{\theline}\hbox{\phantom{xxxxxxxxxxxx}threads$\leftarrow$newMVar myth}

\noindent\stepcounter{line}\makebox[2em][l]{\theline}\hbox{\phantom{}{\color{Purple}\#}ifdef {\color{Green}UNIX}}

\noindent\stepcounter{line}\makebox[2em][l]{\theline}\hbox{\phantom{xxxxxxxxxxxx}installHandler sigTERM}

\noindent\stepcounter{line}\makebox[2em][l]{\theline}\hbox{\phantom{xxxxxxxxxxxxxx}({\color{Green}Catch} (modifyMVar\_ threads ($\lambda\,$a$\rightarrow${\color{RoyalBlue}do} \{{\color{RoyalBlue}if} a==myth {\color{RoyalBlue}then} return () {\color{RoyalBlue}else}}

\noindent\stepcounter{line}\makebox[2em][l]{\theline}\hbox{\phantom{xxxxxxxxxxxxxxxxxxxxxxxxxxxxxxxxxxxxxxxxxxxxxxxxxxxxx}killThread a;return myth\}))) {\color{Green}Nothing}}

\noindent\stepcounter{line}\makebox[2em][l]{\theline}\hbox{\phantom{}{\color{Purple}\#}endif}

\noindent\stepcounter{line}\makebox[2em][l]{\theline}\hbox{\phantom{xxxxxxxxxxxx}{\color{RoyalBlue}let} getAJob={\color{RoyalBlue}do}}

\noindent\stepcounter{line}\makebox[2em][l]{\theline}\hbox{\phantom{xxxxxxxxxxxxxxxxxx}job\_$\leftarrow$signAndSend conf ({\color{Green}GetJob} num (publicKey conf));}

\noindent\stepcounter{line}\makebox[2em][l]{\theline}\hbox{\phantom{xxxxxxxxxxxxxxxxxx}{\color{RoyalBlue}case} job\_ {\color{RoyalBlue}of}}

\noindent\stepcounter{line}\makebox[2em][l]{\theline}\hbox{\phantom{xxxxxxxxxxxxxxxxxxxx}{\color{Green}Job} share j$\rightarrow${\color{RoyalBlue}do}}

\noindent\stepcounter{line}\makebox[2em][l]{\theline}\hbox{\phantom{xxxxxxxxxxxxxxxxxxxxxx}modifyMVar\_ cur ($\lambda\,$\_$\rightarrow${\color{RoyalBlue}do} \{}

\noindent\stepcounter{line}\makebox[2em][l]{\theline}\hbox{\phantom{xxxxxxxxxxxxxxxxxxxxxxxxxxxxxxxxxxxxxxxxxx}modifyMVar\_ jobs ($\lambda\,$\_$\rightarrow$return [j]);}

\noindent\stepcounter{line}\makebox[2em][l]{\theline}\hbox{\phantom{xxxxxxxxxxxxxxxxxxxxxxxxxxxxxxxxxxxxxxxxxx}modifyMVar\_ results ($\lambda\,$\_$\rightarrow$return []);}

\noindent\stepcounter{line}\makebox[2em][l]{\theline}\hbox{\phantom{xxxxxxxxxxxxxxxxxxxxxxxxxxxxxxxxxxxxxxxxxx}return ({\color{Green}Just} j)}

\noindent\stepcounter{line}\makebox[2em][l]{\theline}\hbox{\phantom{xxxxxxxxxxxxxxxxxxxxxxxxxxxxxxxxxxxxxxxxxx}\});}

\noindent\stepcounter{line}\makebox[2em][l]{\theline}\hbox{\phantom{xxxxxxxxxxxxxxxxxxxxxx}workerMVar$\leftarrow$newEmptyMVar;}

}\hfill

At this point, Invariant \ref{h} still holds: the only subjob of the current job
is itself, and no results have been found.

\hfill

{\tt{}\small{}\setcounter{line}{295}
\noindent\stepcounter{line}\makebox[2em][l]{\theline}\hbox{\phantom{xxxxxxxxxxxxxxxxxxxxxx}workerThread$\leftarrow$}

\noindent\stepcounter{line}\makebox[2em][l]{\theline}\hbox{\phantom{xxxxxxxxxxxxxxxxxxxxxxxx}forkFinally}

\noindent\stepcounter{line}\makebox[2em][l]{\theline}\hbox{\phantom{xxxxxxxxxxxxxxxxxxxxxxxx}({\color{RoyalBlue}do} \{}

\noindent\stepcounter{line}\makebox[2em][l]{\theline}\hbox{\phantom{xxxxxxxxxxxxxxxxxxxxxxxxxxxx}x$\leftarrow$work conf jobs results cur doWork share num j;}

\noindent\stepcounter{line}\makebox[2em][l]{\theline}\hbox{\phantom{xxxxxxxxxxxxxxxxxxxxxxxxxxxx}{\color{RoyalBlue}case} x {\color{RoyalBlue}of} \{}

\noindent\stepcounter{line}\makebox[2em][l]{\theline}\hbox{\phantom{xxxxxxxxxxxxxxxxxxxxxxxxxxxxxx}{\color{Green}Just} (j,r)$\rightarrow${\color{RoyalBlue}do} \{}

\noindent\stepcounter{line}\makebox[2em][l]{\theline}\hbox{\phantom{xxxxxxxxxxxxxxxxxxxxxxxxxxxxxxxxx}\_$\leftarrow$signAndSend conf}

\noindent\stepcounter{line}\makebox[2em][l]{\theline}\hbox{\phantom{xxxxxxxxxxxxxxxxxxxxxxxxxxxxxxxxxxxx}({\color{Green}JobDone} \{clientId=num,}

\noindent\stepcounter{line}\makebox[2em][l]{\theline}\hbox{\phantom{xxxxxxxxxxxxxxxxxxxxxxxxxxxxxxxxxxxxxxxxxxxxxx}currentJob=j,}

\noindent\stepcounter{line}\makebox[2em][l]{\theline}\hbox{\phantom{xxxxxxxxxxxxxxxxxxxxxxxxxxxxxxxxxxxxxxxxxxxxxx}jobResults=r \});}

\noindent\stepcounter{line}\makebox[2em][l]{\theline}\hbox{\phantom{xxxxxxxxxxxxxxxxxxxxxxxxxxxxxxxxx}return () \};}

\noindent\stepcounter{line}\makebox[2em][l]{\theline}\hbox{\phantom{xxxxxxxxxxxxxxxxxxxxxxxxxxxxxx}{\color{Green}Nothing}$\rightarrow$return ()}

\noindent\stepcounter{line}\makebox[2em][l]{\theline}\hbox{\phantom{xxxxxxxxxxxxxxxxxxxxxxxxxxxxxx}\}\})}

\noindent\stepcounter{line}\makebox[2em][l]{\theline}\hbox{\phantom{xxxxxxxxxxxxxxxxxxxxxxxx}($\lambda\,$\_$\rightarrow$putMVar workerMVar ())}

}\hfill

Here, the validity condition holds for the {\tt JobDone} message: indeed, the
current job has been completely explored, and found exactly the results in {\tt
r}.

\hfill

{\tt{}\small{}\setcounter{line}{315}
\noindent\stepcounter{line}\makebox[2em][l]{\theline}\hbox{\phantom{xxxxxxxxxxxxxxxxxxxxxx}{\color{RoyalBlue}let} heartbeat={\color{RoyalBlue}do}}

\noindent\stepcounter{line}\makebox[2em][l]{\theline}\hbox{\phantom{xxxxxxxxxxxxxxxxxxxxxxxxxxxx}answer$\leftarrow$signAndSend\_ conf j ({\color{Green}Alive} num)}

\noindent\stepcounter{line}\makebox[2em][l]{\theline}\hbox{\phantom{xxxxxxxxxxxxxxxxxxxxxxxxxxxx}{\color{RoyalBlue}case} answer {\color{RoyalBlue}of}}

\noindent\stepcounter{line}\makebox[2em][l]{\theline}\hbox{\phantom{xxxxxxxxxxxxxxxxxxxxxxxxxxxxxx}{\color{Green}Ack}$\rightarrow${\color{RoyalBlue}do} \{ threadDelay 300000000; heartbeat \}}

\noindent\stepcounter{line}\makebox[2em][l]{\theline}\hbox{\phantom{xxxxxxxxxxxxxxxxxxxxxxxxxxxxxx}\_$\rightarrow${\color{RoyalBlue}do} \{ saveAll; killThread workerThread \}}

\noindent\stepcounter{line}\makebox[2em][l]{\theline}\hbox{\phantom{xxxxxxxxxxxxxxxxxxxxxx}heartbeatMVar$\leftarrow$newEmptyMVar}

\noindent\stepcounter{line}\makebox[2em][l]{\theline}\hbox{\phantom{xxxxxxxxxxxxxxxxxxxxxx}hbThread$\leftarrow$forkIO}

\noindent\stepcounter{line}\makebox[2em][l]{\theline}\hbox{\phantom{xxxxxxxxxxxxxxxxxxxxxxxxxxxxxxxx}(heartbeat{\color{Bittersweet}`finally`}(putMVar heartbeatMVar ()))}

\noindent\stepcounter{line}\makebox[2em][l]{\theline}\hbox{\phantom{xxxxxxxxxxxxxxxxxxxxxx}modifyMVar\_ threads ($\lambda\,$\_$\rightarrow$return workerThread)}

\noindent\stepcounter{line}\makebox[2em][l]{\theline}\hbox{\phantom{xxxxxxxxxxxxxxxxxxxxxx}takeMVar workerMVar}

\noindent\stepcounter{line}\makebox[2em][l]{\theline}\hbox{\phantom{xxxxxxxxxxxxxxxxxxxxxx}killThread hbThread}

\noindent\stepcounter{line}\makebox[2em][l]{\theline}\hbox{\phantom{xxxxxxxxxxxxxxxxxxxxxx}takeMVar heartbeatMVar}

\noindent\stepcounter{line}\makebox[2em][l]{\theline}\hbox{\phantom{xxxxxxxxxxxxxxxxxxxxxx}stop$\leftarrow$withMVar threads ($\lambda\,$t$\rightarrow$return (t==myth))}

\noindent\stepcounter{line}\makebox[2em][l]{\theline}\hbox{\phantom{xxxxxxxxxxxxxxxxxxxxxx}{\color{RoyalBlue}if} stop {\color{RoyalBlue}then} {\color{RoyalBlue}do} \{ saveAll; exitSuccess \} {\color{RoyalBlue}else} getAJob}

\noindent\stepcounter{line}\makebox[2em][l]{\theline}\hbox{\phantom{xxxxxxxxxxxxxxxxxxxx}{\color{Green}Finished}$\rightarrow$return ();}

\noindent\stepcounter{line}\makebox[2em][l]{\theline}\hbox{\phantom{xxxxxxxxxxxxxxxxxxxx}\_$\rightarrow${\color{RoyalBlue}do}}

\noindent\stepcounter{line}\makebox[2em][l]{\theline}\hbox{\phantom{xxxxxxxxxxxxxxxxxxxxxx}threadDelay 10000000}

\noindent\stepcounter{line}\makebox[2em][l]{\theline}\hbox{\phantom{xxxxxxxxxxxxxxxxxxxxxx}getAJob}

\noindent\stepcounter{line}\makebox[2em][l]{\theline}\hbox{\phantom{xxxxxxxxxxxx}getAJob}

\noindent\stepcounter{line}\makebox[2em][l]{\theline}\hbox{\phantom{xxxxxxxxxx}{\color{Green}Left} \_$\rightarrow${\color{RoyalBlue}do}}

\noindent\stepcounter{line}\makebox[2em][l]{\theline}\hbox{\phantom{xxxxxxxxxxxx}threadDelay 5000000}

\noindent\stepcounter{line}\makebox[2em][l]{\theline}\hbox{\phantom{xxxxxxxxxxxx}startConnection}

\noindent\stepcounter{line}\makebox[2em][l]{\theline}\hbox{\phantom{xx}{\color{RoyalBlue}in}}

\noindent\stepcounter{line}\makebox[2em][l]{\theline}\hbox{\phantom{xxx}startConnection}

}\end{proof}
\end{lemma}

\subsection{Proof of the atomic exploration}
\label{subsect:pats}

In this section, we prove that the atomic exploration, done by function
{\tt placeTile} (in module {\tt Pats}), verifies the condition of Lemma
\ref{lem:dowork}.

We first need to define what jobs, subjobs and results are, in our case:

\begin{definition}

In our proof, a \emph{job} is a current position $c$, a directed tileset $t$
represented by a vector of integers, and a partial assembly represented by a
matrix of integers, that can be assembled without mismatches until position $c$
(in the order of Figure \ref{fig:order}) with tiles of $t$.

A \emph{subjob} $k$ of a job $j$ is a job whose partial assembly agrees with
$j$'s, on all positions before $j$'s current position. The tileset of $k$ may
be different from $j$'s.

Finally, a \emph{result} is a job whose partial assembly is total, i.e. is of
the same size as the target pattern.

\end{definition}

In Haskell, ordinary jobs are constructed using constructor {\tt J}, while
results are constructed with {\tt R}. These constructors are defined in module
{\tt Tile}, which also defines the representation of tiles.  In order to keep
the proof short, we do not detail this part of the implementation here. Tiles
are represented by integers, divided into five ``segments'' of five bits
each. Functions {\tt north}, {\tt east}, {\tt south}, {\tt west} and {\tt color}
are then defined by bitwise operation in module {\tt Tile} to retrieve each
segment.  Functions {\tt withN}, {\tt withE}, {\tt withS}, {\tt withW} and {\tt
  withC} are also defined in module {\tt Tile}, by bitwise operations, and
modify these fields (more precisely, to construct a new integer with the
corresponding field modified).

We now proceed to the last piece of our program: the ``{\tt doWork}'' function
needed by Lemma~\ref{lem:dowork}.

\hfill

\ignore{

{\tt{}\small{}\setcounter{line}{0}
\noindent\stepcounter{line}\makebox[2em][l]{\theline}\hbox{\phantom{}{\color{Purple}\{-\#}{\color{Green}LANGUAGE} {\color{Green}BangPatterns} \#-\}}

}}

{\tt{}\small{}\setcounter{line}{2}
\noindent\stepcounter{line}\makebox[2em][l]{\theline}\hbox{\phantom{}{\color{RoyalBlue}module} {\color{Green}Pats}(placeTile) {\color{RoyalBlue}where}}

\noindent\stepcounter{line}\makebox[2em][l]{\theline}\hbox{\phantom{}}

\noindent\stepcounter{line}\makebox[2em][l]{\theline}\hbox{\phantom{}{\color{RoyalBlue}import} {\color{Green}Data}.{\color{Green}Bits}}

\noindent\stepcounter{line}\makebox[2em][l]{\theline}\hbox{\phantom{}{\color{RoyalBlue}import} {\color{RoyalBlue}qualified} {\color{Green}Data}.{\color{Green}IntSet} {\color{RoyalBlue}as} {\color{Green}S}}

\noindent\stepcounter{line}\makebox[2em][l]{\theline}\hbox{\phantom{}{\color{RoyalBlue}import} {\color{RoyalBlue}qualified} {\color{Green}Data}.{\color{Green}Vector} {\color{RoyalBlue}as} {\color{Green}V}}

\noindent\stepcounter{line}\makebox[2em][l]{\theline}\hbox{\phantom{}{\color{RoyalBlue}import} {\color{RoyalBlue}qualified} {\color{Green}Data}.{\color{Green}Vector}.{\color{Green}Unboxed} {\color{RoyalBlue}as} {\color{Green}U}}

\noindent\stepcounter{line}\makebox[2em][l]{\theline}\hbox{\phantom{}{\color{RoyalBlue}import} {\color{RoyalBlue}qualified} {\color{Green}Data}.{\color{Green}Vector}.{\color{Green}Generic} {\color{RoyalBlue}as} {\color{Green}G}}

\noindent\stepcounter{line}\makebox[2em][l]{\theline}\hbox{\phantom{}{\color{RoyalBlue}import} {\color{Green}Tile}}

}\ignore{

{\tt{}\small{}\setcounter{line}{11}
\noindent\stepcounter{line}\makebox[2em][l]{\theline}\hbox{\phantom{}{\color{Purple}(!)::}{\color{Green}G}.{\color{Green}Vector} v a$\Rightarrow$ v a$\rightarrow${\color{Green}Int}$\rightarrow$a}

\noindent\stepcounter{line}\makebox[2em][l]{\theline}\hbox{\phantom{}{\color{Purple}(!)=}{\color{Green}G}.unsafeIndex}

\noindent\stepcounter{line}\makebox[2em][l]{\theline}\hbox{\phantom{}}

\noindent\stepcounter{line}\makebox[2em][l]{\theline}\hbox{\phantom{}{\color{Purple}(//)::}{\color{Green}G}.{\color{Green}Vector} v a$\Rightarrow$v a$\rightarrow$[({\color{Green}Int},a)]$\rightarrow$v a}

\noindent\stepcounter{line}\makebox[2em][l]{\theline}\hbox{\phantom{}{\color{Purple}(//)=}{\color{Green}G}.unsafeUpd}

}}

\begin{lemma}\label{lem:directed}
Given a tileset $\mathtt{tiles}$, $\mathtt{isDirected}\ \mathtt{tiles}=\mathtt{true}$ if and only if
$\mathtt{tiles}$ is a directed DTAS tileset, that is, no two tiles in $\mathtt{tiles}$ have the same input (i.e. south and west) glues.

\begin{proof}

$\mathtt{isDirected}$ works by defining and calling the recursive function
$\mathtt{isDir}$. We prove by induction on {\tt i} that {\tt isDir i s} is
{\tt True} if and only if for all
$j,k\in S_i=\{{\mathtt i},{\mathtt i+1},\ldots |{\tt tiles}|-1\}$,
$(\mathrm{west}_j,\mathrm{south}_j)\not\in{\mathtt s}$,
and $(\mathrm{west}_j,\mathrm{south}_j)\neq(\mathrm{west}_k,\mathrm{south}_k)$.

\hfill

{\tt{}\small{}\setcounter{line}{17}
\noindent\stepcounter{line}\makebox[2em][l]{\theline}\hbox{\phantom{}{\color{Purple}isDirected}::{\color{Green}U}.{\color{Green}Vector} {\color{Green}Tile}$\rightarrow${\color{Green}Bool}}

\noindent\stepcounter{line}\makebox[2em][l]{\theline}\hbox{\phantom{}{\color{Purple}isDirected} !tiles=}

\noindent\stepcounter{line}\makebox[2em][l]{\theline}\hbox{\phantom{xx}{\color{RoyalBlue}let} isDir i s=}

}In the case where $i\geq\mathtt{U.length tiles}$, $S_i=\emptyset$, and the
claim holds.

{\tt{}\small{}\setcounter{line}{36}
\noindent\stepcounter{line}\makebox[2em][l]{\theline}\hbox{\phantom{xxxxxxxx}{\color{RoyalBlue}if} i$\geq${\color{Green}U}.length tiles {\color{RoyalBlue}then} {\color{Green}True} {\color{RoyalBlue}else}}

}
Else, assuming the claim holds for $\mathtt{i+1}$: if the couple
$(\mathrm{west}_i,\mathrm{south}_i)\in\mathtt{s}$, it also holds here (and the
result is {\tt False}). Else, it also holds, because by induction,
$\mathtt{isDir\ (i+1)\ (S.insert\ g\ s)}$ is {\tt True} only if for all for all
$k\in S_i$, $(\mathrm{west}_i,\mathrm{south}_i)$ is different from
$(\mathrm{west}_k,\mathrm{south}_k)$, because $\mathtt g$ to $\mathtt s$ in
the recursive call.

{\tt{}\small{}\setcounter{line}{39}
\noindent\stepcounter{line}\makebox[2em][l]{\theline}\hbox{\phantom{xxxxxxxxxx}{\color{RoyalBlue}let} !t=tiles ! i}

\noindent\stepcounter{line}\makebox[2em][l]{\theline}\hbox{\phantom{xxxxxxxxxxxxxx}!g=(((south t){\color{Bittersweet}`shiftL`}wgl)$.|.$(west t))}

\noindent\stepcounter{line}\makebox[2em][l]{\theline}\hbox{\phantom{xxxxxxxxxx}{\color{RoyalBlue}in}}

\noindent\stepcounter{line}\makebox[2em][l]{\theline}\hbox{\phantom{xxxxxxxxxxx}(not \$ {\color{Green}S}.member g s) $\wedge$ isDir (i+1) ({\color{Green}S}.insert g s)}

\noindent\stepcounter{line}\makebox[2em][l]{\theline}\hbox{\phantom{xx}{\color{RoyalBlue}in}}

\noindent\stepcounter{line}\makebox[2em][l]{\theline}\hbox{\phantom{xxx}isDir 0 {\color{Green}S}.empty}

}\end{proof}
\end{lemma}

We now proceed to the proof of the {\tt merge} function.
Its formal specification of this function is given by the following Lemma:

\begin{lemma}\label{lem:merge}
Given a tileset $\mathtt{ts}$, a color $\mathtt{c}$, an index into the tileset $\mathtt{i}$, a pair of north/south glues $\mathtt{a0\_}$ and $\mathtt{b0\_}$, and a pair of east/west glues $\mathtt{a1\_}$ and $\mathtt{b1\_}$, the function $\mathtt{merge}$ returns a tileset in which the $\mathtt{i}$th tile of $\mathtt{ts}$ is set to color $\mathtt{c}$, and all north/south glues in $\mathtt{ts}$ equal to $\mathtt{max(a0\_,b0\_)}$ are set to $\mathtt{min(a0\_,b0\_)}$ if this value is nonnegative (and left unchanged else), and all east/west glues in $\mathtt{ts}$ equal to $\mathtt{max(a1\_,b1\_)}$ are set to $\mathtt{min(a1\_,b1\_)}$ if this value is nonnegative (and left unchanged else).
\end{lemma}

\begin{proof}

In $\mathtt{merge}$, the pair $\mathtt{(a0,b0)}$ is formed such that
$\mathtt{a0} = \mathtt{min(a0\_,b0\_)}$ and $\mathtt{b0} =
\mathtt{max(a0\_,b0\_)}$.  Similarly, the pair $\mathtt{(a1,b1)}$ is formed such
that $\mathtt{a1} = \mathtt{min(a1\_,b1\_)}$ and $\mathtt{b1} =
\mathtt{max(a1\_,b1\_)}$.

The rest of the function does exactly our claim: a new tileset vector is
generated, in which new tiles are created, that correspond to our claim.

\hfill

{\tt{}\small{}\setcounter{line}{54}
\noindent\stepcounter{line}\makebox[2em][l]{\theline}\hbox{\phantom{}{\color{Purple}merge}::{\color{Green}Int}$\rightarrow${\color{Green}Int}$\rightarrow${\color{Green}Int}$\rightarrow${\color{Green}Int}$\rightarrow${\color{Green}Int}$\rightarrow${\color{Green}Int}$\rightarrow${\color{Green}U}.{\color{Green}Vector} {\color{Green}Tile}$\rightarrow${\color{Green}U}.{\color{Green}Vector} {\color{Green}Tile}}

\noindent\stepcounter{line}\makebox[2em][l]{\theline}\hbox{\phantom{}{\color{Purple}merge} a0\_ b0\_ a1\_ b1\_ i0 c ts=}

\noindent\stepcounter{line}\makebox[2em][l]{\theline}\hbox{\phantom{xx}{\color{RoyalBlue}let} (!a0,!b0)={\color{RoyalBlue}if} a0\_$<$b0\_ {\color{RoyalBlue}then} (a0\_,b0\_) {\color{RoyalBlue}else} (b0\_,a0\_)}

\noindent\stepcounter{line}\makebox[2em][l]{\theline}\hbox{\phantom{xxxxxx}(!a1,!b1)={\color{RoyalBlue}if} a1\_$<$b1\_ {\color{RoyalBlue}then} (a1\_,b1\_) {\color{RoyalBlue}else} (b1\_,a1\_)}

\noindent\stepcounter{line}\makebox[2em][l]{\theline}\hbox{\phantom{xx}{\color{RoyalBlue}in}}

\noindent\stepcounter{line}\makebox[2em][l]{\theline}\hbox{\phantom{xxx}{\color{Green}U}.generate ({\color{Green}U}.length ts)}

\noindent\stepcounter{line}\makebox[2em][l]{\theline}\hbox{\phantom{xxxxx}($\lambda\,$i$\rightarrow$}

\noindent\stepcounter{line}\makebox[2em][l]{\theline}\hbox{\phantom{xxxxxxx}{\color{RoyalBlue}let} !u={\color{RoyalBlue}if} i==i0 {\color{RoyalBlue}then} (ts ! i){\color{Bittersweet}`withC`}c {\color{RoyalBlue}else} ts ! i}

\noindent\stepcounter{line}\makebox[2em][l]{\theline}\hbox{\phantom{xxxxxxxxxxx}!v={\color{RoyalBlue}if} south u==b0 $\wedge$ a0$\geq$0 {\color{RoyalBlue}then} u{\color{Bittersweet}`withS`}a0 {\color{RoyalBlue}else} u}

\noindent\stepcounter{line}\makebox[2em][l]{\theline}\hbox{\phantom{xxxxxxxxxxx}!w={\color{RoyalBlue}if} north v==b0 $\wedge$ a0$\geq$0 {\color{RoyalBlue}then} v{\color{Bittersweet}`withN`}a0 {\color{RoyalBlue}else} v}

\noindent\stepcounter{line}\makebox[2em][l]{\theline}\hbox{\phantom{xxxxxxxxxxx}!x={\color{RoyalBlue}if} west w==b1 $\wedge$ a1$\geq$0 {\color{RoyalBlue}then} w{\color{Bittersweet}`withW`}a1 {\color{RoyalBlue}else} w}

\noindent\stepcounter{line}\makebox[2em][l]{\theline}\hbox{\phantom{xxxxxxx}{\color{RoyalBlue}in}}

\noindent\stepcounter{line}\makebox[2em][l]{\theline}\hbox{\phantom{xxxxxxxx}{\color{RoyalBlue}if} east x==b1 $\wedge$ a1$\geq$0 {\color{RoyalBlue}then} x{\color{Bittersweet}`withE`}a1 {\color{RoyalBlue}else} x}

\noindent\stepcounter{line}\makebox[2em][l]{\theline}\hbox{\phantom{xxxxx})}

}
\end{proof}

Finally, the core of our algorithm is the following recursive function, $\mathtt{placeTile}$, which moves through all locations in the pattern in the ordering shown by Figure~\ref{fig:order} and places tiles from the current tile set (often modifying the tile set, too) as long as it is able to.  By making recursive calls which attempt all possibilities, it ensures that the full set of possible tile sets (up to isomorphism) is explored and returns exactly those which self-assemble the given pattern. The arguments to $\mathtt{placeTile}$ are:

\begin{enumerate}
    \item[] {\tt share}: a boolean value set by the server, possibly telling the client to ``share'' the current job
    \item[] {\tt save}: a function that we have explained in Lemma \ref{lem:dowork}, that {\tt placeTile} can use to ``save'' intermediate results in case it is asked to reshare, or killed (for instance if it runs on a cluster).
    \item[] {\tt results}: a list of results that have been found so far.
    \item[] {\tt jobs}: a list of jobs to treat. Each job contains four relevant fields for the actual computation:
\begin{enumerate}
    \item[] {\tt posX},{\tt posY}: the coordinates of the current position in the assembly which $\mathtt{placeTile}$ should attempt to fill with a tile

    \item[] {\tt tileset}: the current tileset (which is a vector of integer values which are each the concatenated integer values representing the properties of a tile type)

    \item[] {\tt assembly}: the current assembly (which is a two-dimensional vector storing the index of the tile type, in the tileset, which is located at each pair of $(x,y)$ coordinates
\end{enumerate}
\end{enumerate}

Note that $\mathtt{placeTile}$ also makes use of the globally defined
two-dimensional vector $\mathtt{pattern}$ which, at each location representing a
pair of $(x,y)$ coordinates, defines one of two colors (i.e. 0 (black) or 1
(white)) for the pattern at that location.

\begin{lemma}\label{lem:placeTile}

For any list of jobs $\mathtt{j}_{\mathtt 0}$, any value of {\tt share} and any function {\tt
save}, {\tt placeTile share save [] $\mathtt{j}_{\mathtt 0}$} returns the list of all subjobs of jobs of
$\mathtt{j_0}$ that have not been explored, and all results that have been found during
the exploration of the explored subjobs of $\mathtt{j}_{\mathtt 0}$.

Moreover, all its calls to {\tt save} are of the form {\tt save j r}, where
{\tt j} is the list of all subjobs of $\mathtt{j}_{\mathtt 0}$ that have not been completely
explored, and $\mathtt{r}$ is the list of all results that have been found in the
exploration of all other subjobs of $\mathtt{j}_{\mathtt 0}$.

\end{lemma}



\begin{proof}

We will prove, by induction on the number of subjobs of
$\mathtt{j}_{\mathtt{0}}$, that for all values of $\mathtt{r}$ and {\tt j},
{\tt placeTile share save r j} is the concatenation of {\tt r} with all the results
found in the exploration, and all the subjobs of jobs of {\tt j} that have not
been explored.

Moreover, we will prove the following invariant:

\hfill

\begin{invariant}
\label{l}

The recursive calls of {\tt placeTile} are all such that the {\tt results} and
{\tt jobs} arguments verify the condition that {\tt jobs} is the list of all
subjobs of the initial job list that have not been explored and contains no
results, and {\tt results} is the list of all results that have been found
during the exploration of all other subjobs of the initial job list.

\end{invariant}

The first case is when the list of jobs to explore is empty. In this case,
we simply return the list of found results, and the claim holds.

\hfill

{\tt{}\small{}\setcounter{line}{91}
\noindent\stepcounter{line}\makebox[2em][l]{\theline}\hbox{\phantom{}{\color{Purple}placeTile}::{\color{Green}Bool} $\rightarrow$ ([{\color{Green}Job}] $\rightarrow$ [{\color{Green}Job}] $\rightarrow$ {\color{Green}IO} ()) $\rightarrow$ [{\color{Green}Job}] $\rightarrow$ [{\color{Green}Job}] $\rightarrow$ {\color{Green}IO} [{\color{Green}Job}]}

\noindent\stepcounter{line}\makebox[2em][l]{\theline}\hbox{\phantom{}{\color{Purple}placeTile} \_ \_ res []=return res}

}\hfill

Now, even though in our protocol, no results are ever sent to the {\tt
placeTile} function in the job list argument, we need to treat this case for our
claim to hold.

\hfill

{\tt{}\small{}\setcounter{line}{164}
\noindent\stepcounter{line}\makebox[2em][l]{\theline}\hbox{\phantom{}{\color{Purple}placeTile} !share save results (j@{\color{Green}R}\{\}:js)=}

\noindent\stepcounter{line}\makebox[2em][l]{\theline}\hbox{\phantom{xx}placeTile share save (j:results) js}

}\hfill

We now proceed to the proof of the main case.

\hfill

{\tt{}\small{}\setcounter{line}{173}
\noindent\stepcounter{line}\makebox[2em][l]{\theline}\hbox{\phantom{}{\color{Purple}placeTile} !share save results (j@{\color{Green}J}\{\}:js)={\color{RoyalBlue}do} \{}

}The following line is the only place where {\tt placeTile} calls {\tt save}.
By invariant \ref{l} on the recursive calls of {\tt placeTile}, our claim
on the calls to {\tt save} clearly holds.

{\tt{}\small{}\setcounter{line}{179}
\noindent\stepcounter{line}\makebox[2em][l]{\theline}\hbox{\phantom{xxxx}save (j:js) results;}

\noindent\stepcounter{line}\makebox[2em][l]{\theline}\hbox{\phantom{xxxx}{\color{RoyalBlue}let} \{ !x=posX j;}

\noindent\stepcounter{line}\makebox[2em][l]{\theline}\hbox{\phantom{xxxxxxxxxx}!y=posY j;}

\noindent\stepcounter{line}\makebox[2em][l]{\theline}\hbox{\phantom{xxxxxxxxxx}!tiles=tileset j;}

\noindent\stepcounter{line}\makebox[2em][l]{\theline}\hbox{\phantom{xxxxxxxxxx}!assemb=assembly j;}

}
First, the variable $\mathtt{inS}$ is the glue to the south of the location
$(x,y)$ to be tiled (i.e. its south input).  If the location to the south is
outside of the pattern or if there is no tile there, then the glue value of $-1$
is used.  In an analogous manner, the variable $\mathtt{inW}$ is the value of
the input glue to the west.

{\tt{}\small{}\setcounter{line}{187}
\noindent\stepcounter{line}\makebox[2em][l]{\theline}\hbox{\phantom{xxxxxxxxxx}!inS=}

\noindent\stepcounter{line}\makebox[2em][l]{\theline}\hbox{\phantom{xxxxxxxxxxxx}{\color{RoyalBlue}if} y$\geq$1 {\color{RoyalBlue}then}}

\noindent\stepcounter{line}\makebox[2em][l]{\theline}\hbox{\phantom{xxxxxxxxxxxxxx}{\color{RoyalBlue}let} !a=(assemb!(y-1)!x) {\color{RoyalBlue}in}}

\noindent\stepcounter{line}\makebox[2em][l]{\theline}\hbox{\phantom{xxxxxxxxxxxxxx}{\color{RoyalBlue}if} a$\geq$0 {\color{RoyalBlue}then}}

\noindent\stepcounter{line}\makebox[2em][l]{\theline}\hbox{\phantom{xxxxxxxxxxxxxxxx}north \$! tiles!a}

\noindent\stepcounter{line}\makebox[2em][l]{\theline}\hbox{\phantom{xxxxxxxxxxxxxx}{\color{RoyalBlue}else}}

\noindent\stepcounter{line}\makebox[2em][l]{\theline}\hbox{\phantom{xxxxxxxxxxxxxxxx}-1}

\noindent\stepcounter{line}\makebox[2em][l]{\theline}\hbox{\phantom{xxxxxxxxxxxx}{\color{RoyalBlue}else}}

\noindent\stepcounter{line}\makebox[2em][l]{\theline}\hbox{\phantom{xxxxxxxxxxxxxx}-1;}

\noindent\stepcounter{line}\makebox[2em][l]{\theline}\hbox{\phantom{xxxxxxxxxx}!inW=}

\noindent\stepcounter{line}\makebox[2em][l]{\theline}\hbox{\phantom{xxxxxxxxxxxx}{\color{RoyalBlue}if} x$\geq$1 {\color{RoyalBlue}then}}

\noindent\stepcounter{line}\makebox[2em][l]{\theline}\hbox{\phantom{xxxxxxxxxxxxxx}{\color{RoyalBlue}let} !a=(assemb!y!(x-1)){\color{RoyalBlue}in}}

\noindent\stepcounter{line}\makebox[2em][l]{\theline}\hbox{\phantom{xxxxxxxxxxxxxx}{\color{RoyalBlue}if} a$\geq$0 {\color{RoyalBlue}then}}

\noindent\stepcounter{line}\makebox[2em][l]{\theline}\hbox{\phantom{xxxxxxxxxxxxxxxx}east \$! tiles!a}

\noindent\stepcounter{line}\makebox[2em][l]{\theline}\hbox{\phantom{xxxxxxxxxxxxxx}{\color{RoyalBlue}else}}

\noindent\stepcounter{line}\makebox[2em][l]{\theline}\hbox{\phantom{xxxxxxxxxxxxxxxx}-1}

\noindent\stepcounter{line}\makebox[2em][l]{\theline}\hbox{\phantom{xxxxxxxxxxxx}{\color{RoyalBlue}else}}

\noindent\stepcounter{line}\makebox[2em][l]{\theline}\hbox{\phantom{xxxxxxxxxxxxxx}-1}

\noindent\stepcounter{line}\makebox[2em][l]{\theline}\hbox{\phantom{xxxxxxxx}\};}

}
If {\tt inS} (respectively {\tt inW}) is negative, and we are not on the south
(respectively west) border, then there is no more tile we can add on the current
row (respectively column). Therefore, we must start a new column (respectively a
new row). Remark that since we keep alternating between adding rows and columns,
we maintain the invariant that $\mathtt{posX}\geq\mathtt{posY}$ exactly when we
are adding a new row.
This is what the following code does. Invariant \ref{l}, on {\tt placeTile}'s recursive
calls, is clearly preserved by all the calls in this portion of the code.

{\tt{}\small{}\setcounter{line}{213}
\noindent\stepcounter{line}\makebox[2em][l]{\theline}\hbox{\phantom{xxxx}{\color{RoyalBlue}if} inS$<$0 $\wedge$ y$>$0 {\color{RoyalBlue}then}}

\noindent\stepcounter{line}\makebox[2em][l]{\theline}\hbox{\phantom{xxxxxx}{\color{RoyalBlue}if} x$<${\color{Green}U}.length (pattern ! 0) {\color{RoyalBlue}then}}

\noindent\stepcounter{line}\makebox[2em][l]{\theline}\hbox{\phantom{xxxxxxxx}placeTile share save results (j \{ posX=x,posY=0 \}:js)}

\noindent\stepcounter{line}\makebox[2em][l]{\theline}\hbox{\phantom{xxxxxx}{\color{RoyalBlue}else}}

\noindent\stepcounter{line}\makebox[2em][l]{\theline}\hbox{\phantom{xxxxxxxx}{\color{RoyalBlue}if} y+1$\geq${\color{Green}V}.length (pattern) {\color{RoyalBlue}then}}

\noindent\stepcounter{line}\makebox[2em][l]{\theline}\hbox{\phantom{xxxxxxxxxx}placeTile share save ({\color{Green}R} \{ tileset=tiles,assembly=assemb \}:results) js}

\noindent\stepcounter{line}\makebox[2em][l]{\theline}\hbox{\phantom{xxxxxxxx}{\color{RoyalBlue}else}}

\noindent\stepcounter{line}\makebox[2em][l]{\theline}\hbox{\phantom{xxxxxxxxxx}placeTile share save results (j \{ posX=0,posY=y+1 \}:js)}

\noindent\stepcounter{line}\makebox[2em][l]{\theline}\hbox{\phantom{xxxx}{\color{RoyalBlue}else}}

\noindent\stepcounter{line}\makebox[2em][l]{\theline}\hbox{\phantom{xxxxxx}{\color{RoyalBlue}if} inW$<$0 $\wedge$ x$>$0 {\color{RoyalBlue}then}}

\noindent\stepcounter{line}\makebox[2em][l]{\theline}\hbox{\phantom{xxxxxxxx}{\color{RoyalBlue}if} y$<${\color{Green}V}.length (pattern) {\color{RoyalBlue}then}}

\noindent\stepcounter{line}\makebox[2em][l]{\theline}\hbox{\phantom{xxxxxxxxxx}placeTile share save results (j \{ posX=0,posY=y \}:js)}

\noindent\stepcounter{line}\makebox[2em][l]{\theline}\hbox{\phantom{xxxxxxxx}{\color{RoyalBlue}else}}

\noindent\stepcounter{line}\makebox[2em][l]{\theline}\hbox{\phantom{xxxxxxxxxx}{\color{RoyalBlue}if} x+1$\geq${\color{Green}U}.length (pattern ! 0) {\color{RoyalBlue}then}}

\noindent\stepcounter{line}\makebox[2em][l]{\theline}\hbox{\phantom{xxxxxxxxxxxx}placeTile share save ({\color{Green}R} \{ tileset=tiles,assembly=assemb \}:results) js}

\noindent\stepcounter{line}\makebox[2em][l]{\theline}\hbox{\phantom{xxxxxxxxxx}{\color{RoyalBlue}else}}

\noindent\stepcounter{line}\makebox[2em][l]{\theline}\hbox{\phantom{xxxxxxxxxxxx}placeTile share save results (j \{ posX=x+1,posY=0 \}:js)}

\noindent\stepcounter{line}\makebox[2em][l]{\theline}\hbox{\phantom{xxxxxx}{\color{RoyalBlue}else}}

}
Else, both the south and west glues are defined, or we are at the beginning of a
row or a column. Hence, there are two possible cases: either there is already a
tile with matching south and west glues, or there is none.  In the first case,
we have no choice but to place that tile at the current position, and move on to
the next position, which is done in case {\tt Just p}:

{\tt{}\small{}\setcounter{line}{241}
\noindent\stepcounter{line}\makebox[2em][l]{\theline}\hbox{\phantom{xxxxxxxx}{\color{RoyalBlue}let} \{ (!nextX,!nextY)={\color{RoyalBlue}if} y$>$x {\color{RoyalBlue}then} (x+1,y) {\color{RoyalBlue}else} (x,y+1);}

\noindent\stepcounter{line}\makebox[2em][l]{\theline}\hbox{\phantom{xxxxxxxxxxxxxx}col=pattern!y ! x;}

\noindent\stepcounter{line}\makebox[2em][l]{\theline}\hbox{\phantom{xxxxxxxxxxxxxx}possible={\color{Green}U}.findIndex ($\lambda\,$a$\rightarrow$south a==inS $\wedge$ west a==inW) tiles \}}

\noindent\stepcounter{line}\makebox[2em][l]{\theline}\hbox{\phantom{xxxxxxxx}{\color{RoyalBlue}in}}

\noindent\stepcounter{line}\makebox[2em][l]{\theline}\hbox{\phantom{xxxxxxxxx}{\color{RoyalBlue}case} possible {\color{RoyalBlue}of} \{}

\noindent\stepcounter{line}\makebox[2em][l]{\theline}\hbox{\phantom{xxxxxxxxxxx}{\color{Green}Just} p$\rightarrow$}

\noindent\stepcounter{line}\makebox[2em][l]{\theline}\hbox{\phantom{xxxxxxxxxxxxxx}{\color{RoyalBlue}let} \{ !color\_h=color (tiles ! p) \} {\color{RoyalBlue}in}}

\noindent\stepcounter{line}\makebox[2em][l]{\theline}\hbox{\phantom{xxxxxxxxxxxxxx}{\color{RoyalBlue}if} color\_h==col $\vee$ color\_h==mgl {\color{RoyalBlue}then}}

\noindent\stepcounter{line}\makebox[2em][l]{\theline}\hbox{\phantom{xxxxxxxxxxxxxxxx}placeTile share save results}

\noindent\stepcounter{line}\makebox[2em][l]{\theline}\hbox{\phantom{xxxxxxxxxxxxxxxx}(j \{ posX=nextX,posY=nextY,}

\noindent\stepcounter{line}\makebox[2em][l]{\theline}\hbox{\phantom{xxxxxxxxxxxxxxxxxxxxx}tileset=}

\noindent\stepcounter{line}\makebox[2em][l]{\theline}\hbox{\phantom{xxxxxxxxxxxxxxxxxxxxxxx}{\color{RoyalBlue}if} color\_h==mgl {\color{RoyalBlue}then} tiles // [(p,(tiles !p){\color{Bittersweet}`withC`}col)]}

\noindent\stepcounter{line}\makebox[2em][l]{\theline}\hbox{\phantom{xxxxxxxxxxxxxxxxxxxxxxx}{\color{RoyalBlue}else} tiles,}

\noindent\stepcounter{line}\makebox[2em][l]{\theline}\hbox{\phantom{xxxxxxxxxxxxxxxxxxxxx}assembly=assemb // [(y,(assemb!y) // [(x,p)])]}

\noindent\stepcounter{line}\makebox[2em][l]{\theline}\hbox{\phantom{xxxxxxxxxxxxxxxxxxx}\}:js)}

\noindent\stepcounter{line}\makebox[2em][l]{\theline}\hbox{\phantom{xxxxxxxxxxxxxx}{\color{RoyalBlue}else}}

\noindent\stepcounter{line}\makebox[2em][l]{\theline}\hbox{\phantom{xxxxxxxxxxxxxxxx}placeTile share save results js;}

\noindent\stepcounter{line}\makebox[2em][l]{\theline}\hbox{\phantom{}}

}
Or there is no matching tile, in which case we simply try all tiles that can be placed
at the current position, which fall in either of two cases:

\begin{itemize}

\item tiles whose color matches the pattern's color at the current position.
\item tiles that have not yet been used, i.e. whose color is not yet defined.
We only need to consider one of them, because we do not consider solutions that are
equivalent by renaming. Therefore, we use the {\tt seenBlank} argument of {\tt
tryAllTiles} to keep track of whether we have already tried an uncolored tile.

\end{itemize}

In both cases, we merge these tiles' west and south glues with {\tt inW} and
{\tt inS}, respectively: according to Lemma \ref{lem:merge}, this means that we
adjust the tileset so that the chosen tile can be placed without mismatches at
the current position.

Formally, we can easily prove by induction on {\tt i}, that {\tt tryAllTiles i
False l} (respectively {\tt tryAllTiles i True l}) is the list of all subjobs of
{\tt j}, that try to place a tile of index {\tt i} or larger in the tileset, not
including (respectively including) uncolored tiles, along with the jobs of list
{\tt l}.

{\tt{}\small{}\setcounter{line}{266}
\noindent\stepcounter{line}\makebox[2em][l]{\theline}\hbox{\phantom{xxxxxxxxxxxxxx}{\color{Green}Nothing}$\rightarrow$}

\noindent\stepcounter{line}\makebox[2em][l]{\theline}\hbox{\phantom{xxxxxxxxxxxxxxxx}{\color{RoyalBlue}let} \{ tryAllTiles i seenBlanks list=}

\noindent\stepcounter{line}\makebox[2em][l]{\theline}\hbox{\phantom{xxxxxxxxxxxxxxxxxxxxxxxxx}{\color{RoyalBlue}if} i$\geq${\color{Green}U}.length tiles {\color{RoyalBlue}then}}

\noindent\stepcounter{line}\makebox[2em][l]{\theline}\hbox{\phantom{xxxxxxxxxxxxxxxxxxxxxxxxxxx}list}

\noindent\stepcounter{line}\makebox[2em][l]{\theline}\hbox{\phantom{xxxxxxxxxxxxxxxxxxxxxxxxx}{\color{RoyalBlue}else}}

\noindent\stepcounter{line}\makebox[2em][l]{\theline}\hbox{\phantom{xxxxxxxxxxxxxxxxxxxxxxxxxxx}{\color{RoyalBlue}let} !t=tiles !i {\color{RoyalBlue}in}}

\noindent\stepcounter{line}\makebox[2em][l]{\theline}\hbox{\phantom{xxxxxxxxxxxxxxxxxxxxxxxxxxx}{\color{RoyalBlue}if} color t==col $\vee$ (color t==mgl $\wedge$ not seenBlanks) {\color{RoyalBlue}then}}

\noindent\stepcounter{line}\makebox[2em][l]{\theline}\hbox{\phantom{xxxxxxxxxxxxxxxxxxxxxxxxxxxxx}{\color{RoyalBlue}let} !tiles=}

\noindent\stepcounter{line}\makebox[2em][l]{\theline}\hbox{\phantom{xxxxxxxxxxxxxxxxxxxxxxxxxxxxxxxxxxx}merge inS (south t) inW (west t) i col tiles}

\noindent\stepcounter{line}\makebox[2em][l]{\theline}\hbox{\phantom{xxxxxxxxxxxxxxxxxxxxxxxxxxxxx}{\color{RoyalBlue}in}}

\noindent\stepcounter{line}\makebox[2em][l]{\theline}\hbox{\phantom{xxxxxxxxxxxxxxxxxxxxxxxxxxxxxx}{\color{RoyalBlue}if} isDirected tiles {\color{RoyalBlue}then}}

\noindent\stepcounter{line}\makebox[2em][l]{\theline}\hbox{\phantom{xxxxxxxxxxxxxxxxxxxxxxxxxxxxxxxx}{\color{RoyalBlue}let} !next=}

\noindent\stepcounter{line}\makebox[2em][l]{\theline}\hbox{\phantom{xxxxxxxxxxxxxxxxxxxxxxxxxxxxxxxxxxxxxx}{\color{Green}J} \{ posX=nextX,posY=nextY,}

\noindent\stepcounter{line}\makebox[2em][l]{\theline}\hbox{\phantom{xxxxxxxxxxxxxxxxxxxxxxxxxxxxxxxxxxxxxxxxxx}tileset=tiles,}

\noindent\stepcounter{line}\makebox[2em][l]{\theline}\hbox{\phantom{xxxxxxxxxxxxxxxxxxxxxxxxxxxxxxxxxxxxxxxxxx}assembly=}

\noindent\stepcounter{line}\makebox[2em][l]{\theline}\hbox{\phantom{xxxxxxxxxxxxxxxxxxxxxxxxxxxxxxxxxxxxxxxxxxxx}(assemb // [(y,(assemb!y) // [(x,i)])]),}

\noindent\stepcounter{line}\makebox[2em][l]{\theline}\hbox{\phantom{xxxxxxxxxxxxxxxxxxxxxxxxxxxxxxxxxxxxxxxxxx}k=0 \}}

\noindent\stepcounter{line}\makebox[2em][l]{\theline}\hbox{\phantom{xxxxxxxxxxxxxxxxxxxxxxxxxxxxxxxx}{\color{RoyalBlue}in}}

\noindent\stepcounter{line}\makebox[2em][l]{\theline}\hbox{\phantom{xxxxxxxxxxxxxxxxxxxxxxxxxxxxxxxxx}tryAllTiles (i+1)}

\noindent\stepcounter{line}\makebox[2em][l]{\theline}\hbox{\phantom{xxxxxxxxxxxxxxxxxxxxxxxxxxxxxxxxx}(seenBlanks $\vee$ (color t==mgl))}

\noindent\stepcounter{line}\makebox[2em][l]{\theline}\hbox{\phantom{xxxxxxxxxxxxxxxxxxxxxxxxxxxxxxxxx}(next:list)}

\noindent\stepcounter{line}\makebox[2em][l]{\theline}\hbox{\phantom{xxxxxxxxxxxxxxxxxxxxxxxxxxxxxx}{\color{RoyalBlue}else}}

\noindent\stepcounter{line}\makebox[2em][l]{\theline}\hbox{\phantom{xxxxxxxxxxxxxxxxxxxxxxxxxxxxxxxx}tryAllTiles (i+1) seenBlanks list}

\noindent\stepcounter{line}\makebox[2em][l]{\theline}\hbox{\phantom{xxxxxxxxxxxxxxxxxxxxxxxxxxx}{\color{RoyalBlue}else}}

\noindent\stepcounter{line}\makebox[2em][l]{\theline}\hbox{\phantom{xxxxxxxxxxxxxxxxxxxxxxxxxxxxx}tryAllTiles (i+1) seenBlanks list;}

\noindent\stepcounter{line}\makebox[2em][l]{\theline}\hbox{\phantom{xxxxxxxxxxxxxxxxxxxxxx}nextJobs=tryAllTiles 0 {\color{Green}False} js \}}

\noindent\stepcounter{line}\makebox[2em][l]{\theline}\hbox{\phantom{xxxxxxxxxxxxxxxx}{\color{RoyalBlue}in}}

}First, if we were asked to share, return all newly created jobs, along with
{\tt js}, the remaining jobs after {\tt j} is divided into subjobs. In this
case, the claim holds: {\tt placeTile} does return all jobs and results found
during the exploration.

{\tt{}\small{}\setcounter{line}{318}
\noindent\stepcounter{line}\makebox[2em][l]{\theline}\hbox{\phantom{xxxxxxxxxxxxxxxxx}{\color{RoyalBlue}if} share {\color{RoyalBlue}then} return (results++nextJobs)}

}
From our proof of the {\tt tryAllTiles} function, the following recursive
call to {\tt placeTile} preserves the invariant on {\tt placeTile}'s recursive
calls. Indeed, {\tt results} has not changed, and {\tt nextJobs} now contains
{\tt js}, along with all subjobs of {\tt j} (up to renaming of unused tiles).
By the definitions of jobs and subtasks (see Definitions \ref{def:tasks} and
\ref{def:jobs}), invariant \ref{l} is clearly preserved.

{\tt{}\small{}\setcounter{line}{323}
\noindent\stepcounter{line}\makebox[2em][l]{\theline}\hbox{\phantom{xxxxxxxxxxxxxxxxx}{\color{RoyalBlue}else} placeTile share save results nextJobs}

\noindent\stepcounter{line}\makebox[2em][l]{\theline}\hbox{\phantom{xxxxxxxxxxx}\}}

\noindent\stepcounter{line}\makebox[2em][l]{\theline}\hbox{\phantom{xxxx}\}}

}\end{proof}

%
%
%

\subsection{Proof of Lemma \ref{prop:gadget}}
\label{subsect:programmatic-final}

We can finally combine all the results of Section \ref{implementation} to get
our Lemma:

\begin{replemma}{prop:gadget}
If the RSA signatures of all messages used when checking the proof were not
counterfeit, then the gadget pattern $G$, shown in Figure~\ref{fig:gadget}, can
only be self-assembled with 13 tile types if a tile set is used which is
isomorphic to $T$.

\begin{proof}
The result follows from the combination of Lemmas \ref{lem:server},
\ref{lem:dowork} and \ref{lem:placeTile}.
\end{proof}
\end{replemma}

\end{document}